\documentclass[12pt,a4paper]{article}
\usepackage{amsmath,amsfonts,amssymb,amsthm}
\usepackage{graphicx}
\usepackage{fullpage}

 \usepackage{subfig}
 \usepackage{wrapfig}
 \usepackage{enumerate}
\usepackage{color}

\makeatletter
\newcommand\footnoteref[1]{\protected@xdef\@thefnmark{\ref{#1}}\@footnotemark}
\makeatother

\usepackage{hyperref} 
\definecolor{modra3}{rgb}{.1,.0,.4}
\hypersetup{colorlinks=true, linkcolor=modra3, urlcolor=modra3, citecolor=modra3, pdfpagemode=UseNone, pdfstartview=} 
\if 10
\usepackage[inline]{showlabels} 
\definecolor{zelena}{rgb}{0,.35,0}

\showlabels{cite}
\showlabels{bibitem}
\fi

\def\emphh#1{\emph{#1}}  
\def\yDelta{Y-$\Delta$ }

\def\degr{\mathrm{deg}}

\def\nCross#1#2#3{\mathrm{cr}_{#3}(#1, #2)}

\def\pdeg#1{{pdeg}(#1)}

\def\pNeigh#1{{\mathrm{P}}(#1)}

\def\mapping{\zeta}

\def\st#1{{\mathrm{St}}(#1)}

\def\deltaY{$\Delta$-Y }
\newcommand{\rephrase}[3]{\noindent\textbf{#1 #2}.~\emph{#3}}

\definecolor{modra}{rgb}{0,0,.8}
\definecolor{piros}{rgb}{.8,0,0}
\def\marrow{{\marginpar[\hfill$\Rrightarrow$]{$\Lleftarrow$}}}
\if 11    
\def\jk#1{}
\def\rf#1{{\color{piros} {\sc RF: }{\marrow\sf #1} }}
\else
\def\jk#1{}
\def\rf#1{}
\fi
\newif\iflong
\longtrue

\newtheorem{theorem}{Theorem} 

\newtheorem{corollary}[theorem]{Corollary}

\newtheorem{claim}[theorem]{Claim}

\theoremstyle{remark}
\newtheorem{remark}[theorem]{Remark}

\theoremstyle{definition}
\newtheorem{definition}[theorem]{Definition}

\begin{document}

\title{Hanani--Tutte for approximating maps of graphs\footnote{Extended abstract appeared in Proceedings of The 34th International Symposium on Computational Geometry (SoCG 2018).
We are grateful to M. Skopenkov for his critical remarks to previous versions of this paper. We are still working on some of the remarks, and are expecting to hear from him whether the current version still has some more critical issues.}}

\author{Radoslav Fulek\thanks{Department of Mathematics, Stanford University; \texttt{radoslav.fulek@gmail.com}. The research leading to these results has received funding from the People Programme (Marie Curie Actions) of the European Union's Seventh Framework Programme (FP7/2007-2013) under REA grant agreement no [291734]. The author gratefully acknowledges support from Austrian Science Fund (FWF): M2281-N35.}
\and
Jan Kyn\v{c}l\thanks{Department of Applied Mathematics and Institute for Theoretical Computer Science, Charles University, Faculty of Mathematics and Physics, 
Malostransk\'e n\'am.~25, 118 00~ Praha 1, Czech Republic;
\texttt{kyncl@kam.mff.cuni.cz}. Supported by project 
16-01602Y of the Czech Science Foundation (GA\v{C}R) and by Charles University project UNCE/SCI/004.}}



\maketitle

\begin{abstract}
We resolve in the affirmative a conjecture of Repov\v{s} and A.~Skopenkov (1998), and essentially resolve a conjecture M. Skopenkov (2003) generalizing the classical Hanani--Tutte theorem to the setting of approximating maps of graphs on 2-dimensional surfaces by embeddings. Our proof of this result is constructive and almost immediately implies an efficient algorithm for testing whether a given piecewise linear map of a graph in a surface is approximable by an embedding.
More precisely, an instance of this problem consists of (i) a  graph $G$ whose vertices are partitioned into clusters and whose inter-cluster edges are partitioned into bundles, and (ii) a region $R$ of a 2-dimensional compact surface $M$  given as the union of a set of pairwise disjoint discs corresponding to the clusters and a set of pairwise disjoint ``pipes''
corresponding to the bundles, connecting certain pairs of these discs.
We are to decide whether $G$ can be embedded inside $M$ so that the vertices in every cluster are drawn in the corresponding disc, the edges in every bundle pass only through its corresponding pipe, and every edge crosses the boundary of each disc at most once. 

 \end{abstract}
 
 \newpage
 \tableofcontents
 \thispagestyle{empty}
 
 \newpage
\setcounter{page}{1} 
\section{Introduction}
\label{sec:intro}

It is a classical result of Hopcroft and Tarjan that graph planarity can be tested in linear time~\cite{HoTa74_planarity}, and a linear-time algorithm is also known for testing whether $G$ can be embedded into an arbitrary compact 2-dimensional manifold $M$~\cite{Moh99}, though
computing the orientable genus (as well as Euler genus and non-orientable genus) of a graph is NP-hard~\cite{T89_genusNP}

Seeing a graph $G$ as a $1$-dimensional topological space, an embeddability-testing algorithm decides whether there exists
an injective continuous map, also called an \emphh{embedding}, $\psi: G\rightarrow M$, where $M$ is a given triangulated compact 2-dimensional manifold without boundary.
We study a variant of this algorithmic problem in which we are given a piecewise linear continuous map $\varphi: G \rightarrow M$, which is typically not an embedding, and we are to decide
whether for every $\varepsilon>0$ there exists an embedding $\psi: G\rightarrow M$ such that $\|\psi-\varphi\|<\varepsilon$, where $\|.\|$ is the supremum norm.
Such a map $\psi$ is called an \emph{$\varepsilon$-approximation} of $\varphi$, and in this case we say that $\varphi$ is \emph{approximable by an embedding};
or as in~\cite{AFT18+_weak}, a \emph{weak embedding}.
If $\varphi$ is a constant map, the problem is clearly equivalent to the classical planarity testing.
Obviously, an instance of our problem is negative if there exists a pair of edges $e$ and $g$ in $G$ such that the curves $\varphi(e)$ and $\varphi(g)$ induced by $\varphi$ properly cross. Hence, in a typical instance of our problem the map $\varphi$ somewhat resembles an embedding except that we allow a pair of edges to overlap and an edge to be mapped onto a single point.

The problem of approximability of $\varphi$ by an embedding and its higher-dimensional analogues appeared in the literature under different names in several contexts in both mathematics and computer science. Its investigation in mathematics goes back to the 1960s, when Sieklucki proved a theorem~\cite[Theorem 2.1]{S69_realization}
implying the following. For a given $G$, every continuous piecewise linear map $\varphi: G \rightarrow \mathbb{R}\subset \mathbb{R}^2$ is approximable by an embedding if and only if every connected component of $G$ is a subcubic graph, i.e., a graph of maximum vertex degree at most 3, 
with at most one vertex of degree $3$\footnote{The theorem of Sieklucki is more general than the claim and is formulated as a result about ``realization of mappings''. In this setting we are given a topological space $Y$ and a map $\varphi: G\rightarrow X$ and we are interested in the existence of an embedding $\mathcal{E}$ of $X$ in $Y$ such that the composition of $\varphi$ and $\mathcal{E}$ 
is approximable by an embedding.}.
In a more recent paper by M.~Skopenkov~\cite{Sko03_approximability},
an algebraic characterization via van Kampen obstructions of maps $\varphi$ approximable by an embedding in the plane
is given in the case when $G$ is a cycle or when $G$ is subcubic and
the image of $\varphi$ is a simple closed curve.
This implies a polynomial-time algorithm for the decision problem in the corresponding cases
and can be seen as a variant of the following characterization of planar graphs
due to Hanani and Tutte.

The Hanani--Tutte theorem is a classical result~\cite{Ha34_uber,Tutte70_toward} stating that a graph $G$ is planar if it can be drawn in the plane so that every pair of edges not sharing a vertex cross an even number of times. 
According to Schaefer~\cite[Remark 3.6]{Sch13b_towards}, \emph{``The planarity criterion of Hanani--Tutte brings together computational, algebraic and combinatorial aspects of the planarity problem.''} 
Perhaps the most remarkable algorithmic aspect of this theorem is that it implies the existence of a polynomial-time algorithm for planarity testing~\cite[Section 1.4.2]{Sch13_hananitutte}. In particular, the Hanani--Tutte theorem 
reduces planarity testing to solving a system of linear equations over $\mathbb{Z}_2$.

In what follows, we measure the time complexity of the algorithm in terms of the number of real values specifying
$\varphi$, denoted by $|\varphi|$, where a pair of real values specify $\varphi(v)$, for every $v\in V(G)$, and for every edge $e\in E(G)$
the polygonal line $\varphi(e)$ is additionally specified by a sequence of coordinates of its internal vertices.

Roughly speaking, in the corresponding special cases the result of M.~Skopenkov says that $\varphi$ is approximable by an embedding if and only if $\varphi$ is approximable by a generic continuous map under which every pair of non-adjacent edges cross
an even number of times.
The running time of the algorithm based on the characterization is $O(|\varphi|^{2\omega})$, where 
$O(n^{\omega})$ is the complexity of multiplication of square $n\times n$ matrices~\cite[Section 2]{FKMP15}.
The best current algorithms for matrix multiplication give $\omega<2.3729$~\cite{LeGall14_powers,WVV12}\footnote{Since a linear system appearing in the solution is sparse, it is also possible to use Wiedemann's randomized algorithm~\cite{Wie86_sparse}, with expected running time $O(|\varphi|^4\log{|\varphi|}^2)$ in our case.}.
Independently of the aforementioned developments, results of recent papers~\cite{AAET16,ChEX15,CDPP09} on weakly simple embeddings imply that the problem of deciding the approximability of $\varphi$ by an embedding is tractable and can be carried out in $O(|\varphi|\log |\varphi|)$  time if $G$ is a cycle. 
Besides the mentioned results, prior to our work, restrictive partial results with much worse running time were achieved~\cite{ADDF17,F16_bounded,F17_pipes,FKMP15}.

We show that the problem of deciding whether $\varphi$ is approximable by an embedding admits an efficient algorithm for every graph $G$ and continuous piecewise linear map $\varphi:G\rightarrow M$.
To  this end we essentially confirm a conjecture of M.~Skopenkov~\cite[Conjecture 1.6]{Sko03_approximability}. 

In spite of the analytic definition, the algorithmic problem of deciding whether $\varphi$ is approximable by an embedding admits a polynomially equivalent reformulation that is of combinatorial flavor and that better captures the essence of the problem. 
Therefore we state our results in terms of the reformulation, whose planar  
case is a fairly general restricted version of the c-planarity problem~\cite{FCE95a_how,FCE95b_planarity} of Feng, Cohen and Eades introduced by Cortese et al.~\cite{CDPP09}.
The computational complexity of c-planarity testing is a well-known notoriously difficult open problem in the area of graph visualization~\cite{CB05_invited}.
To illustrate this state of affairs we mention that Angelini and Da Lozzo~\cite{AL16_pipes} have recently studied our restricted variant (as well as its generalizations) under the name of \emphh{c-planarity with embedded pipes}
and provided an FPT algorithm for it~\cite[Corollary 18]{AL16_pipes}. 

Roughly speaking, in the \emphh{clustered planarity}, shortly \emphh{c-planarity}, problem we are given a planar graph $G$ equipped with a hierarchical structure of subsets of its vertex 
set. The subsets are called \emphh{clusters}, and two clusters are either disjoint or one contains the other. The question is whether a planar embedding of $G$ with the following property exists:
the vertices in each cluster are drawn inside a disc corresponding to the cluster so that the boundaries of the discs do not intersect, the discs respect the hierarchy of the clusters, and every edge in the embedding crosses the boundary of each disc at most once. 

\jk{pridat ref~\cite{Sko95_thickening} algorithm for approximability by reduction to thickenability of 2-complexes into 3-manifolds}
\jk{ - (G,H,phi) is a weak embedding~\cite{ChEX15}, is a weak independently even embedding, weak independent Z2-embedding... independently even clustered drawing - to je mozna ono!!! - pojem pro clustered planaritu s trubkami!!! (c-planar by bylo pretizene)
..weak embedding by bylo G -> R2, to H je tam pak skoro zbytecny, ale asi lepsi to rict s tim}

\jk{mitosis - jiny pojem pro derivaci :) jak se tam tahaji ty hrany do jadra a k okraji... nebo meiosis? protoze si nebere vsechen geneticky material?}

\jk{poznamky:}

\jk{The derivative operation looks much like the "node expansion" in Cortese at al. }

\jk{Samozrejme bude potreba vysvetlit vsechny souvislosti, a poukazat na to, jak graph drawing community ignored these earlier topological papers.}

\jk{ze nakresleni pipe grafu je dano kombinatoricky, ne geometricky jako nejake PL vnoreni, kde jsou potreba souradnice bodu.}

\jk{jednodussi dukaz pro 2 clustery, ze existuje diky stronger ht, by se mohl v poznamce napsat - oprava rotaci exposed vrcholu + split, pouziti stronger HT, bipartitni lemma. Mozna stale potreba cykly z kaktusu, ktere budou mit kompletne opravene rotace, a budou disjunktni, ale uz ne kuchani.}

\jk{je nutne, aby po derivaci hrany pripojene na ruzne vrcholy jedne 0-komponenty byly propojene v novem pipe-clusteru? U te verze, co jsem popisoval, tohle splneno byt nemuselo.}

\jk{formulovat to cele jako stronger HT? - asi jo, staci jen overit, ze v kazdem kroku se zachovavaji rotace pokud byly sude}

\jk{Also, everything seems to be the same on any orientable surface. Perhaps even non-orientable, just there the condition on monotone cycles is a little bit different if the final cycle is nonorientable.}

\jk{a new section about algorithm - because we cannot use our general algorithm directly - we cannot allow edge-cluster switches, just edge-vertex switches where the vertex is in the same cluster as one vertex of the edge. And this is easily seen to be enough.}

\paragraph{Notation.}
Let us introduce the notation necessary for precisely stating the problem that we study.
 Let $G=(V,E)$ be a multigraph  without loops. If we treat $G$ as a topological space, then a \emphh{drawing} $\psi$ of $G$ is a piecewise linear map from $G$ into a triangulated 2-dimensional manifold $M$ where every vertex in $V(G)$ is mapped to a unique point and every
edge $e\in E$ joining $u$ and $v$ is mapped bijectively to a simple curve joining $\psi(u)$ and $\psi(v)$. 
Unless stated otherwise by a \emphh{curve} or an \emphh{arc} we mean the image of a continuous piecewise linear map from the closed interval $[0,1]\subset \mathbb{R}$ to $M$.
 A curve is \emphh{closed} if its map sends $0$ and $1$ to the same point. 
We understand $E$ as a multiset, and by a slight abuse of notation
we refer to an edge $e$ joining $u$ and $v$ as $uv$ even though there might be other edges joining the same pair of vertices. Multiple edges are mapped to distinct curves meeting at their endpoints.
Given a map $m$ we denote by $m|_X$, where $X$ is a subset of the domain of $m$, the function obtained from $m$ by restricting its domain to $X$. If $H$ is a graph equipped with an embedding, 
we denote by $H|_X$, where $X$ is a subgraph of $H$, the graph $X$ with the embedding inherited from $H$.

If it leads to no confusion,  we do not distinguish between
a vertex or an edge and its image in the drawing and we use the words ``vertex'' and ``edge'' in both
 contexts. Also when talking about a drawing we often mean its image.
 To describe modifications of drawings we opt for a less formal and more intuitive style and use terms such as ``dragging'', ``closely following'', ``deforming'', etc., similarly as in~\cite{PSS06_removing}, which we believe should not lead to confusion.
 Since these terms have in the given context very clear meaning and to formalize them is rather  straightforward, but also technical, we think that being more formal would not really help most of the readers to follow the exposition.

We assume that drawings satisfy the following standard general position conditions. No edge passes through a vertex,
 every pair of edges intersect in finitely many points, no three edges intersect at the same inner point,
and every intersection point between a pair of edges is realized either by a proper crossing or a common endpoint. Here, by a \emph{proper crossing}
we mean a transversal intersection that is a single point.

An \emphh{embedding} of a graph $G$ is a drawing of $G$ in $M$ without crossings. 
An embedding is \emphh{cellular} if every connected component of its complement in $M$ is homeomorphic to an open dics.
The \emphh{rotation} at a vertex $v$ in a drawing of $G$ is the clockwise cyclic order of the edges incident to $v$ in a small neighborhood of $v$ in the drawing w.r.t a chosen orientation at the vertex. The \emphh{rotation system} of a drawing of $G$ is the set of rotations of all the vertices in the drawing. 
A cellular embedding of $G$ is determined up to isotopy class by the rotation system if $M$ is orientable. If $M$ is non-orientable we need to additionally specify the signs of the edges as follows.
We assume that $M$ is constructed from a punctured 2-sphere by turning all the holes into cross-caps, i.e., by identifying the pairs of opposite points on every hole. 
A \emphh{sign} on an edge is positive if overall the edge passes an even number of times through the cross-caps, and negative otherwise.

Refer to Figure~\ref{fig:ex}.
We refer the reader to the monograph by Mohar and Thomassen~\cite{MT02_graphs} for a detailed introduction into surfaces and graph embeddings.\footnote{Working on surfaces with higher genus is not really essential in our problem, since everything important happens already in the plane. Therefore the reader that is not familiar with them is encouraged to skip  parts of the   arguments  dealing with them.}
Let $\varphi:G \rightarrow M$ be a piecewise linear map with finitely many linear pieces. Suppose that $\varphi$ is free of edge crossings, and in $\varphi$, edges do not pass through vertices. As we will see later, the image of $\varphi$ can be naturally identified with a graph $H$ embedded in $M$.
Throughout the paper we denote both vertices and edges of $H$ by Greek letters. Let the \emphh{thickening} $\mathcal{H}$ of $H$ be a  2-dimensional surface with boundary obtained as a quotient space of a union of pairwise disjoint topological discs as follows. We take a union of pairwise disjoint closed discs $\mathcal{D}(\nu)$, called \emph{clusters}, for all $\nu\in V(H)$ and closed rectangles $\mathcal{P}(\rho)$, called \emph{pipes}, for all $\rho\in E(H)$.
We connect every pair of discs $\mathcal{D}(\nu)$ and $\mathcal{D}(\mu)$, such that $\rho=\nu\mu\in E(H)$, by $\mathcal{P}(\rho)$ in correspondence with the rotations at vertices of the given embedding of $H$ as described next. We consider a subset of $\partial \mathcal{D}(\nu)$, for every $\nu \in V(H)$, consisting of $\degr(\nu)$  pairwise disjoint closed (non-trivial) arcs $\mathcal{A}(\nu,\mu)$,  one for every $\nu\mu\in E(H)$, appearing along $\partial \mathcal{D}(\nu)$ in correspondence with the rotation of $\nu$.
For every $\mathcal{D}(\nu)$, we fix an orientation of $\partial \mathcal{D}(\nu)$ and $\partial \mathcal{P}(\nu\mu)$.
If $M$ is orientable, for every $\mathcal{P}(\nu\mu)$, we identify by an orientation reversing homeomorphism its opposite sides with $\mathcal{A}(\nu,\mu)$ and $\mathcal{A}(\mu,\nu)$ w.r.t.  the chosen orientations  of $\partial \mathcal{D}(\nu)$,
$\partial \mathcal{D}(\mu)$, and $\partial \mathcal{P}(\nu\mu)$. 
If $M$ is non-orientable,
for every $\mathcal{P}(\nu\mu)$ with the positive sign we proceed as in the case when $M$ is orientable and for every every $\mathcal{P}(\nu\mu)$ with the negative sign, we identify  by an orientation preserving homeomorphism   its opposite sides with $\mathcal{A}(\nu,\mu)$ and $\mathcal{A}(\mu,\nu)$ w.r.t.  the chosen orientations of  $\partial \mathcal{D}(\nu)$ and $\partial \mathcal{P}(\nu\mu)$, and the reversed orientation of $\partial \mathcal{D}(\mu)$.
We call the intersection of $\partial\mathcal{D}(\nu)\cap \partial \mathcal{P}(\nu\mu)$ 
a \emphh{valve} of $\nu\mu$.

The rotation system of $H$ with the signs on its edges are \emphh{realized} on a 2-dimensional surface $M$ if by attaching discs to $\mathcal{H}$ by identifying homeomorphically each boundary component of $\mathcal{H}$ with the boundary of a single disc we get a union of pairwise disjoint surfaces, whose connected sum is a surface homeomorphic to $M$.

\paragraph{Instance.} 
\jk{to jeste zobecnime na pripad, kdy H je multigraf bez smycek, kde se bude predpokladat podminka na dvojice multihran v ruznych trubkach, ze se nesmi protnout lisekrat "vlevo ani vpravo". Redukce je podrozdeleni trubky clusterem a kazde hrany vrcholem, a preswitchovani na indep. even drawing (jde to, lokalne je to podprostor dimenze 3 pro 2 cesty delky 2). Toto je pripad, kdy $H$ je cluster adjacency graph.}

\jk{odsud; jeste sem vlozit komentar z minule verze, prepsat bez epsilonu - k definici embedded pipes snad epsilony nejsou potreba... proste vrcholy budou reprezentovany jako disky, hrany taky jako disky, vsechno interior-disjoint, a hranove disky se protinaji pouze s vrcholovymi podel spolecne hranice - valve. Valves are disjoint. Proc to delat jednoduse...  misto tehle epsilon-definice napsat, ze problem jde take formulovat jako epsilon-approximability of simplicial maps of graphs to the plane, where a simplicial map can map multiple vertices to the same point, ane every edge is mapped to a straight-line segment, which may be degenerate.}

An instance of the problem that we study is defined as follows.
The instance is a triple $(G,H,\varphi)$ of an (abstract) graph $G$, a graph $H$
embedded in a closed 2-dimensional manifold $M$, and a
map $\varphi: V(G)\rightarrow V(H)$ such that every pair of vertices joined by an edge in $G$ are mapped either to a pair of vertices joined by an edge in $H$ or to the same vertex of $H$.
We naturally extend the definition of $\varphi$ to each subset $U$ of $V(G)$ by putting $\varphi(U)=\{\varphi(u)| \ u\in U\}$, and to each subgraph $G_0$ of $G$ by putting $\varphi(G_0)=(\varphi(V(G_0)),\{\varphi(e)| \ e\in E(G_0), |\varphi(e)|= 2\})$.
The map $\varphi$ induces a partition of the vertex set of $G$ into \emph{clusters} $V_{\nu}$, where $V_{\nu}=\varphi^{-1}[\nu]$.

\paragraph{Question:} 
Decide whether there exists an embedding $\psi$ of $G$ in the interior of a thickening $\mathcal{H}$ of $H$ so that the following hold.

\begin{enumerate}[(A)]
\item \label{it:1st} Every vertex $v\in V_{\nu}$ is drawn in the interior of $\mathcal{D}(\nu)$, i.e., $\psi(v)\in \mathrm{int}(\mathcal{D}(\nu))$.
\item \label{it:2nd}  For every $\nu\in V(H)$, every edge $e\in E(G)$ intersecting $\partial \mathcal{D}(\nu)$  does so in a single proper crossing, i.e., $|\psi(e)\cap \partial\mathcal{D}(\nu)|\le 1$.
\end{enumerate}

Note that conditions~(\ref{it:1st}) and (\ref{it:2nd}) imply that every edge of $G$ is allowed to pass through at most one pipe as long as $G$ is drawn in $\mathcal{H}$. The instance is \emphh{positive} if an embedding $\psi$
of $G$ satisfying~(A) and (B)  exists and \emphh{negative} otherwise.
If $(G,H,\varphi)$ is a positive instance we say that $(G,H,\varphi)$ is \emphh{approximable by the embedding}  $\psi$, shortly \emphh{approximable}. 
We call $\psi$ the \emphh{approximation} of $(G,H,\varphi)$.
When the instance $(G,H,\varphi)$ is clear from the context, we call $\psi$ the \emphh{approximation} of $\varphi$.

\begin{figure}
\centering
\includegraphics[scale=0.8]{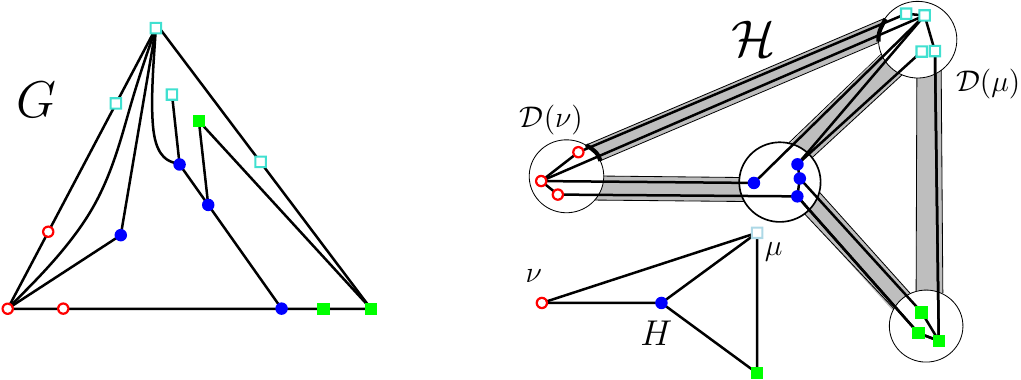}
\caption{An instance $(G,H,\varphi)$ and its approximation by an embedding contained in the thickening $\mathcal{H}$ of $H$. The valves of the pipe of $\rho=\nu\mu$ at $\mathcal{D}(\nu)$ and $\mathcal{D}(\mu)$ are highlighted by bold arcs.}
\label{fig:ex}
\end{figure}

\begin{figure}
\centering
\includegraphics[scale=0.8]{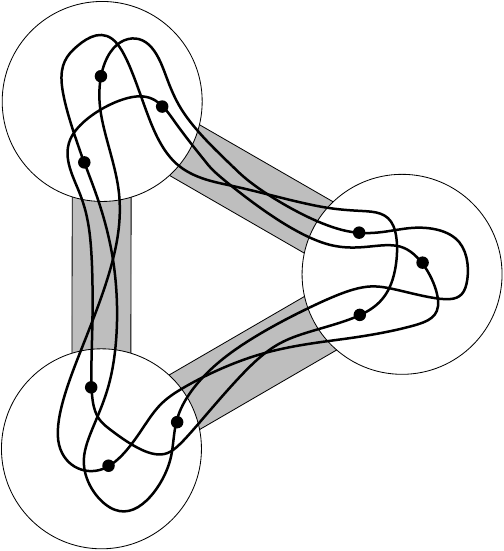}
\caption{``Standard winding example'' witnessing the fact that the $\mathbb{Z}_2$-approximability of an instance $(G,H,\varphi)$ does not imply its approximability by an embedding. In the figure, $G$ and $H$ are cycles of length 9 and 3, respectively. The instance is negative, since $\varphi$ forces $G$ to ``wind three times'' around $H$.}
\label{fig:swe}
\end{figure}

The instance $(G,H,\varphi)$, or shortly $\varphi$, is \emphh{locally injective} if for every vertex $v\in V(G)$, the restriction of $\varphi$ to the union of $v$ and the set of its neighbors is injective, or equivalently, no two vertices that are adjacent or have a common neighbor in $G$ are mapped by $\varphi$ to the same vertex in $H$.
An edge of $G$ is a \emphh{pipe edge} if it is mapped by $\varphi$ to an edge of $H$. 
When talking about pipe edges, we have a particular instance in mind, which is clear from the context.
Let $C$ denote a subgraph of $G[V_\nu]$, that is, the subgraph of $G$ induced by $\varphi^{-1}[\nu]=V_{\nu}$.
The \emphh{pipe degree}, $\pdeg{C}$, of  $C$  is the number of 
edges $\rho$ of $H$  for which there exists a pipe edge $e$ with one vertex in $C$ such that $\varphi(e)=\rho$. Let the \emphh{pipe neighborhood} $\pNeigh{C}$ of $C$ denote the set of end vertices of the edges in $H$ different from $\nu$ witnessing $\pdeg{C}$.

In Section~\ref{sec:equivalence} we show that the problem of deciding whether a piecewise linear continuous map $\varphi_0$ of from $G$ to  $M$ is approximable by an embedding is polynomially equivalent to the problem of deciding whether an instance constructed from $\varphi_0$ is positive.

\subsection{The result}
\label{sec:result}

An edge in a drawing is \emphh{even} if it crosses every other edge an even number of times.
A vertex in a drawing is \emphh{even} if every pair of its incident edges cross an even number of times.
An edge in a drawing is \emphh{independently even} if it crosses every other non-adjacent edge an even number of times.
A drawing of a graph is \emphh{(independently) even} if all edges are (independently) even. Note that every embedding is an even drawing.

\jk{tady by mozna bylo lepsi zase formalne definovat problem Z2-aproximace - protoze "instance" nemuze jen tak plavat ve vakuu, je to instace nejakeho problemu.}
We formulate our main theorem in terms of a relaxation of the notion of an approximable instance $(G,H,\varphi)$.
An instance $(G,H,\varphi)$ is \emphh{$\mathbb{Z}_2$-approximable} if there exists an independently even drawing of $G$ in the interior of $\mathcal{H}$ satisfying~(\ref{it:1st}) and~(\ref{it:2nd}). We call such a drawing a $\mathbb{Z}_2$-approximation of $(G,H,\varphi)$.

The proof of the Hanani--Tutte theorem from~\cite{PSS06_removing} proves that given an independently even drawing of a graph in the plane, there exists an embedding of the 
graph in which the rotations at even vertices are preserved, that is, they are the same as in the original independently even drawing. We refer to this statement
as to the \emphh{unified Hanani--Tutte theorem}~\cite[Theorem 1.1]{FKP17_unified}. Our result can be thought of as a generalization of this theorem, which also motivates the following definition.
A drawing $\psi$ of $G$ is \emphh{compatible} with a drawing $\psi_0$ of $G$ if every even vertex in $\psi_0$ is also even in $\psi$ and has the same rotation up to the choice of an orientation in both drawings $\psi_0$ and $\psi$.\footnote{Since in general we work also with non-orientable surfaces the rotation is determined only up to the choice of an orientation at each vertex if the surface $M$, on which    the rotation system of $H$ with signs is realized, is non-orientable.} If $M$ is orientable our result is a proper generalization of the unified Hanani--Tutte theorem, since ``up to the choice of an orientation'' can be dropped in the previous definition in this case.

It is known that $\mathbb{Z}_2$-approximability of $(G,H,\varphi)$ does not have to imply its approximability by an embedding~\cite[Figure 1(a)]{ReSk98_deleted}.
Our main result characterizes the instances
$(G,H,\varphi)$ for which such implication holds.
The characterization is formulated in terms of 
the derivative of $(G,H,\varphi)$, whose formal definition is postponed to Section~\ref{sec:normal-form}, since its definition relies on additional concepts that we need to introduce, which would unnecessarily further delay stating of our main result. We give an informal description of the derivative later in this subsection.

\begin{theorem}
\label{thm:main}
If an instance $(G,H,\varphi)$ is $\mathbb{Z}_2$-approximable by an independently even drawing $\psi_0$ then either
$(G,H,\varphi)$ is approximable by an embedding $\psi$  compatible with $\psi_0$, or it is not approximable by an embedding and in the $i$th derivative $(G^{(i)},H^{(i)},\varphi^{(i)})$, for some $i \in \{1,2,\dots, 2|E(G)|\}$, there exists a connected component $C\subseteq G^{(i)}$ such that $C$ is a cycle, $\varphi^{(i)}$ is locally injective and $(C,H^{(i)}|_{\varphi^{(i)}(C)},\varphi^{(i)}|_C)$ is not approximable by an embedding.
\end{theorem}

\jk{ve zneni vety nastava nejaka automaticka "restrikce" te "instance" na cykl-komponentu C, nevim, jestli to je nekde formalne definovane (hlavne to ponechani stejneho fi) ... nebo crvrta derivace? :) ... Tady bych si jeste predstavoval rict konkretneji, jak ta obstrukce vypada - ze to je nekolikrat namotany cyklus - to by se zase definovalo pred tou vetou, slovy ze je to "essentially the only obstruction", a pak se (v te vete) objasni, ze essentially znamena ze se ta obstrukce muze objevit v derivaci.}

\begin{remark}
Refer to Figure~\ref{fig:swe}.
The obstruction $(C,H^{(i)}|_{\varphi^{(i)}(C)},\varphi^{(i)}|_C)$ from
the statement of the theorem has the form of the ``standard winding example''~\cite[Figure 1(a)]{ReSk98_deleted},
in which the cycle $C$ is forced by $\varphi^{(i)}$ to wind around a point inside a face of $H$ more than once (and an odd number of times, since it has a $\mathbb{Z}_2$-approximation). 
\end{remark}

Our main result implies the following.

\begin{corollary}
\label{cor:forest}
If $G$ is a forest, the $\mathbb{Z}_2$-approximability implies approximability by an embedding.\jk{our theorem implies this, right? this should be stated explicitly.}
\end{corollary}

Theorem~\ref{thm:main} essentially confirms a conjecture of M.~Skopenkov~\cite[Conjecture 1.6]{Sko03_approximability}\footnote{The conjecture predicts $i$ in Theorem~\ref{thm:main} to be at most $|V|$ rather than $2|E|$, hence, the word ``essentally''.}, since our definition of the derivative agrees with his definition in the case when $G$ is a cycle and since for every
cycle $C$ in $G^{(i)}$  there exists
 a cycle $D$ in $G$ 
such that $(D^{(i)},H^{(i)}|_{\varphi(D^{(i)})},\varphi^{(i)}|_{D^{(i)}})=(C,H^{(i)}|_{\varphi^{(i)}(C)},\varphi^{(i)}|_C)$, by Claim~\ref{claim:hereditary}.
The previous claim  also immediately implies Corollary~\ref{cor:forest} which
 confirms a conjecture of Repov\v{s} and A. B. Skopenkov~\cite[Conjecture 1.8]{ReSk98_deleted}.
 The corollary follows also implicitly from the proof of Theorem~\ref{thm:main}, since it will be straightforward to see that $G^{(i)}$ is a forest, if $G$ is a forest.
The main consequence of Theorem~\ref{thm:main} is the following.

\begin{theorem}
\label{thm:main2}
We can test in $O(|\varphi|^{2\omega})$, where $O(n^{\omega})$ is the complexity of multiplication of square $n\times n$ matrices, whether $(G,H,\varphi)$ is approximable by an embedding.
\end{theorem}

Theorem~\ref{thm:main2} implies tractability of c-planarity with embedded pipes~\cite{CDPP09} and therefore solves a related open problem of Chang, Erickson and Xu~\cite
[Section 8.2]{ChEX15} and Akitaya et al.~\cite{AAET16}. The theorem also implies that strip planarity introduced by Angelini et al.~\cite{ADDF17} is tractable, and hence, solves the main problem asked therein.
The theorem generalizes results of~\cite{FKMP15} and~\cite{FPSS12_monotone}, and implies that c-planarity~\cite{FCE95a_how,FCE95b_planarity} for flat clustered graphs is tractable for instances with three clusters (Theorem~\ref{thm:c-planarity} in Section~\ref{sec:c-planarity}); this has been open, to the best of our knowledge. 
We remark that only solutions to the problem for two clusters were given so far~\cite{B98_2cl,GUT02_2cl,HN16_twopage}.
Nevertheless, after the completion of this work our running time was  improved to $O(|\varphi|^2\log |\varphi|)$~\cite{AFT18+_weak}. The improvement on the running time was achieved by using a similar strategy as in the present work, while eliminating the need to solve the linear system and employing a very careful running time analysis.
On the other hand, Theorem~\ref{thm:main} and Theorem~\ref{thm:main2} easily generalize to the setting when clusters are homeomorphic to cylinders   and $\mathcal{H}$ is homeomorphic to a torus or a cylinder
as discussed in Section~\ref{sec:toroidal}.
It is an interesting open problem to find out if the technique
of~\cite{AFT18+_weak} generalizes to this setting as well.

The main problem with extending these results to more general setting is that  
the tools such as PQ/PC-trees~\cite{BR16_simPQtrees,BL76_PQtrees,HMcC03_PCtrees} that are often used for devising algorithms for similar  problems did not seem to work here in spite of the fact that various special cases were solved by applying such methods. Though, in~\cite{AFT18+_weak} related SPQR-trees were employed  to improve upon our running time. 
A recent work~\cite{ADDF17,F16_cyclic} suggested an approach via computation of a flow/perfect matching in a graph. However, this approach is tailored for the setting in which also the isotopy class of an embedding of $G$ is fixed. An attempt to make it work also in the general setting was made in~\cite{F16_bounded} by the first author, but fell short of providing an algorithm even when $H$ is a path, except if $G$ is a forest.
The approach via a variant of the Hanani--Tutte theorem~\cite{FKMP15,Sch13b_towards}, which we take, also did not look very promising due to the counterexample found in~\cite[Figure 1(a)]{ReSk98_deleted}, which was discovered independently in~\cite[Section 6]{FKMP15}. 

The discussion in the rest of this section is rather informal and aims only at conveying the intuition underlying our technique.

\paragraph{The derivative of maps of graphs.}

\begin{figure}
\centering
\includegraphics[scale=0.8]{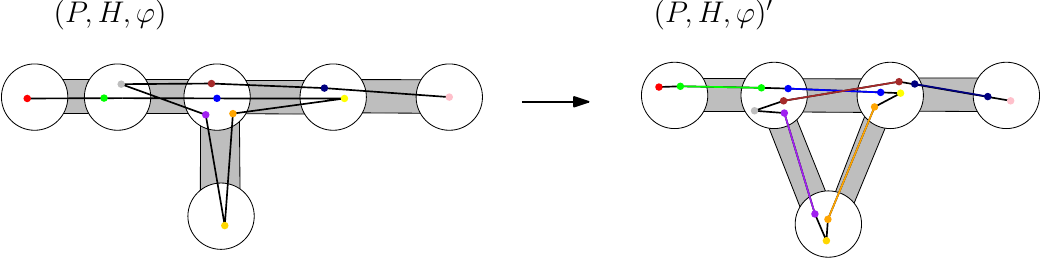}
 \caption{The construction of the derivative $(P,H,\varphi)^{'}$ of $(P,H,\varphi)$. 
}
\label{fig:exSimple0}
\end{figure}

\begin{figure}
\centering
\includegraphics[scale=0.8]{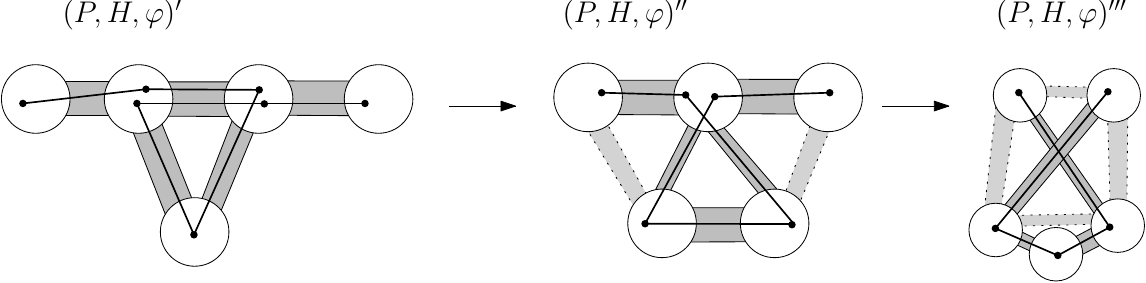}
\caption{Subsequent derivatives of $(P,H,\varphi)$. We contract every connected component of $(\varphi^{(i)})^{-1}[\nu]$, for $\nu\in V(H^{(i)})$, $i=1,2,3$, to a vertex without changing the approximability of the instance. The third derivative certifies that the original instance was not approximable by an embedding. The light shaded thickened edges of $H$ do not belong to the image of $\varphi^{(i)}$ in the corresponding instance.}
\label{fig:exSimple}
\end{figure}

\jk{"light shade" in Figure~\ref{fig:exSimple} neni prilis vyrazny oproti tomu darker.}

The main contribution of our work is an extension to arbitrary continuous piecewise linear maps of graphs in $\mathbb{R}^2$ of the approach of M.~Skopenkov~\cite{Sko03_approximability} for detecting whether a continuous map of a graph into a circle in $\mathbb{R}^2$, or
 an arbitrary continuous piecewise linear map of a cycle into $\mathbb{R}^2$ is approximable by an embedding. The technique therein is based on the notion of the \emphh{derivative} of a map of a graph introduced by Minc in~\cite{M94_derivative, Minc97_arcs}.
 The novelty of our technique lies in extending the techniques from this line of research and combining them with the independently developed techniques for variants of the Hanani--Tutte theorem~\cite{F14+_towards,FKMP15,PSS06_removing}. \jk{tahle veta mi tu pripada nejak nepatricne - najednou hovori obecne o vsech predchozich pristupech, kdyz se tu ma hovorit konkretne o derivaci (i kdyz je to historicka survey zatim). Mozna by se vic hodila pred tenhle odstavec o derivaci, na konec "obecnejsi" casti uvodu.}

To illustrate the notion of the derivative, we first discuss the case treated by Minc, when $G$ in an instance $(G,H,\varphi)$ is a path $P$.
Intuitively, the derivative of $(P,H,\varphi)$ smooths out the instance while simultaneously zooming into the structure of the map $\varphi$.
For example, if $H$ is a path of length $k$ then $k$ successive applications of the derivative result in an instance with a single cluster,
and if $H$ is a cycle, an application of the derivative might not change $H$, but shortens ``spurs''.

Refer to Figures~\ref{fig:exSimple0} and~\ref{fig:exSimple}.
It is easy to see that if $P$ is a path we can assume without loss of generality that $\varphi^{-1}[\nu]$,
for $\nu\in V(H)$, is always an independent set; in other words, $\varphi(e)$ is an edge of $H$ for every $e\in E(G)$.
The derivative in this case is defined as a new instance $(P',H',\varphi')=(P,H,\varphi)'$,
where $H'$ is a subgraph of the line graph\footnote{The line graph $L(H)$ of the graph $H$ is the graph, whose vertices are the edges of $H$, and a pair of vertices in $L(H)$ is joined by an edge
if the corresponding pair of edges in $H$ share a vertex.}
 of $H$ that is not necessarily planar.
However, $H'$ is planar
if $(P,H,\varphi)$ is approximable by an embedding.
Thus, $H'$ is still equipped with a drawing, but this drawing does not have to be an embedding.
The connected components of $\varphi^{-1}[\rho]$, for $\rho\in E(H)$, are mapped by $\varphi'$ to the vertex of $H'$ corresponding to $\rho$.
Such connected components are then joined by edges both of whose end vertices correspond to the same original vertex of $P$, and hence, the resulting graph $P'$ is again a path. 

The drawing of $H'$ is defined as a restriction of a drawing of the line graph $L(H)$ of $H$ naturally inherited from the embedding of $H$ as follows. We draw a vertex corresponding 
to an edge of $H$ as a point in its interior. The edges of $L(H)$ corresponding to a vertex $\nu$\jk{to uz chce znacnou davku predstavivosti... na prvni pohled to vypada jako chyba, protoze naopak vrcholy linegrafu koresponduji hranam grafu.} 
of $H$ form a clique $K_\nu$. The cliques $K_{\nu}$ are drawn so that combinatorially the drawing is equivalent to
a straight-line drawing in which the vertices are drawn in convex position ordered according to the rotation of $\nu$. Every pair of such cliques is drawn in the outer face of each other.
For us, the most interesting property of $(P,H,\varphi)'$ is that $(P,H,\varphi)'$ is approximable
by an embedding if and only if $(P,H,\varphi)$ is, and that the length of $P'$ after contracting all its subpaths mapped by $\varphi'$ to a single vertex of $H'$ is smaller than the length of $P$.

A priori, it is hard to see if the algorithmic problem of deciding whether $(G,H,\varphi)$ is approximable by an embedding is tractable even in the case when $G$ is a path. Here, the main problem are extremes of $\varphi$, which can be though of as ``tips of spurs'',\jk{Cortese et al. tomu rikaji "cusps"} such as yellow and gray vertices in Figure~\ref{fig:exSimple0}, since combinatorially there are two ways of approximating the drawing of the two edges incident to such an extreme by an embedding.
However, the algorithmic question of deciding whether $(P,H,\varphi)$ is approximable by an embedding can be reduced to constructing the derivative and checking whether the obtained drawing of $H'$ is an embedding. Since the latter can be easily carried out in a polynomial time this makes the problem tractable.

M. Skopenkov~\cite{Sko03_approximability} extended the ideas of Minc to the case of subcubic graphs mapped into a circle; that is, the case when $H$ is a cycle. In this case a $\mathbb{Z}_2$-approximation of the derivative is constructed given that the original instance was $\mathbb{Z}_2$-approximable. In Skopenkov's version of the derivative, the graph $H'$ is still a subgraph of the line graph of the original graph, and the $\mathbb{Z}_2$-approximability of the instance implies that $H'$ is planar; this last fact is also true in our definition of the derivative.

Independently discovered operations of expansion and contraction of a base from~\cite{CDPP09}, that were used also in~\cite{AAET16,ChEX15}, can be thought of as a local version of the derivative of Minc and M.~Skopenkov.\jk{tady tem bych asi chtel dat vice prostoru - trochu muj bias, protoze jako prvni jsem cetl Cortese et al.} 
We remark that the developments in~\cite{AAET16,ChEX15,CDPP09} were obtained independently from the 
 line of research pioneered by Minc, and the connection between these two research directions was not realized in the past.
 In fact, similarly as in~\cite{AFT18+_weak} we could use the operations of cluster and pipe expansions which seem to be more efficient from the algorithmic perspective, but in the context of $\mathbb{Z}_2$-approximation the derivative seems to be more natural.

\paragraph{Our extension of the derivative.}
In our extension 
of the Minc's and M.~Skopenkov's derivative\footnote{In comparison with the general definition of the derivative in~\cite{M94_derivative}, our definition is more-or-less the same when $G$ has maximum degree 2, and differs substantially for graphs with maximum degree more than 2.}, $H'$ is no longer a subgraph of the line graph of $H$, but $H'$ is rather a graph obtained by suppressing certain degree-$2$ vertices in a subgraph of a blow-up of the edge-vertex incidence graph of $H$.

One of the major obstacles in extending M.~Skopenkov's approach was the existence of vertices whose incident edges cannot be deformed locally near these vertices so that they all cross each other an even number of times.
Such a local deformation is always possible near vertices of degree at most $3$, thus in subcubic graphs this obstacle does not occur.
Even vertices help us, because the rotation at these vertices in a desired approximation is already decided. Indeed, Theorem~\ref{thm:main} claims the existence of an approximation that is compatible with the given $\mathbb{Z}_2$-approximation.

To overcome this problem we 
alter the given instance, thereby producing a $\mathbb{Z}_2$-approximable instance that is approximable by an embedding if and only if the former instance was, and that has a number of advantages. We say that such instance is in the \emph{subdivided normal form}. In particular, the vertex set of $G$ in such instance contains an independent set $V_s$ of central vertices, whose removal splits $G$ into a set $\mathcal{C}$ of connected components. These components can be understood as analogues of $\varphi$-components defined by Skopenkov~\cite{Sko03_approximability}. Each $C\in \mathcal{C}$ is mapped by $\varphi$ to an edge $\rho=\varphi(C)$ of $H$,
and the problematic parts in its $\mathbb{Z}_2$-approximation are relocated into connected components of $\mathcal{C}$\jk{neni jasne, jestli se ty komponenty nejak dynamicky meni v case; pokud jo, chtelo by to reflektovat v tom znaceni} so that they can be dealt with later.
By suppressing certain vertices of the instance in the subdivided normal form we obtain an instance in the \emph{normal form}.
 The other reason for introducing  the normal form is to impose on the instance conditions analogous to the properties of a \emphh{contractible base} of Cortese et al.~\cite{CDPP09}, or a \emphh{safe pipe}~\cite{AFT18+_weak}, which make the derivative reversible.

In the proof of Theorem~\ref{thm:main}, the obtained $\mathbb{Z}_2$-approximation of the instance in the normal form is repeatedly reduced by using our extended definition of the derivative and subsequently brought into the normal form again, where each time we also produce a 
$\mathbb{Z}_2$-approximation of the reduced instance.\jk{ten odstavec je prilis dlouhy... tady nekde by se asi hodilo ho rozpojit...} In fact, our definition of the derivative requires an instance to be in the normal form. In $(G',H',\varphi')$, the graph $G$ is not changed and we simply have $G'=G$. Thus, the graph $G$ is changed only when it is being brought into the normal form. Every connected component $C$ of $G[V\setminus V_s]$ in the subdivided normal form
is unchanged in the normal form, and is mapped by $\varphi'$ to the vertex $\varphi(C)^*\in V(H')$ corresponding to an edge $\varphi(C) \in E(H)$.
Every vertex $v_s\in V_s$ of degree at least $3$ is mapped to the vertex $v_s^*\in V(H')$ corresponding to a vertex $\nu\in V(H)$ such that
$\varphi(v_s)=\nu$.
The latter type of vertices of $H'$ are not in one-to-one correspondence with the vertices of $H$,
but rather in one-to-many correspondence, where each vertex in $V(H)$ corresponds to a set of vertices 
in $V(H')$.


In order to construct a $\mathbb{Z}_2$-approximation of the instance in the normal form we use a redrawing technique
of Pelsmajer et al.~\cite{PSS06_removing} inspired by modular decomposition of Hsu and McConnel~\cite{HMcC03_PCtrees}.
In particular, we will use redrawing techniques to render vertices in certain trees of $G$ even. Such a subtree can be then safely contracted into a vertex. Indeed, in Theorem~\ref{thm:main} we assume that the obtained approximation is compatible with the $\mathbb{Z}_2$-approximation, and thus, contracted trees consisting of even vertices can be recovered in the approximation of the reduced graph.
The other technical difficulty in this approach is the construction of the $\mathbb{Z}_2$-approximation of the derivative of an instance in the normal form. Here, we proceed in two steps. In the first step, we construct a $\mathbb{Z}_2$-approximation of a slight modification of the derivative by pretty much following the redrawing method of M.~Skopenkov. The Skopenkov's method can be seen as a surgery in which we first cut out and then reconnect pieces of the $\mathbb{Z}_2$-approximation
of the original instance induced by the edges of $H$.
The most delicate part of the argument is to define the drawing of the edges reconnecting the severed pieces, and to prove that in the obtained drawing we do not obtain a pair of non-adjacent edges crossing an odd number of times.
In the second step, we further alter the $\mathbb{Z}_2$-approximation thereby obtaining a $\mathbb{Z}_2$-approximation of the desired instance.

\paragraph{Organization and outline of the proof of Theorem~\ref{thm:main} and~Theorem~\ref{thm:main2}.}
We prove Theorem~\ref{thm:main} in Section~\ref{sec:reduction} and Theorem~\ref{thm:main2} in
Section~\ref{sec:alg}, where the proof of 
Theorem~\ref{thm:main2} is merely an ``algorithmic version'' of the proof of
Theorem~\ref{thm:main}. 
The basic tools needed in the proofs are presented in Section~\ref{sec:preliminaries}. 
The ultimate goal in the proof of Theorem~\ref{thm:main} is to reduce the instance $(G,H,\varphi)$ together with its $\mathbb{Z}_2$-approximation so that $\varphi$ is locally injective and $G$ does not contain paths as connected components. Such instances are easy to handle by using results from our work with I.~Malinovi\'c and D.~P\'alv\"olgyi~\cite[Section 6]{FKMP15}.

To this end we keep iteratively applying the derivative (as defined in Section~\ref{sec:normal-form} and further discussed in Section~\ref{sec:derivative_z2approx})  until we arrive at an instance in which $\varphi$ is locally injective. In order to show that we won't be applying the derivative indefinitely we define a potential function $p(G,H,\varphi)=(|E_p(G)| - |E(H)|)\ge 0$, where $E_p(G)$ is the set of all pipes edges in $G$, whose value decreases after an application of the derivative if $\varphi$ is not locally injective.

In the proof we start with pre-processing the instance thereby bringing it into the normal form defined in Section~\ref{sec:normal-form} and further discussed in Section~\ref{sec:normalform2}.
If $(G,H,\varphi)$ is in the normal form, and $\varphi$ is not locally injective, we apply the derivative, to simplify the instance.
In Section~\ref{sec:derivative_z2approx}, we prove that the obtained instance $(G',H',\varphi')$ is also $\mathbb{Z}_2$-approximable,
and approximable by an embedding if $(G,H,\varphi)$ is approximable
by an embedding.
Moreover, in Section~\ref{sec:integration} we show that $(G',H',\varphi')$ is approximable by an embedding only if $(G,H,\varphi)$ is approximable by an embedding, and hence, we can pass to the ``simpler'' instance $(G',H',\varphi')$.
This can be seen as a consequence of the planar case of Belyi's theorem~\cite{B83_self}.
The corresponding formal statements in Section~\ref{sec:derivative_z2approx} and~\ref{sec:integration} are actually more complicated due to Theorem~\ref{thm:main} claiming the existence of an approximation, which is an embedding, compatible with the given $\mathbb{Z}_2$-approximation.
Finally, we prove that after finitely many steps of ``normalizing'' and ``differentiating'' we eventually arrive at a locally injective instance, which is easily to deal with as discussed above.

\jk{tady k tomu a tomu dalsimu mam nevyresene komentare v historicke verzi "c. 0"...}

 
\section{Preliminaries}
\label{sec:preliminaries}

\subsection{Graph operations}
\label{sec:preliminariesGraph}

Throughout the paper we use the standard graph-theoretical notions~\cite{D16_graph_theory} such as path, cycle, walk, tree, forest, vertex degree, induced subgraph and others.
By $G\setminus v$ and $G\setminus V_0$, where $v \in V$ and $V_0\subseteq V$, we denote the graph obtained from $G$ by removing $v$ or all the vertices in $V_0$, respectively, together with all the incident edges. Similarly by $G\setminus e$, $G\setminus E_0$ and $G\setminus G_0$,
where $e\in E,$ $E_0\subseteq E$, and $G_0\subseteq G$ is a subgraph of $G$, we denote the graph obtained from $G$ after removing $e$, edges in $E_0$, or all the edges contained in $G_0$, respectively.

Let $v$ be a vertex of degree at least $2$ in a graph $G_0$ with a rotation $(vv_0,\ldots, vv_{\degr(v)-1})$ defined by a drawing of $G_0$.

\begin{definition}
\label{def:Ydelta}
\emphh{The generalized \yDelta operation} applied to $v$ results in the (abstract) graph $G_1$ obtained from $G_0$ by removing the vertex $v$ (with all its incident edges) and introducing the cycle $v_0\ldots v_{\degr(v)-1}v_0$.
\end{definition}

The inverse operation of the generalized \yDelta operation is called a \deltaY operation. Let $C=v_0\ldots v_{k-1}v_0$ be a cycle in $G_0$.

\begin{definition}
\label{def:deltaY}
\emphh{The generalized \deltaY operation} applied to $C$ results in the graph $G_1$ obtained from $G_0$ by removing the edges of $C$ and introducing a new vertex $v$ of degree $k$ and the edges $vv_0,\ldots, vv_{k-1}$. 
\end{definition}

\begin{remark}
On the one hand, performing the generalized \yDelta operation can lead to a creation of multiple edges, since we add an edge $v_iv_{i+1 \mod \deg(v)}$ even if such an edge is already present in the graph.
On the other hand, when performing the generalized \deltaY operation, if an edge of $C$ belongs to a collection of multiple edge, we remove only a single edge from the collection, the one belonging to $C$.
Usually, an application of one of the two operations will be followed by a construction of a drawing of $G_1$. Regarding 
the generalized {\deltaY} operation, in the drawing/embedding of $G_1$ the rotation of $v$ is naturally inherited from the cycle and equals $(vv_0,\ldots, vv_{k-1})$.
\end{remark}

A \emphh{subdivided edge} in a graph $G$ is a path of length at least $2$ whose middle vertices has degree $2$ in $G$.
\emphh{Suppressing} a vertex $v$ of degree $2$ in a graph $G$ is an operation that removes $v$ and its incident edges  from $G$, and joins its two neighbors by an edge in the resulting graph.

\subsection{Working with independently even drawings}
\label{sec:tools}

We present notions and facts that we use when working with independently even drawings.

\paragraph{Crossing number.} By $\nCross{e_1}{e_2}{\psi_0}$ we denote the number of crossings between edges $e_1$ and $e_2$ in a drawing $\psi_0$.
Formally, $\nCross{e_1}{e_2}{\psi_0}=|\psi_0(e_1)\cap \psi_0(e_2)|-|e_1\cap e_2|$ due to our general position assumption on $\psi_0$.

\paragraph{Inside and outside of a contractible closed curve.}
Let $C$ be a closed piece-wise linear curve, whose self-intersections are proper crossings, contained in the interior of a closed disc.
Let us two-color the regions in the complement of $C$ so that
two regions sharing a non-trivial part of the boundary receive distinct colors. The fact that this is possible is well known.
A point not lying on $C$ is \emphh{outside} $C$ if it is contained in the region with the same color as the region containing the boundary of the disc, otherwise it is \emphh{inside} $C$.

\begin{figure}
\centering
\subfloat[]{\label{fig:flip}
\includegraphics[scale=0.7]{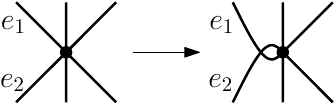}} 
\hspace{30pt}
\subfloat[]{\label{fig:pull}
\includegraphics[scale=0.7]{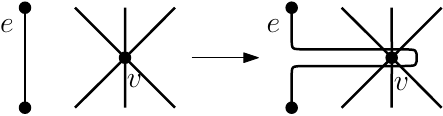} }
\caption{(a) The operation of flip applied to the edges $e_1$ and $e_2$; (b) The operation of pulling the edge $e$ over vertex $v$.}
\end{figure} 

\paragraph{Concatenation of curves.}
We often use the standard operation of concatenation of a pair of curves which is defined as follows. (Recall that we work with non-parametrized curves.)
Let $C_1$ and $C_2$ be a pair of curves sharing an end vertex.
The \emphh{concatenation} of $C_1$ and $C_2$ is a curved obtained as the union of $C_1$ and $C_2$.

\paragraph{Flips and finger moves.}
Refer to Figure~\ref{fig:flip}.
A \emphh{flip} in a drawing of a graph is a modification of the drawing performed in a close neighborhood of a vertex $v$ that switches the order of two consecutive end pieces of edges in the rotation of a vertex while introducing  a single crossing between $e_1$ and $e_2$, but
not affecting the number of crossings between any other pair of edges. The flip of a pair $e_1$ and $e_2$ causes $e_1$ and $e_2$ to cross an odd number of times if $e_1$ and $e_2$ cross an even number of times before the flip and vice-versa. The parity of the number of crossings of no other pair of edges is affected by the flip. 
A flip of $e_1$ and $e_2$ is always possible to perform in a sufficiently small neighborhood $v$.

Refer to Figure~\ref{fig:pull}.
The operation of \emphh{pulling an edge $e$ over a vertex $v$},
also known as \emphh{finger move}, in a drawing is a (generic) continuous deformation of $e$ of the drawing in which $e$ passes over $v$ exactly once and does not pass over any other vertex in the drawing.
Pulling an edge $e$ over a vertex $v$ in a drawing results in a change of the parity
of the number of crossings between $e$ and every edge incident to $v$.
Thus, if we are interested only in the parity of the number of crossings between pairs of edges, the operation of pulling an edge $e$ over a vertex incident to $e$ can be simulated by flips.\jk{ok; ale je potreba se o tom zminovat? pulling operace ma smysl prave jen kdyz ten vrchol neni incidentni s $e$...?}

\begin{figure}
\centering
\subfloat[]{\label{fig:contraction}
\includegraphics[scale=0.7]{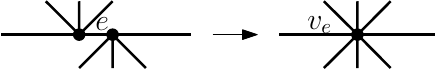}} 
\hspace{30pt}
\subfloat[]{\label{fig:split}
\includegraphics[scale=0.7]{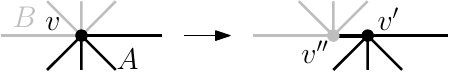} }
\caption{(a) Contraction of the edge $e$; (b) The vertex $A$-split applied to $v$.}
\end{figure} 

\paragraph{Contraction and split.}
Refer to Figure~\ref{fig:contraction}.
We  use the following operation from~\cite{PSS06_removing}.
A \emphh{contraction} of an edge $e=uv$, where $u\not=v$, in a drawing of a graph $G$ is a modification of the drawing  resulting in a drawing of $G/e$ carried out as follows. We turn $e$ into a vertex $v_e$ by moving $v$ along $e$ towards $u$ while dragging all the other edges incident to $v$ along $e$ as illustrated in~\cite[Figure 1]{PSS06_removing}.
Note that by contracting an edge in an (independently) even drawing, we obtain again an (independently) even drawing. We will use this observation tacitly throughout the paper.
By the contraction we can introduce multiple edges or loops incident to the vertex that $e$ is contracted into.
There is a natural correspondence between the edges in the graph before and after the application of the contraction, where the identical edges \emphh{correspond} to each other and an edge $uw$ and $vw$, respectively, \emphh{corresponds} to the edge $v_ew$.

Refer to Figure~\ref{fig:split}.
We  also use the following operation which can be thought of as the inverse operation of the edge contraction
in a drawing of a graph.
Let $v\in V(G)$ and let $A$ be a subset of neighbors of $v$ or the subset of edges in $G$ incident to $v$. Let $B$ denote the complement of $A$ in the set of neighbors of $v$ or in the set of edges incident to $v$, respectively.
A \emphh{vertex $A$-split} in $G$ results in the graph $G_0$ obtained from $G$ by first removing $v$ with all its incident edges, and second inserting a pair of vertices $v'$ and $v''$, the edge $v'v''$, all the edges $v'u$, where $u\in A$ or $vu\in A$, respectively, and all the edges $v''u$, where $u\in B$ or $vu\in B$, respectively.
There is a natural correspondence between the edges in the graph before and after the application of the vertex $A$-split, where the identical edges \emphh{correspond} to each other and an edge  $uv'$ and $uv''$, respectively, \emphh{corresponds} to an edge $uv$.
A \emphh{vertex $A$-split} in a drawing of  $G$ is a local modification of the drawing resulting in a drawing of $G_0$, in which $v'$ and $v''$ are drawn in a small neighborhood of $v$ and are joined by a short crossing-free edge such that the rotations at $v'$ and $v''$ are inherited from the
rotation at $v$, and the new edges are drawn in the small neighborhood of the edges they correspond to in $G$. Here, we assume that the edges between $v$
and $A$ appear consecutively in the rotation at $v$.
We assume that the modification is carried out so that the parity of the number of crossings between 
a pair of edges in $G$ in the drawing we started with is the same as between the corresponding pair of edges in the modified drawing of $G_0$.
When $A$ is clear from the context we will refer to
 the vertex $A$-split as to a vertex split.

The previously defined correspondence relation is a transitive relation; and we extend it to its transitive closure, which will be useful in case of successive applications of contractions and splits. We also naturally extend the correspondence relation to the sets of edges and subgraphs induced by a set of edges. \\

A self-crossing of an edge created by an application of one of the above operations is eliminated by a standard argument, see e.g.~\cite[Figure 2]{PSS06_removing}. The same argument is tacitly also applied in the sequel when such self-crossings are introduced.

\paragraph{Uncorrectable vertices.}
Many technical difficulties we face in the proof of our result stem from the presence of vertices of $G$ in a $\mathbb{Z}_2$-approximation of $(G,H,\varphi)$ that cannot be made even by flips.
It is easy to see that such a vertex cannot be of degree less than $4$. On the other hand,
it is not hard to see that if in a drawing every 4-tuple of edges incident to a vertex does not cause a problem then the vertex can be made even by flips\footnote{We are not aware of anybody explicitly mentioning this fact in the literature.}.
\jk{mozna to plyne z dukazu v Unified Hanani Tutte, pripad 3-souvislych grafu (kde se nasla hvezdovita obstrukce z vrcholu $v$ a vrcholu s indexy $0,i,k,j$). Ale explicitne to tam recene asi neni. Ten dukaz je vpodstate stejny jako ten nize.}

\begin{claim}
If the edges in every 4-tuple of edges incident to a vertex $v$ in a drawing of a graph can be made cross one another an even number of times by flips then we can make $v$ even by flips.
\label{claim:4}
\end{claim}
\begin{proof}
Consider a vertex $v$ that is not even in the drawing. Let $e$ be an edge incident to $v$.
By performing flips at $v$ we easily achieve that every other edge incident to $v$ crosses $e$ evenly. Here, we do not care if we change the parity of crossings between other pairs of edges incident to $v$.
Then we keep flipping pairs of edges incident to $v$ crossing each other oddly and contiguous (neighboring) in the rotation at $v$. Note that this procedure must terminate, since at every step we decrease the number of pairs of edges crossing an odd number of times. Suppose that the procedure does not render $v$ even.
Let $vu$ and $vw$ be the closest pair in the rotation at $v$ (w.r.t. the clockwise order starting with $e$) crossing each other an odd number of times. Since $vu$ and $vw$ are not consecutive in the rotation at $v$
(otherwise we could flip them) there exists an edge $e'$ between $vu$ and $vw$ (in the above mentioned order) crossing both $vu$ and $vw$ an even number of times.
Then $vu, vw,e$ and $e'$ form a desired 4-tuple of edges incident to $v$ that cannot be made even by flips. This can be seen by a simple case analysis, but we provide a nicer proof below.

We consider the graph $G_v=(\{vu,vw,e,e'\},E(G_v))$, in which two vertices $y$ and $z$ (that are also edges of $G$) are joined by an edge if and only if  $y$ and $z$ cross an odd number of times in the drawing of $G$.
Let ${\bf z}=(z_1,z_2,z_3,z_4)\in \mathbb{Z}_2^4$ denote the degree sequence of $G_v$, in which $z_i=\degr(v_i) \mod 2$, where $(v_1=e,\ldots,v_4)$ is the restriction of the rotation at $v$  in the drawing of $G$ to $\{vu,vw,e,e'\}$. 
Observe that ${\bf z}=(1,0,1,0)$ and that an application of a flip in this particular case either leaves ${\bf z}$ unchanged or adds all 1s vector to ${\bf z}$. The same holds if ${\bf z}=(0,1,0,1)$.
However, ${\bf z}=(0,0,0,0)$ if every pair among $\{vu,vw,e,e'\}$ cross an even number of times, which concludes the proof.
\end{proof}

\jk{nejaka subsection nazvana cycle reduction? nejaka uvodni veta ze se chystame definovat *a* cycle reduction operation? referenci na Revisited, kde jsme tu operaci uz delali? (byla specialnejsi nebo stejna jako tady?) Obrazek?}

\begin{figure}
\centering
\includegraphics[scale=0.7]{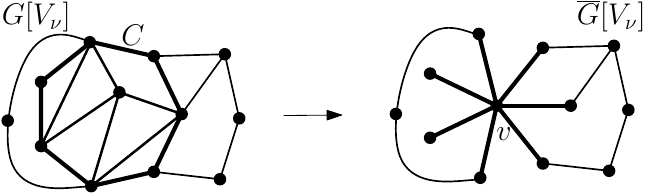}
\caption{Deleting the interior of $C$ in $G[V_{\nu}]$ and applying the generalized \deltaY-operation to $C$.}
\label{fig:deltaY}
\end{figure} 

\paragraph{Clone.}
Given a $\mathbb{Z}_2$-approximation $\psi_0$ of $(G,H,\varphi)$
in Section~\ref{sec:normalform2} we show that there exists an instance $(\hat{G},\hat{H},\hat{\varphi})$ in a so-called normal form that is more convenient to work with admitting a $\mathbb{Z}_2$-approximation $\hat{\psi_0}$ such that $(\hat{G},\hat{H},\hat{\varphi})$ is approximable by an embedding if and only if $(G,H,\varphi)$ is approximable.
Since Theorem~\ref{thm:main} requires that the obtained 
embedding is compatible with the given $\mathbb{Z}_2$-approximation, we additionally require that 
if $(\hat{G},\hat{H},\hat{\varphi})$ is approximable by an embedding compatible with $\hat{\psi_0}$ then $(G,H,\varphi)$ is approximable by an embedding compatible with $\psi_0$.
Thus, we can safely pass to the modified instance. This motives the following definition.

We say that $(\hat{G},\hat{H},\hat{\varphi},\hat{\psi}_0)$  is
a \emphh{clone} of $(G,H,\varphi,\psi_0)$ if the following holds.
If $(\hat{G},\hat{H},\hat{\varphi})$ is approximable by an embedding compatible with $\hat{\psi_0}$ then $(G,H,\varphi)$ is approximable by an embedding compatible with $\psi_0$; and
if $(G,H,\varphi)$ is approximable by an embedding then $(\hat{G},\hat{H},\hat{\varphi})$ is approximable by an embedding.
The 4-tuple $(\hat{G},\hat{H},\hat{\varphi},\hat{\psi}_0)$ is \emphh{approximable by an embedding} if $(\hat{G},\hat{H},\hat{\varphi})$ is approximable by an embedding. 
\begin{remark}
Note that being a clone is a transitive relation.
However, the relation is not symmetric, and thus, it is not an equivalence relation. 
\end{remark}

In the following we need the following claim that was observed by Pelsmajer et al.~\cite[proof of Theorem 3.1]{PSS06_removing}.

\begin{claim}
\label{claim:cycle_correction}
Let $C$ be a cycle of $G$ and let $v\in V(C)$. For any drawing of $G$, there exists a finite sequence of flips at $v$ that makes every pair of edges incident to $v$  involving an edge of $C$ cross an even number of times.
\end{claim}

\begin{proof}
Let $e$ and $f$ denote the edges of $C$ incident to $v$.
By performing appropriate flips we make $e$ and $f$ cross an even number of times.
Let $k$ denote the number of pairs of edges incident to $v$ that involve $e$ or $f$ and 
that cross an odd number of times.
If $k=0$ we are done. 
Suppose for the sake of contraction that $k>0$ is smallest possible after application of any finite sequence of flips at $v$ in the given drawing.

Hence, there exists an edge $g$ incident to $v$ such that $g\not=e,f$ and $g$ crosses, let's say $e$, an odd number of times.
We perform appropriate flips, all involving $g$, and not involving $f$, after which $g$ and $e$ cross an even number of times. It is easy to see that this is indeed possible.
Since $f$ was not involved in any of the flips  we did not change the property of $e$ and $f$ crossing
an even number of times, and  we did not change the parity of the number of crossings between $f$ and any other edge. Since $g$ was involved in all the flips we did not change the parity of the crossings between $e$ and any other edge than $g$. Hence, we obtained a drawing contradicting the choice of $k$ which concludes the proof. 
\end{proof}

\paragraph{Cycle reduction.}
Refer to Figure~\ref{fig:deltaY}.
In order to bring a given instance into a normal form we use the operation of cycle reduction, whose repeated application turns 
the subgraphs of $G$ induced by clusters into forests.
Given $(G,H,\varphi)$, suppose that there exists a cycle $C\subseteq G[V_{\nu}]$.
Let $\psi_0$ be a $\mathbb{Z}_2$-approximation of $(G,H,\varphi)$ which is an independently even drawing.
Roughly, the cycle reduction removes the subgraph of $G$ inside of $C$ and replaces $C$ with a star by
using the generalized \deltaY operation.
Here, we use the notion of being ``inside'' and ``outside'' as defined in the beginning of this subsection.
The \emphh{cycle reduction} of $C$ in $G$ is the operation that returns $(\overline{G},\overline{H},\overline{\varphi})$ together with its $\mathbb{Z}_2$-approximation $\overline{\psi_0}$ obtained from $(G,H,\varphi)$ and $\psi_0$ as follows.

We put $\overline{H}:=H$.
If there exists an edge of $C$ that does not cross every other edge in $G$ an even number of times in $\psi_0$, we modify $\psi_0$ locally at the vertices of $C$ by performing appropriate flips so that this is not the case by using Claim~\ref{claim:cycle_correction}. The set of vertices of the cycle $C$ forms a (possibly trivial) cut in $G$, i.e.,  either side of the cut might be empty, splitting the vertex set of $G\setminus V(C)$ into the set of vertices $V_{in}=V_{in}(C)$  inside of $C$ and $V_{out}=V_{out}(C)$ outside of $C$, where $V_{in}\subset V_{\nu}$. 
Let the set of \emphh{inner diagonals} $D_{in}=D_{in}(C)$ and \emphh{outer diagonals} $D_{out}=D_{out}(C)$ be the set of edges of $G[V(C)]$ both of whose end pieces at $C$ start inside and outside of $C$, respectively. Note that the edge set of $G[V(C)]$ is partitioned into $E(C),D_{in}$ and $D_{out}$.
The graph $\overline{G}$ is obtained, first, by removing the vertices in $V_{in}$ (and their incident edges) from $G$ and the edges in $D_{in}$, and second, by performing the generalized \deltaY operation on $C$; recall Definition~\ref{def:deltaY}.
The map $\overline{\varphi}$ is naturally inherited from $\varphi$. 
By following the approach used in the proof of~\cite[Theorem 2]{FKMP15}, the new vertex $v$ introduced by the generalized \deltaY operation is drawn in $\overline{\psi_0}$ very close to an arbitrary vertex of $C$ inside of $C$, and its adjacent edges closely follow the edges of $C$ in $\psi_0$ (now deleted)  while staying on the same side. We draw the new edges so that the rotation at $v$ in the resulting drawing corresponds to the order in which the neighbors of $v$ appear on $C$, i.e.,
every crossing between a pair of edges incident to $v$ corresponds to a crossing between a pair of edges on $C$. Clearly, new edges can be drawn so that the resulting drawing satisfies all the general position assumptions stated in the introduction.
We show that the resulting drawing is independently even.

\begin{claim}
\label{claim:evenness}
Suppose that ${\psi_0}$ is independently even.
Then the drawing $\overline{\psi_0}$ is independently even and the vertex $v$ is even in $\overline{\psi_0}$ and the rotation at $v$ in $\overline{\psi_0}$ is the same (up to the orientation) as the order of the neighbors of $v$ on $C$.
\end{claim}
\begin{proof}
Since the edges of $C$ cross every other edge of $G$ evenly, the vertex $v$ is even in $\overline{\psi_0}$.
Then the drawing $\overline{\psi_0}$ is independently even,
since the edges of $\overline{G}$ not adjacent to $v$ are drawn exactly as in $\psi_0$.
The last part of the claim follows by the construction of $\overline{\psi}_0$.
\end{proof}

By using the previous claim we can establish the following useful tool.

\begin{claim}
\label{claim:cycle_red}
Let $(\overline{G},\overline{H},\overline{\varphi},\overline{\psi}_0)$ be obtained by a cycle reduction from $(G,H,\varphi,\psi_0)$ as described above.
$(\overline{G},\overline{H},\overline{\varphi},\overline{\psi}_0)$ is a clone of  $(G,H,\varphi,\psi_0)$.
Moreover, if $G$ is connected we have $|E(G)|-|V(G)|+1>|E(\overline{G})|-|V(\overline{G})|+1\ge 0$.
\end{claim}
\begin{proof}
Note that the drawings $\psi_0$ and $\overline{\psi_0}$ are the same on the restriction to the subgraph $(G\setminus V_{in})\setminus (E(C)\cup D_{in})$. Thus, an embedding of $\overline{G}\setminus v$, where $v$ was introduced by the generalized \deltaY operation, which is compatible with the  restriction of $\psi_0$ to $(G\setminus V_{in})\setminus (E(C)\cup D_{in})$ is also compatible with $\overline{\psi_0}$ and vice-versa.
We show the two required conditions of a clone.

{\bf First,} we show that if $(\overline{G},\overline{H},\overline{\varphi})$ is approximable by an embedding compatible with $\overline{\psi_0}$ then $(G,H,\varphi)$ is approximable by an embedding compatible with $\psi_0$.
Given the approximation $\overline{\psi}$ (compatible with $\overline{\psi_0}$) of $(\overline{G},\overline{H},\overline{\varphi})$ we obtain an approximation $\psi$ compatible with $\psi_0$ of 
$(G,H,\varphi)$ as follows. First, we perform the generalized \yDelta operation on the vertex $v$ introduced by the generalized \deltaY operation. Let $C$ be the cycle obtained by this operation. The resulting graph
is $G_{out}=(G\setminus V_{in})\setminus D_{in}$. We obtain an approximation $\psi_{out}$ of $(G_{out},H|_{\varphi(G_{out})},\varphi|_{G_{out}})$, where $\varphi|_{G_{out}}$ is the restriction of $\varphi$ to $G_{out}$, by the following modification of $\overline{\psi}$. Let $v_0,v_1,\ldots, v_{k-1}$ denote the vertices of $C$ listed in the order of appearance on $C$, i.e., $v_iv_{i+1 \mod k}\in E(C)$ for $i=0,\ldots,k-1$. 
We first draw the edges $v_iv_{i+1 \mod k}$ by closely following the path $v_ivv_{i+1 \mod k}$ in $\overline{\psi}$, and second, discard the edges $v_iv$.
By Claim~\ref{claim:evenness}, up to the choice of orientation, the rotation at $v$ is $(vv_0,vv_1,\ldots, vv_{k-1})$, and thus, the modification can be performed without introducing any edge crossings. Note that the embedding of $C$ in $\psi_{out}$ bounds a topological disc. We assume that $\psi_{out}$ is contained in a surface with a single boundary component formed by $\psi_{out}(C)$. This is clearly possible even if $G_{out}$ has more than one connected component.

Suppose for a while that there exists an embedding $\psi_{in}$ of $G_{in}=(G\setminus V_{out})\setminus D_{out}$ in the plane such that $\psi_{in}(C)$ bounds a face in $\psi_{in}$ and such that $\psi_{in}$ is compatible with the restriction of $\psi_0$ to $G_{in}$. Without loss of generality, we assume that $\psi_{in}$ is contained in a discs in which $\psi_{in}(C)$ forms the boundary. This is clearly possible even if $G_{in}$ has more than one connected component.
Second, we merge both embeddings by identifying surfaces containing them on their boundaries by a homeomorphism mapping bijectively $\psi_{out}(v_iv_{i+1 \mod k})$, for $i=0,\ldots, k-1$, to  $\psi_{in}(v_iv_{i+1 \mod k})$.
By the discussion in the first paragraph of the proof, the obtained embedding is compatible with $\psi_0$. Indeed, the compatibility cannot be violated by the vertices of $C$, since an even vertex, let's say $v_i$, in $\psi_0$ has in both $\psi_{in}$ and $\psi_{out}$
the same rotation as in the respective restriction of $\psi_0$.

It remains to show that the restriction of $\psi_0$ to $G_{in}$ 
implies that there exists $\psi_{in}$ 
such that $\psi_{in}(C)$ bounds a  face and so that $\psi_{in}$ is compatible with the restriction of $\psi_0$ to $G_{in}$.
To this end, in the drawing $\psi_0$ of $G_{in}$, we apply the vertex $A_v$-split to every $v\in V(C)$, where $A_v=\{u,w| \ vw,vu\in E(C)\}$.
Applying the unified Hanani--Tutte theorem to the resulting
independently even drawing, and contracting the edges introduced previously by vertex $A_v$-splits in the obtained embedding yields an embedding of $G_{in}$ with the required property.
 
{\bf Second,}  given an approximation $\psi$ of $({G},{H},\varphi)$ we obtain an approximation $\overline{\psi}$ of 
$(\overline{G},\overline{H},\overline{\varphi})$ analogously by removing $V_{in}$ and $D_{in}$, using the generalized \deltaY operation on $C$, and embedding $v$ together with its adjacent edges in the face of $\psi$ bounded by $C$.

Finally, we prove  the ``moreover'' part of the claim. The number of edges incident to the vertices of $V_{in}$ belonging to the same connected component as $C$ in $G$ is at least $|V_{in}|$. Thus, $|E(\overline{G})|\le |E(G)|-|V_{in}|$,  and
$|V(G)|+1-|V_{in}|=|V(\overline{G})|$. 
By summing the equation with the inequality, we obtain $|E(\overline{G})|+|V(G)|+1-|V_{in}|\le |E(G)|-|V_{in}|+|V(\overline{G})|$ and the claim follows.
\end{proof}

\paragraph{Local obstructions.}
Characterizing $\mathbb{Z}_2$-approximable instances $(G,H,\varphi)$ by a Kuratowski-style characterization seems to be unfeasible. Nevertheless, there exists a pair of local obstructions to $\mathbb{Z}_2$-approximability corresponding to $K_{3,3}$ and $K_5$, which was observed already by Minc~\cite{Minc97_arcs};
see Figure~\ref{fig:config}.

A \emphh{YY-configuration} in $(G,H,\varphi)$ is a pair of $3$-stars $S_1\subseteq G$ and $S_2\subseteq G$ intersecting possibly only in leaves, such that $\varphi(S_1)$ is a $3$-star and $\varphi(S_1)=\varphi(S_2)$.

An \emphh{X-configuration} in $(G,H,\varphi)$ is a pair of $2$-stars $S_1\subseteq G$ and $S_2\subseteq G$, all of whose edges are pipe edges, and such that there are vertices $\nu,\mu_1,\ldots, \mu_4$ of $H$ satisfying 
\begin{itemize}
\item $\varphi(S_1)=(\{\nu,\mu_1,\mu_{2}\}, \{\nu \mu_1,\nu \mu_{2}\})$,
\item $\varphi(S_2)=(\{\nu,\mu_3,\mu_{4}\}, \{\nu \mu_{3},\nu \mu_{4}\})$, and 
\item in $H$, the edges $\nu\mu_1,\nu\mu_2$ alternate with $\nu\mu_3,\nu\mu_4$ in the rotation at $\nu$. 
\end{itemize}

\begin{claim}
\label{claim:YY}
A YY-configuration cannot occur in $(G,H,\varphi)$ if $(G,H,\varphi)$ is $\mathbb{Z}_2$-approximable.
\end{claim}

\begin{claim}
\label{claim:cross}
An X-configuration cannot occur in $(G,H,\varphi)$ if $(G,H,\varphi)$ is $\mathbb{Z}_2$-approximable.
\end{claim}

The two previous claims follow from a more general statement that we formulate next.

Refer to Figure~\ref{fig:derivative05} (left).
For an instance $(G,H,\varphi)$, we define a graph that captures the connectivity between the connected components of $G$ induced by a cluster $V_{\nu}$, for $\nu\in V(H)$, and the rest of the graph $G$.
Let $G_{\nu}$ be the graph whose vertex set  is
in the bijection with the disjoint union of the set of the edges of $H$ containing
$\nu$ and the set of connected components of $G[V_\nu]$. This bijection is denoted by superscript~$*$.
\[
V(G_\nu)=\{(\nu\mu)^*| \ \nu\mu\in E(H) \} \cup \{ C^*| \ C \mathrm{\ is \ a \ connected \ component \ of } \ G[V_{\nu}] \}.
\]
The edge set of $G_\nu$ is defined as follows. A pair of vertices $(\nu\mu_1)^*$ and $(\nu\mu_2)^*$ is joined by an edge in $G_{\nu}$ if and only if $\nu\mu_1$ and $\nu\mu_2$ are consecutive in the rotation at $\nu$ (in $H$).
If $\nu$ is incident to at least three edges, let $C_{\nu}$ be the cycle in $G_{\nu}$ induced by $\{(\nu\mu)^*| \ \nu\mu\in E(H) \}$.
Furthermore, we have an edge $C^*(\nu\mu)^*$ in $G_\nu$, if and only if $\mu\in \pNeigh{C}$.
See Figure~\ref{fig:derivativeDef} (left) for an illustration.

\begin{figure}
\centering
\includegraphics[scale=0.8]{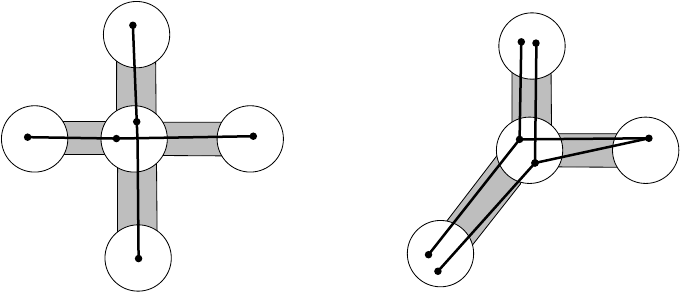}
\caption{Analogs of Kuratowski obstructions to graph planarity for approximating maps of graphs by embeddings, X-configuration (on the left) and YY-configuration (on the right).}
\label{fig:config}
\end{figure}

\begin{claim}
\label{claim:Gi}
If $(G,H,\varphi)$ is $\mathbb{Z}_2$-approximable 
then $G_{\nu}$ is planar;
and if additionaly, the vertex $\nu$ in $H$ has degree at least $3$ then $G_{\nu}$ admits an embedding in the plane in which the cycle $C_{\nu}$ bounds the outer face.
\end{claim}

\begin{proof}
The claim is trivial if $\nu$ has degree $2$, since
then $G_{\nu}$ is formed by a collection of paths
of length at most $2$, all joining the same pair of vertices,
and possibly some additional edges, each incident to a degree-1 vertex.
Hence, we assume that $\nu$ has degree at least $3$.
We construct an independently even drawing of $G_{\nu}$ in which the cycle $C_{\nu}$ is crossing-free and the rest of the drawing is inside $C_{\nu}$.
For every vertex $v\in V(C_{\nu})$, let $A_v=\{u,w\}$ denote the set of the two neighbors of $v$ in $V(C_{\nu})$ such that $uv,vw\in E(C_{\nu})$.
In $G_\nu$, we apply the vertex $A_v$-split to every vertex $\nu$ of $C_{\nu}$.
The set of vertices $v'$ adjacent to the vertices in $A_v$,
for all $v\in C_{\nu}$, form a cycle in the modified graph which we, by a slight abuse of notation, denote by $C_{\nu}$.
 Note that now all the vertices of $C_{\nu}$ are even.
Then by the unified Hanani--Tutte theorem we obtain an embedding of the modified graph that yields a required embedding of $G_{\nu}$ after contracting edges created previously by vertex $A_v$-splits.

The independently even drawing, to which we apply  the unified Hanani--Tutte  theorem, is obtained from the restriction of a $\mathbb{Z}_2$-approximation of $(G,H,\varphi)$ 
to $\mathcal{D}(\nu)$. In the restriction, we contract the valve of $\nu\mu$, for every $\nu\mu\in E(H)$, together with the severed end points of the parts of pipe edges in $\mathcal{D}(\nu)$ to a point. Such point represents the vertex $(\nu\mu)^*$ in the drawing. Then the edges of $C_{\nu}$ are drawn along the boundary of $\mathcal{D}(\nu)$.
Finally, we contract every connected component of $G[V_{\nu}]$ to a vertex and discard created loops and multiple edges.
\end{proof}


\section{Normal form and derivative}
\label{sec:normal-form}

We assume that $(G,H,\varphi)$ is $\mathbb{Z}_2$-approximable throughout this section.

\begin{figure}
\centering
\includegraphics[scale=0.8]{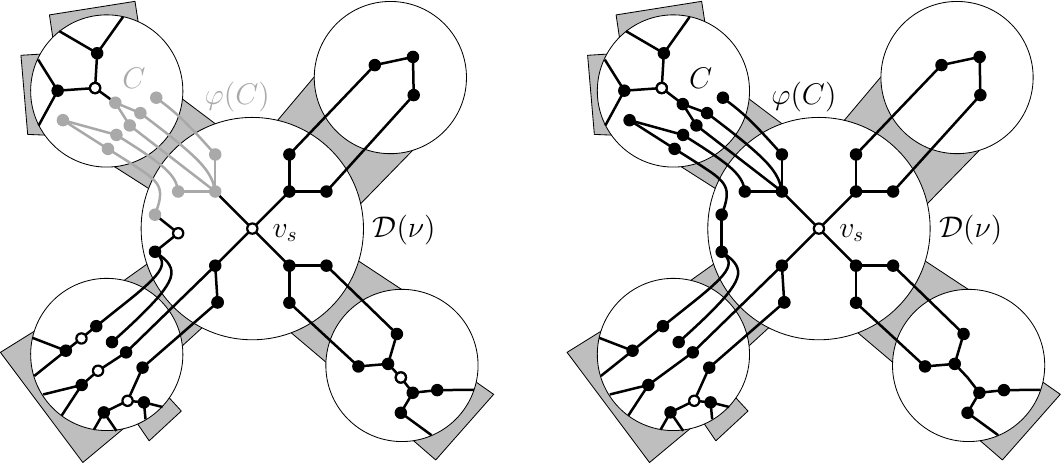}
\caption{A part of $(G,H,\varphi)$ in the subdivided normal form (left)  illustrated by its approximation in $\mathcal{H}$, and in the normal form (right). The empty vertices belong to the independent set $V_s\subset V(G)$. A connected component $C$ of $G[V\setminus V_s]$ is colored gray.}
\label{fig:normalform}
\end{figure}

\subsection{Normal form}

\paragraph{Intuitive explanation.} We define the normal form of an instance $(G,H,\varphi)$ to which we can apply the derivative. Recall that $\varphi:G \rightarrow H$. 
In order to keep the definition more compact we define the normal form via its topologically equivalent subdivided variant.
This variant also facilitates the definition of the derivative.
Roughly speaking, $(G,H,\varphi)$ is in the subdivided normal form if there exists an independent set $V_s\subset V(G)$ without degree-$1$ vertices such that every connected component $C$ of $G[V\setminus V_s]$ is mapped by $\varphi$ to an edge $\varphi(C)=\nu\mu$ of $H$ and both its parts mapped to $\nu$ and $\mu$ are forests. We call vertices in $V_s$ \emphh{central vertices}, which conveys an intuition that every vertex of $V_s$ constitutes in some sense a center of a connected component induced by a cluster.

The normal form is obtained from the subdivided normal form by suppressing in $V_s$ any vertices of degree~2, i.e., by replacing each such vertex $v_s$ and both its incident edges by a single edge, and performing the same replacement for $\varphi(v_s)$; see Figure~\ref{fig:normalform} for an illustration.

\begin{definition}
\label{def:normal}
The instance $(G,H,\varphi)$ is in the \emphh{subdivided normal form} if 
 there exists an independent set $V_s\subset V(G)$ with the following properties. \\
 (I) For every connected component $C$ of $G[V\setminus V_s]$:
 \begin{enumerate}[(1)]
 \item
 \label{it:first}
the image $\varphi(C)$ is an edge of $E(H)$; and
 \item
 \label{it:second}
 $\varphi^{-1}|_{\nu}[\varphi(C)]$ is a forest for both vertices $\nu\in \varphi(C)$.
 \end{enumerate}
(II) For every connected component $C$ of $G[V_\nu]$, for every $\nu\in V(H)$, we have $|V(C) \cap V_s|= 1$ if $\pdeg{C}\ge 2$, and $V(C) \cap V_s= \emptyset$ otherwise;
 for every $v_s\in V_s\cap V(C)$ we have that $deg(v_s)=\pdeg{C}$; and no two edges incident to $v_s$ join $v_s$ with connected components $C_1$ and $C_2$ of $G[V\setminus V_s]$ such that $\varphi(C_1)=\varphi(C_2)$.

The instance obtained from an instance $(G,H,\varphi)$ in the subdivided normal form by suppressing all degree-$2$ vertices in $V_s$ is in the \emphh{normal form}.
Such an instance in the normal form \emphh{corresponds} to the original instance $(G,H,\varphi)$ in the subdivided normal form, and vice-versa.
\end{definition}

\begin{remark}
Regarding the subdivided normal form,
the condition~(\ref{it:first}) implies that in an instance in the normal form there exists no connected component induced by $V_\nu$ 
 of pipe degree $0$ in $G[V\setminus V_s]$, for any $\nu\in V(H)$.
It follows from the definition that the degree of $v_s\in V_s$ is at least two and equals to the number of connected components
of $G[V\setminus V_s]$, that $v_s$ is adjacent to, each of which
is mapped by $\varphi$ to a different edge of $H$.
\end{remark}

\subsection{Derivative}
The rest of the section is inspired by the work of M.~Skopenkov~\cite{Sko03_approximability}
and also Minc~\cite{Minc97_arcs}.
In particular, the notion of the derivative in the context of map approximations was introduced by Minc and adapted to the setting of $\mathbb{Z}_2$-approximations by M.~Skopenkov for instances $(G,H,\varphi)$ where $G$ is subcubic and $H$ is a cycle.
We extend his definition to instances $(G,H,\varphi)$ in the normal form. 
(A somewhat simplified extension was already used in~\cite{F17_pipes}.)
Thus, by derivating $(G,H,\varphi)$ we, in fact, mean bringing the instance $(G,H,\varphi)$ into the normal form and then derivating the instance in the normal form.
Given that the instance is in the normal form or the subdivided normal form, the operation of the derivative outputs an instance $(G',H',\varphi')$, where $G'=G$ and $H'$ along with the manifold $M'$ containing $H'$ are defined next.

In order to keep the exposition more compact we formulate it first for the instances in the subdivided normal form. Thus, in the following we assume $(G,H,\varphi)$ to be in the subdivided normal form.

\subsubsection{Derivative for (subdivided) normal form.}

\paragraph{Construction of $H'$.} Refer to Figure~\ref{fig:derivative3}. 
Let $V_s$ be a set of central vertices in $G$.
Let $\mathcal{G}$ be a bipartite graph with the vertex set $V_s \cup \mathcal{C}$, where $\mathcal{C}=\{C| \ C$ is a connected component of $G[V\setminus V_s]\}$, in which $v_s\in V_s$ and $C \in \mathcal{C}$ are joined by an edge if and only if $v_s$ is joined by an edge with a vertex of $C$ in~$G$. 
Let the \emphh{star} of $v_s$, $\st{v_s}$,  with $v_s\in V_s$ be the subgraph
of $G$ induced by  $\{v_s\} \cup \bigcup_{C\in V(\mathcal{G}): \ v_sC\in E(\mathcal{G})}V(C)$.
The vertices of $H'$ are in the  bijection with the union of $V_s$ with the set of the edges of $H$. This bijection is denoted by superscript~$*$.
Let $H'$ be a  bipartite graph such that
$V(H')=\{\rho^*| \rho\in E(H)\} \cup \{v_s^*| \ v_s\in V_s \}$.
 We have $\rho^*v_s^*\in E(H')$ if and only if $\rho \in E(\varphi(\st{v_s}))$. We use the convention of denoting a vertex in $V(H')$ whose corresponding edge in $E(H)$ is $\rho=\nu\mu$ by both $\rho^*$ or~$(\nu\mu)^*$. \\

 Recall that $V_\nu = \varphi^{-1}[\nu]$, for $\nu\in V(H)$. Let $H_{\nu}'$ be the subgraph of $H'$ induced by $\{v_s^*|\ v_s\in V_{\nu} \cap V_s\} \cup \{(\nu\mu)^*|\nu\mu\in E(H)\}$.
 Note that $H_{\nu}'$ is a subgraph of $G_{\nu}$ (defined in the ``Local obstructions'' paragraph in Section~\ref{sec:tools}).
 Hence, by Claim~\ref{claim:Gi} every $H_{\nu}'$ is a planar graph since $(G,H,\varphi)$ is $\mathbb{Z}_2$-approximable, which we assume throughout this section.

Next, we construct the embedding of $H'$ by specifying its rotation system, and if $M$ is non-orientable also specifying the signs of its edges. The rotation system and the signs will combinatorially encode the cellular embedding of $H'$ on a surface $M'$, whose Euler genus is at most as large as the Euler genus of $M$. 
We read the rotation system off an embedding of $H'$ on $M$ constructed next.

 In the sequel we work with embeddings $G_{\nu}$'s given  by Claim~\ref{claim:Gi}.
In what follows we combine the  embeddings of $G_{\nu}$'s,  $\nu\in V(H)$, with the help of an auxiliary graph $H_{aux}$, so that in the rotation at $(\nu \mu)^*$,
for all $\nu\mu\in E(H)$, the edges of $G_{\nu}$ do not alternate with the edges of $G_{\mu}$.
 After deleting  edges in $G_{\nu}\setminus H_{\nu}'$, for every $\nu\in V(H)$, this will result in an embedding of $H'$ on $M$ (which is not necessarily cellular).

\begin{figure}
\centering
\includegraphics[scale=0.8]{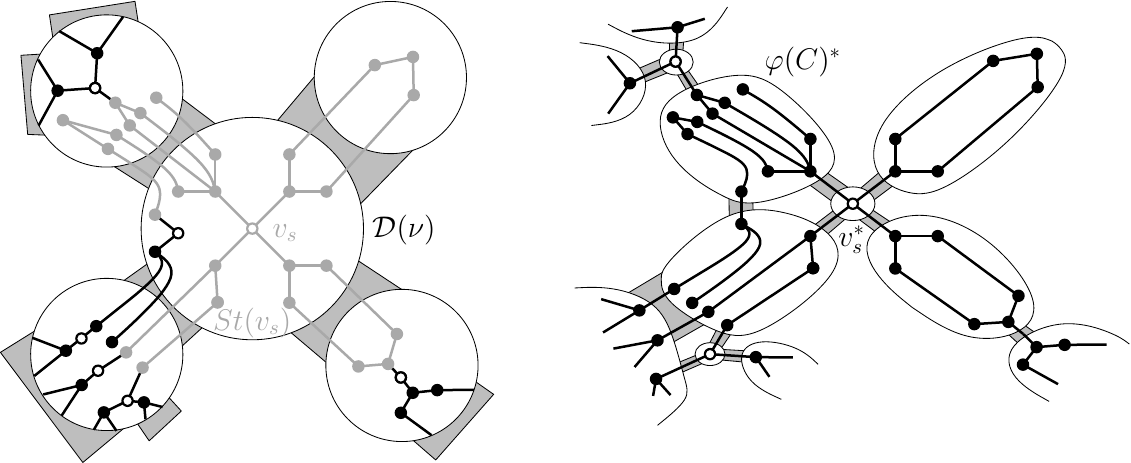}
\caption{A part of $(G,H,\varphi)$ in the subdivided normal form  illustrated by its approximation in $\mathcal{H}$ with $\st{v_s}$  colored gray (left).  Derivative of the same instance after being brought to the normal form  (right).}
\label{fig:derivative3}
\end{figure}

\paragraph{Construction of  the embedding of $H'$.} 
 Refer to Figure~\ref{fig:derivativeDef}.
  The auxiliary graph  $H_{aux}$ is obtained from $H$ as follows.
 First, we subdivide all the edges of $H$. Second, in the resulting graph $H_0$ we apply the generalized \yDelta operation to all the vertices of degree at least $3$.  Third, we suppress the  former vertices of degree $2$, that is, those vertices that were not introduced by subdivisions in the first step.
 
 We turn the embedding of $H$ on $M$ into an embedding on $M$ of  $H_{aux}$ by following the three steps of the construction of $H_{aux}$. To this end,  for every edge $\nu\mu\in E(H)$, we fix one of its vertices to be an \emphh{attractor} of $\nu\mu$. 
 
 In the first step, we perform the subdivision on every edge $\nu\mu=\rho\in E(H)$ with the attractor $\mu$ so that 
the subdividing vertex, let's denote it $\nu_\rho$, is embedded very close to the attractor  of $\rho$. In particular, the edge joining $\nu_\rho$ and the attractor vertex $\mu$ of $\rho$ does not pass through any cross-cap $M$ if $M$ is non-orientable.
In the second step, we embed cycles introduced by the generalized 
 \yDelta operation so that for every $\nu$ of degree at least 3
 and every pair of consecutive edges $\nu\mu_1$ and $\nu\mu_2$ in the rotation at $\nu$ in the embedding of $H$ the edge $\nu_{\nu\mu_1}\nu_{\nu\mu_2}$ of $H_{aux}$ closely follows 
 the embedding of the edges $\nu_{\nu\mu_1}\nu$ and $\nu\nu_{\nu\mu_2}$ in the embedding of $H_0$.
 In the third step, we do not change the image of the embedding.
 
  In order to construct  a desired embedding of $H'$, we  merge  the constructed embedding of $H_{aux}$ with the embeddings of $G_{\nu}$'s, $\nu\in V(H)$, obtained by Claim~\ref{claim:Gi},  one by one  as follows. If $\nu\in V(H)$ has degree at least 3, let $C_{\nu}^\Delta$ denote the cycle in $H_{aux}$ obtained from $\nu$ by the generalized \yDelta operation.
We cut out from $M$ the interior of the disc bounded by the embedding of
$C_\nu^\Delta$.
Let  $D_\nu$ denote a disc containing the embedding of $G_{\nu}$ obtained by Claim~\ref{claim:Gi}, where the embedding of $C_\nu$ forms the boundary of $D_{\nu}$.
We fill the created hole on $M$ with $D_\nu$ by identifying $(\nu\mu)^*\in C_{\nu}$ with the vertex $\nu_{\nu\mu}$ on $C_{\nu}^\Delta$ subdividing $\nu\mu$ in $H_0$; and homeomorphically extending this identification on the edges
embedded along the boundary of $M$ and $D_\nu$.
For $\nu\in V(H)$, let $\nu\mu_1,\ldots, \nu\mu_{\deg({\nu})}$ be the set of its incident edges in $H$.
We further deform (while preserving the isotopy class) the obtained embedding $G_\nu$,  for every $\nu\in V(H)$, so that its edges incident to  $(\nu\mu_i)^*=\rho_i^*$, for every $i=1,\ldots, \deg({\nu})$, closely follow the embedding of $\nu\mu_i$ towards $\nu$ from $\rho_i^*$ (which is embedded close $\nu_{\rho_i}$)  and divert from  $\nu\mu_i$ only very close to $\nu$ after passing through all the cross-caps along  $\nu_{\rho_i}\nu$. Besides parts of the edges incident to $\rho_i^*$  closely following $\nu_{\rho_i}\nu$ the rest of $G_\nu$ is embedded in a small neighborhood of $\nu$ in the embedding of $H_{aux}$.

 \begin{figure}
\centering
\includegraphics[scale=0.7]{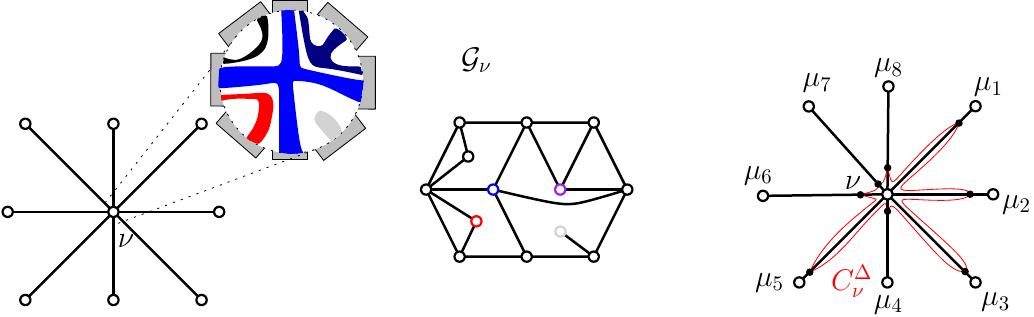}
\caption{The graph $G_{\nu}$ corresponding to a vertex $\nu$ of $H$ (left). Colors encode the correspondence between the connected components of $\varphi^{-1}[{\nu}]$ in $G$ and vertices of $G_{\nu}$. The embedding of $C_{\nu}^{\Delta}$ (right). The black vertices are subdiving the edges of $H$ in $H_0$. The vertex $\nu$ is the attractor of the edges $\nu\mu_4,\nu\mu_6,\nu\mu_7$ and $\nu\mu_8$.}
\label{fig:derivativeDef} 
\end{figure}

Analogously for every former degree-$2$ vertex $\nu\in V(H)$ with neighbors $\mu_1,\mu_2\in V(H)$, we  replace by $G_{\nu}$ the edge $\mu_1\mu_2$  obtained after suppressing $\nu$  as follows. 
If $\nu$ has degree 2, $G_\nu$ is a union of (subdivided) edges between $(\nu\mu_1)^*$ and $(\nu\mu_2)^*$ and some additional edges connecting either of $(\nu\mu_1)^*$ and 
$(\nu\mu_2)^*$ to a degree-1 vertex.
We identify $(\nu\mu_1)^*,(\nu\mu_2)^*\in V(G_{\nu})$ with $\mu_1,\mu_2\in V(H)$, respectively;
embed the rest of $G_{\nu}$ along $\mu_1\mu_2$ and discard the edge $\mu_1\mu_2$. 
Finally, we discard the edges not appearing in $H'$, which might be only the edges on a cycle $C_{\nu}$.

If $M$ is non-orientable we additionally need to assign signs to the edges of $H'$. Roughly, we can picture the above construction of the embedding of $H'$ on $M$ being performed so that a vertex $(\nu\mu)^*$ is drawn very close to 
where the attractor $\mu$ of $\nu\mu$ was drawn in the embedding of $H$. The edges of $H_\nu'$ incident to $(\nu\mu)^*$ closely follow  $\nu\mu$
towards $\nu$. Thus, all the edges of $H_\nu'$  incident to  $(\nu\mu)^*$ will pass through all the cross-caps that $\nu\mu$ passed through.

Formally, we enumerate all the edges in $E(H)$ with negative sign and process them successively in the corresponding order each time flipping the signs of certain edges in $E(H')$.
In the beginning, all the edges in $H'$ have the positive sign.
Processing an edge $\nu\mu\in E(H)$ with the negative sign amounts to 
flipping the signs of all the edges in $H_\nu'$ incident to $(\nu\mu)^*$,
where $\mu$ is the attractor of $\nu\mu$.
This finished the construction of the embedding of $H'$.

\begin{claim}
\label{claim:newSurface}
The rotation system of $H'$ and the signs of its edges in the constructed embedding of $H'$ on $M$,  define a cellular embedding of $H'$ on a surface of $M'$, whose Euler genus is at most as large as the Euler genus of $M$. 
\end{claim}

\begin{proof}
If $M$ is orientable this follows  immediately by the construction in which we did not add any handle to $M$.
If $M$ is non-orientable, we  did not add any cross-cap to $M$ so the Euler genus of $M'$ cannot be larger than the Euler genus of $M'$.

Finally, by the construction, the sign of an edge in $H'$ (given by the above described procedure) is positive if and only if the edge passes through the cross-caps an even number of times.
\end{proof}

\jk{trochu moc operaci najednou... hodil by se obrazek, spis teda film :)} 



\begin{definition}
\label{def:derivativeAbstract}
Let $(G,H,\varphi)$ be $\mathbb{Z}_2$-approximable and in the subdivided normal form with a fixed independent set $V_s$ satisfying~(I) and~(II).
The \emphh{derivative} $(G,H,\varphi)'$ of $(G,H,\varphi)$  is the instance $(G',H',\varphi')$ such that $\varphi'(v_s)=v_s^*$, for $v_s\in V_s$, and 
$\varphi'(v)=\varphi(C)^*$, for every $v\in V(C)$, where $C$ is a connected component of $G[V\setminus V_s]$.
 (Hence, $\varphi(C)$ is an edge of $H$ by~(\ref{it:first}) in
the definition of the subdivided normal form.)

The \emphh{derivative} $(G,H,\varphi)'$ of $(G,H,\varphi)$, where $(G,H,\varphi)$ is $\mathbb{Z}_2$-approximable and in the normal form, is the instance obtained from the derivative of the corresponding instance in the subdivided normal form by suppressing every vertex $v_s$ of degree $2$ in $V_s$ and its image $\varphi(v_s)$ in $H'$, and eliminating multiple edges in $H'$.
\end{definition}

\begin{remark}
Since $(G',H',\varphi')$ is defined only for instances in the (subdivided) normal form, by derivating an instance, which is not in the normal form, we will mean an operation that, first, brings the instance into the normal form by a deterministic algorithm thereby constructing $V_s$,  and second, applies the derivative to the instance. Here, the construction of $V_s$ and $(G',H',\varphi')$ will depend besides $(G,H,\varphi)$ also on its $\mathbb{Z}_2$-approximation $\psi_0$.
We take the liberty of denoting by  $G',H'$ and $\varphi'$ an object that does not depend only on $G,H$ and $\varphi$, respectively, but on the whole instance $(G,H,\varphi)$ and $\psi_0$.
\end{remark}

A derivative of $(G',H',\varphi')$ is then denoted by $(G^{(2)},H^{(2)},\varphi^{(2)})$, and in general we put $(G^{(i)},H^{(i)},\varphi^{(i)})=(G^{(i-1)},H^{(i-1)},\varphi^{(i-1)})'$.

In Section~\ref{sec:derivative_z2approx}, we show that 
if $(G,H,\varphi)$ in the normal form is $\mathbb{Z}_2$-approximable then $(G',H',\varphi')$ is 
$\mathbb{Z}_2$-approximable as well.
More precisely, we prove the following claim.

\begin{claim}
\label{claim:derivative}
If the instance $(G,H,\varphi)$ in the normal form  with fixed $V_s$ is $\mathbb{Z}_2$-approximable by a drawing $\psi_0$ then the derivative $(G',H',\varphi')$, where $G'=G$, is
$\mathbb{Z}_2$-approximable by the drawing $\psi_0'$ such that $\psi_0$ is compatible with $\psi_0'$. Moreover, if $\psi_0$ is crossing free so is $\psi_0'$.
\end{claim}

We will use the previous claim to derive that 
if $(G,H,\varphi)$ (not necessarily in the normal form) is approximable by an embedding the same holds for $(G',H',\varphi')$, which we use in the proof of 
Theorem~\ref{thm:main} to conclude that if 
$(G',H',\varphi')$ is not approximable by an embedding
the same holds for $(G,H,\varphi)$.
However, we need also the converse of this to hold, which is indeed the case.

\begin{claim}
\label{claim:integration}
If the derivative  $(G',H',\varphi')$ of $(G,H,\varphi)$ in the normal form with fixed $V_s$ is approximable by an embedding  compatible with $\psi_0'$ (given by Claim~\ref{claim:derivative}) then $(G,H,\varphi)$, where $G=G'$, is
 approximable by an embedding compatible with $\psi_0$. 
\end{claim}

We prove the claim in Section~\ref{sec:integration}.


\section{Reduction to our problem}
\label{sec:equivalence}

\jk{hodilo by se mit presne definice a nazvy jednotlivych problemu. Na prvni pohled nechapu, cim se ty dva problemy "fundamentalne" lisi - aneb nevim jak vyresit "najdi 5 rozdilu" :)}

The aim of this section is to show that the problem of deciding whether $\varphi_0:G\rightarrow M$
is approximable by an embedding is polynomially reducible to our problem of determining if $(G,H,\varphi)$ is approximable by an embedding.

First, by subdiving edges of $G$ with vertices we obtain that $G$ has no multiple edges and for every pair of edges $e$ and $g$ of $G$
either $\varphi_0(e)=\varphi_0(g)$ or the relative interiors of $\varphi_0(e)$ and $\varphi_0(g)$ are disjoint.
The number of required subdivisions is $O(|\varphi|^2)$, since every edge is subdivided $O(|\varphi|)$ times.

Second, we construct an instance $(G,H,\varphi)$, where $V(H):=\varphi_0(V(G))$, 
$E(H):=\varphi_0(E(G))$, $\varphi(v):=\varphi_0(v)$ for all $v\in V(G)$, and the isotopy class of an embedding of $H$ is inherited from $\varphi_0(G)$.

\begin{claim}
\label{claim:reductionToPipes}
The instance $(G,H,\varphi)$ is approximable by an embedding if and only if $\varphi_0$ is approximable by an embedding.
\end{claim}
\begin{proof}
The ``only if'' direction is trivial and therefore we prove 
the ``if'' direction.
Let $\psi$ denote an embedding of $G$ in $M$ that is an $\epsilon$-approximation of $\varphi_0$ for $\epsilon>0$. 
We show that $\psi$ gives rise to an embedding of $G$ in $\mathcal{H}$ satisfying (\ref{it:1st}) and~(\ref{it:2nd}).

By shrinking $\psi$ we can assume  that $\psi:G\rightarrow \mathcal{H}$ and that the property~(\ref{it:1st}) holds.
It is suficient to show that $\psi$ can be deformed slightly so that~(\ref{it:2nd}) holds, i.e., every edge of $G$ intersecting the boundary of $\mathcal{D}(\nu)$ does so in a single proper crossing.
By a small generic perturbation we achieve that $\psi(e)$, for every $e\in E(G)$, intersects the boundary of $\mathcal{D}(\nu)$, for every $\nu\in V(H)$, in finitely many proper crossings.
Let us choose $\psi$ so that the total number of crossings of edges with boundaries of $\mathcal{D}(\nu)$ is minimized.

We show that the existence of $\psi(e)$, for some $e\in E(G)$, crossing the boundary of $\mathcal{D}(\nu)$ at least twice leads to contradiction with the choice of $\psi$.
We consider a shortest piece $p_e$ of $\psi(e)$ between a pair of its crossings with $\mathcal{D}(\nu)$.
The end vertices of $p_e$ are contained in a valve $\omega$ of $\varphi(e)$.
We choose $p_e$ so that the area of the disc bounded by the curve obtained by concatenating $p_e$ with 
the part of $\omega$, that is bounded by the endpoints of $p_e$, is minimized. 
Note that the disc is contained in the pipe of $\varphi(e)$, and therefore its interior is disjoint from $\psi(G)$. Thus, $\psi(e)$ can be cut at the endpoints
of $p_e$ and the severed ends reconnected by a curve contained in the interior of $\mathcal{D}(\nu)$, see Figure~\ref{fig:lensremoval}.
In the resulting embedding, the total number of crossings of edges with boundaries of $\mathcal{D}(\nu)$ is smaller than in $\psi$ (contradiction).
\end{proof}

\begin{figure}
\centering
\includegraphics[scale=0.8]{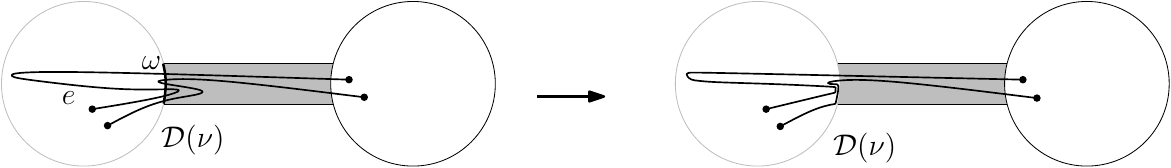}
\caption{Deforming an approximation so that every valve is crossed by an edge at most once.}
\label{fig:lensremoval}
\end{figure}


\section{Normal form}
\label{sec:normalform2}

The aim of this section is to prove the following claim.

\begin{claim}
\label{claim:normal_form}
Let $(G,H,\varphi)$ be a $\mathbb{Z}_2$-approximable instance and
let $\psi_0$ be its $\mathbb{Z}_2$-approximation.
There exist \ 
(1) an instance $(\hat{G},\hat{H},\hat{\varphi})$ in the normal form that is
$\mathbb{Z}_2$-approximable  such that $(G,H,\varphi)$ and $(\hat{G},\hat{H},\hat{\varphi})$ have the same number of pipe edges;  \ 
and (2) a $\mathbb{Z}_2$-approximation $\hat{\psi_0}$
of $(\hat{G},\hat{H},\hat{\varphi})$ such that 
 $(\hat{G},\hat{H},\hat{\varphi},\hat{\psi}_0)$ is
a {clone} of $(G,H,\varphi,\psi_0)$.
Furthermore,  we can choose $\hat{\psi}_0$ so that 
every vertex of $V_s$ (as in the definition of the normal form) is even in $\hat{\psi}_0$.
\end{claim}

%

\begin{proof}
Clearly, it is enough to prove the existence 
of a desired instance $(\hat{G},\hat{H},\hat{\varphi})$ along with $\hat{\psi}_0$ in the subdivided normal form.
We obtain $(\hat{G},\hat{H},\hat{\varphi},\hat{\psi}_0)$ from $(G,H,\varphi, \psi_0)$ by a sequence of modifications, where 
$\hat{H} = H$.
Throughout the proof we assume that multiple edges incident to an even vertex $v$ are eliminated by subdiving all but one edge participating in a multiple edge by a vertex drawn very close to $v$. This is to ensure that 
by performing flips at a vertex in an independently even drawing we do not destroy ``evenness'' of another vertex.
It will be straightforward to check that every step of the construction preserves  the number of pipe edges as required by the claim.
The proof is carried out in three phases.
During the first two phases of the proof we want to satisfy the following for every connected component $C$ of every $\varphi^{-1}[\nu]= G[V_{\nu}]$.
\begin{enumerate}[(i)]
\item \label{it:tree} $C$ is a tree; and 
\item \label{it:pruning} every leaf of $C$ is incident to at least one pipe edge.
\end{enumerate}

First, we eliminate cycles induced by clusters thereby producing a clone satisfying~(\ref{it:tree}). Second, we prune branches of the trees induced by clusters not incident to pipe edges thereby producing a clone satisfying~(\ref{it:pruning}).
Third, we further process $(G,H,\varphi)$ thereby producing a clone in the subdivided normal form.
We put $(\hat{G}_1,\hat{H}_1,\hat{\varphi}_1):=(G,H,\varphi)$ and $\hat{\psi_1}:=\psi_0$. At the $t$th step, for $t\ge 1$, 
we produce the clone $(\hat{G}_{t+1},\hat{H}_{t+1},\hat{\varphi}_{t+1},\hat{\psi}_{t+1})$ of $(\hat{G}_{t},\hat{H}_{t},\hat{\varphi}_t,\hat{\psi_t})$, where
$\hat{H}_{t+1}=H$.
By the transitivity of the clone relation, once we produce a clone in the subdivided normal form we are done.
It will be easy to see that if $\hat{\psi}_t$ is an embedding then $\hat{\psi}_{t+1}$
is also an embedding, and thus, we will use this fact tacitly to show that 
$(\hat{G}_{t+1},\hat{H}_{t+1},\hat{\varphi}_{t+1},\hat{\psi}_{t+1})$ is a clone at every step.
We proceed with a detailed description of the three previously described phases. We begin with the first two stages which are described in the following two short paragraphs.


\paragraph{Eliminating cycles inside clusters.}
We turn every connected component $C$ of $G[V_{\nu}]$ into a tree.
To this end we repeatedly apply the cycle reduction. Since the cycle reduction preserves the connectivity of $C$, by Claim~\ref{claim:cycle_red} applied to $C$ after finitely many applications of the reduction we obtain a $\mathbb{Z}_2$-approximation $\hat{\psi}_2$ of $(\hat{G}_2,\hat{H}_2,\hat{\varphi}_2)$, in which every $V_{\nu}(\hat{G}_2)$ induces a forest. By Claim~\ref{claim:cycle_red}, $(\hat{G}_2,\hat{H}_2,\hat{\varphi}_2,\hat{\psi}_2)$
is a desired clone.

\paragraph{Pruning.}
We want to achieve that every leaf in every connected component induced by a cluster is incident to at least one pipe edge.
We prune all the ``unimportant'', i.e., not incident to a pipe edge, branches of the connected components, now trees, induced by $V_{\nu}$'s in $(\hat{G}_2,\hat{H}_2,\hat{\varphi}_2)$ so that~(II) of the subdivided normal form is satisfied.
Formally, this means successively deleting leaves that are not incident to pipe edges.
Let $(\hat{G}_3,\hat{H}_3,\hat{\varphi}_3)$ be the resulting instance.
We also delete all the connected components induced by clusters with pipe degree $0$.
The pruned parts can be introduced into an approximation of $(\hat{G}_3,\hat{H}_3,\hat{\varphi}_3)$ so that we 
get a desired approximation of $(\hat{G}_2,\hat{H}_2,\hat{\varphi}_2)$.
Hence, $(\hat{G}_3,\hat{H}_3,\hat{\varphi}_3,\hat{\psi}_3)$ is a desired clone. \\

What follows is inspired by~\cite{HMcC03_PCtrees} and is more demanding than the previous two stages.
It will be almost straightforward to see that every modification carried out at this stage produces a clone,
and hence, in order not to make the proof unnecessarily long we sometimes omit formal arguments regarding that.
We enumerate connected components  induced by all clusters $V_{\nu}(\hat{G}_3)$ with pipe degree at least $2$, and process them in this order. Here, the goal is to alter each such component $C$ so that there exists a single vertex $v_C\in C$ such that
the desired set $V_s$ of special vertices is formed by the set of $v_C$'s, for every $C$, such that $\pdeg{C}\ge 2$.
The altering of the components uses operations of vertex split, contractions, flips and finger moves in $\hat{\psi}_3$. Note that the operation of contractions might introduce loops or multiple edges. In our 
case a loop cannot arise since we will be contracting only edges of  components induced by clusters, which are trees. 

The multiple edges cause minor difficulties, since flipping a pair of multiple edges $uv$ in $\hat{\psi}_3$ at $u$, where $v$ is even in $\hat{\psi}_3$, turns $v$ 
into a non-even vertex in $\hat{\psi}_3$, which is undesirable since at every step we want to produce a clone.  However, there is an easy remedy to this problem, which lies in subdividing edges incident to an even vertex $v$ by vertices that are drawn in a $\mathbb{Z}_2$-approximation very close to $v$ so that $v$ is still even after subdivisions and the multiple edges incident to $v$ are eliminated.
We tacitly assume that such an operation is performed whenever this situation occurs, and therefore,
by doing  flips  of edges with its incident vertices we will always produce a clone.
This being said we proceed to the description of the altering procedure;
and we start with a claim that restricts the way in which pipe edges can be attached to $C$.

\begin{figure}
\centering
\includegraphics[scale=0.8]{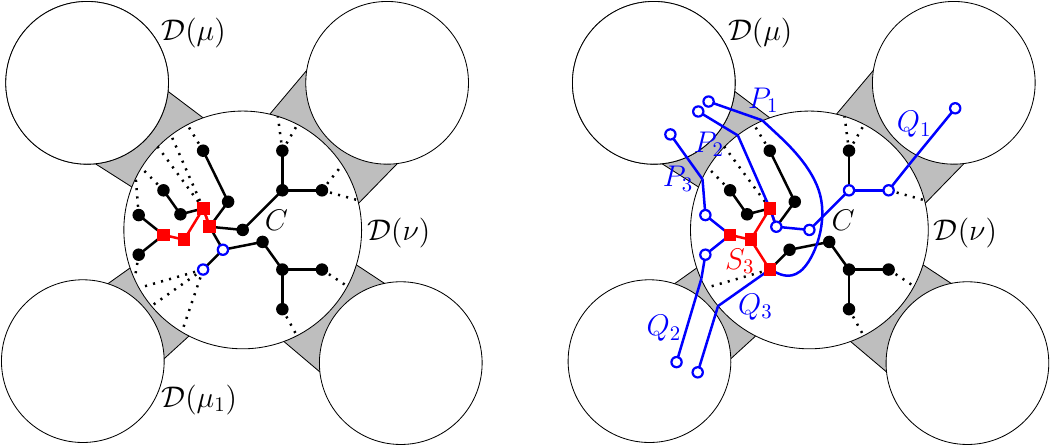}
\caption{The $\mu$-important edges colored red (squares) and
$\mu_1$-important vertices colored blue (empty disks) (left); construction of a drawing of a subdivision of $K_{3,3}$ in the proof of Claim~\ref{claim:impPath} (right). }
\label{fig:jImportant}
\end{figure}

Refer to Figure~\ref{fig:jImportant} (left).
 Recall that $\pNeigh{C}$ of $C$ denotes the pipe neighborhood of $C$ defined along with the instance in the introduction. 
Let $\nu \mu$ be an edge in $E(\hat{H}_3)$ such that at least one pipe edge incident to $C$ is mapped by $\hat{\varphi}_3$ to $\nu \mu$.
For an edge $e\in E(C)$, let $C_1(e)$ and $C_2(e)$ be the pair of trees whose union is $C\setminus e$.
An edge $e$ of $E(C)$ is $\mu$-\emphh{important} if $C_1(e)$ and $C_2(e)$  satisfy $\{\mu\}\subsetneq\pNeigh{C_1(e)}$ and $\{\mu\}\subsetneq \pNeigh{C_2(e)}$. A vertex $v$ of $V(C)$ is $\mu$-\emphh{important} if for every connected component $T$ of $C\setminus v$ we have  either $\pNeigh{T}=\{\mu\}$ or $\mu\not\in\pNeigh{T}$.

\begin{claim}
\label{claim:impPath}
Let $C$ be a connected component of $V_\nu(\hat{G}_3)$ such that $\pdeg{C}\ge 2$.
Let $\nu \mu$ be an edge in $E(\hat{H}_3)$ such that at least one pipe edge incident to $C$ is mapped by $\hat{\varphi}_3$ to $\nu \mu$.
In $C$, there exists  either
\begin{enumerate}[(1)]
\item a path $P$ with at least one edge formed by all the $\mu$-important edges of $C$; \item exactly one $\mu$-important vertex; or 
\item $\mu$-important vertices induce an edge or a subdivided edge in $\hat{G}_3$.
\end{enumerate}
\end{claim}
\begin{proof}
First, we show that $\mu$-important edges form a subtree $T$ of $C$.  Let $I\subset E(C)$ denote the set of $\mu$-important edges in $C$. Suppose for the sake of contradiction that the subgraph $C[I]$ of $C$ induced by $I$ is not connected. Hence, there exists a pair of edges  $e_1\in E(C)$ and $e_2\in E(C)$ belonging to different connected components of $C[I]$. Let $P$ be a path in $C$ between an end vertex of $e_1$ and an end vertex of $e_2$. Note that all the edges on $P$ are $\mu$-important, which is contradiction with the choice of $e_1$ and $e_2$. Indeed, let $f\in E(P)$,
and let $e_k\in E(C_k(f))$ and $f\not\in E(C_1(e_k))$, for $k=1,2$.
Then $\{\mu\}\subsetneq\pNeigh{C_1(e_1)}\subseteq \pNeigh{C_1(f)}$ and $\{\mu\}\subsetneq\pNeigh{C_1(e_2)}\subseteq \pNeigh{C_2(f)}$. Thus, $f$ is $\mu$-important.

Refer to Figure~\ref{fig:jImportant} (right).
Second, we show that $T$ is a path with at least one edge.
To this end we show that if $T$ contains a vertex $v$ of degree at least $3$, we obtain an independently even drawing of a subdivision of $K_{3,3}$ in the plane, and thus, we contradict the Hanani--Tutte theorem. The claimed drawing of $K_{3,3}$ is obtained as follows. Let $S_3$ be a 3-star consisting of $v$ and its three incident $\mu$-important edges of $C$.
Let $P_1, P_2, P_3$ be  three minimal paths  (disjoint from $v$) joining each leaf vertex of $S_3$ with a vertex in $V_{\mu}$. Let $Q_1, Q_2, Q_3$ be three minimal paths (disjoint from $v$) joining each leaf vertex of $S_3$ with a vertex in $V\setminus (V_{\nu} \cup V_{\mu})$. 
The paths $P_1,P_2,P_3$ and $Q_1,Q_2,Q_3$  exist, since $S_3$ is formed by $\mu$-important edges, that is, every leaf of $S_3$ is connected by a path in $C$, that is edge-disjoint from $S_3$, with a vertex  that is incident to a pipe edge mapped by $\hat{\varphi}_3$ to $\nu\mu$ and $\nu\mu_1$, respectively, for some $\mu_1\not=\mu$.
Let us modify $\hat{\psi}_3$ as follows. We contract the closure of the complement of the valve of $\nu\mu$ in the boundary of $\partial \mathcal{D}(\nu)$ to a point.
Similarly, we contract the shortest closed arc contained in the valve of $\nu\mu$ containing all the three crossing points with $P_1,P_2$ and $P_3$ to a point.
The desired drawing of a subdivision of $K_{3,3}$ is obtained
as the restriction of the previous modification $\hat{\psi}_3$ to the union of $S_3$ with the parts
$P_1,P_2,P_3$ and $Q_1,Q_2,Q_3$ drawn in $\mathcal{D}(\nu)$.
Here, we  regard the leaf vertices of $S_3$ to form one part of the vertex set of $K_{3,3}$; and $v$ together with the pair of points that the two disjoint parts of $\partial \mathcal{D}(\nu)$ were contracted into as the other part.

If there exists no $\mu$-important edge there must exist
at least one $\mu$-important vertex. Indeed, let us direct an
edge $e=uv$ of $C$ towards $u$, if there exists $k\in\{1,2\}$ and $\mu_1\not=\mu$ such that $u\in C_k(e)$  and $\mu,\mu_1\in \pNeigh{C_k(e)}$. (This possibly leaves some edges of $C$ undirected. Also as there is no $\mu$-important edge in $C$, every edge is directed in at most one direction.)
A vertex $u$ in $C$ with out-degree $0$, which must exist since $C$ is a tree, is $\mu$-important, for if not there would be an edge in $C$ incident to $u$ directed away from $u$ (contradiction). 

Now, if there exists a pair of $\mu$-important
vertices $u$ and $v$ then the unique path $P$ between $u$ and $v$ in $C$, has none of its edges directed. 
Indeed, let $e$ denote an edge on the path $P$ in $C$ between $u$ and $v$. Suppose for the sake of contradiction that $e$ is directed away from, let's say $u$. Then $u$ cannot have out-degree $0$, because in this case also the edge incident to $u$ on $P$ is directed away from $u$ (contradiction).

It remains to show that undirected edges in $C$ form a (subdivided) edge in $\hat{G}_3$.
Let us consider the subgraph $\overline{C}$ of $C$ consisting of undirected edges. Since~(\ref{it:pruning}) holds,  $\hat{G}_3$ cannot contain three edges $e_1, e_2$ and $e_3$ incident to the same vertex $v$
such that $e_2,e_3\in E(\overline{C})$. For if not, let w.l.o.g.  $C_1(e_1),C_1(e_2)$ and $C_1(e_3)$ be trees not containing $v$. 

W.l.o.g. we assume that $e_1$ is not a pipe edge. (If $e_1$ is a pipe edge, we can subdivide it so that this is no longer the case, which does not affect whether  $e_2$ and $e_3$ are directed or not.) First, we consider the case when  $\mu\in \pNeigh{C_1(e_1)}$, which is not empty by~(\ref{it:pruning}).  Then $\{\mu\}= \pNeigh{C_1(e_2)}$ and $\{\mu\}= \pNeigh{C_1(e_3)}$, since otherwise one of $e_2$ and $e_3$ would be directed towards $v$. Indeed, $\mu\in \pNeigh{C_1(e_1)}$, $\pNeigh{C_1(e_k)} \subseteq \pNeigh{C_2(e_l)}$, for every $k\not =l$, and both of  $\pNeigh{C_1(e_2)}$  and $\pNeigh{C_1(e_3)}$ are not empty  by~(\ref{it:pruning}).  Since $\pdeg{C}\ge 2$ and
$P(C)=\pNeigh{C_1(e_2)} \cup\pNeigh{C_2(e_2)}$, this implies that there exists 
 $\mu_1\not=\mu$ such that $\mu_1\in \pNeigh{C_2(e_2)}$. Hence, 
$e_2$ is directed towards $v$, since also $\mu\in \pNeigh{C_1(e_3)} \subseteq \pNeigh{C_2(e_2)}$ (contradiction).

If $\mu\not\in \pNeigh{C_1(e_1)}$ 
 an analogous argument, in which the roles of $\mu_1$ and $\mu$ are exchanged,  applies:  
 In this case, there exists $\mu_1\in \pNeigh{C_1(e_1)}$,  $\mu_1\not=\mu$.  Then $\mu\not\in\pNeigh{C_1(e_2)}$ and $\mu\not\in \pNeigh{C_1(e_3)}$, since otherwise one of $e_2$ and $e_3$ would be directed towards $v$ by an analogous argument as in the previous case. Since $\mu \in P(C)$, this implies that $\mu\in \pNeigh{C_2(e_2)}$. Hence, 
$e_2$ is directed towards $v$, since there exists $\mu_2\in \pNeigh{C_1(e_3)} \subseteq \pNeigh{C_2(e_2)}$, $\mu_2\not=\mu$ (contradiction).

Finally, it also follows that $\overline{C}$ is connected. Indeed, every edge on a path $P=uPv$ between two connected components of $\overline{C}$ is directed. Suppose that an edge $e$ of $P$ incident to $u$ is direced away (resp., towards) from $u$, that 
is, w.l.o.g.  $u\in C_2(e)$ (resp., $u\in C_1(e)$)  and $\mu\subsetneq C_1(e)$. It follows that every edge $e_0$ of $\overline{C}$ incident to $u$ (resp., $v$) is directed toward $u$ (resp., $v$), since w.l.o.g. $C_1(e)\subseteq C_1(e_0)$, which  is a contradiction. 
\end{proof}

We proceed by successively altering connected components $C$ in $\hat{\varphi_3}^{-1}[\nu]$, which are all trees, so that the following condition holds for every such component $C$: \\

$(\diamond)$ for every the $\mu\in P(C)$, $\mu$-important vertices in $V(C)$ induce an edge. \\

To this end we enumerate all the elements in $\pNeigh{C}$ in an arbitrary order $\mu_1,\ldots, \mu_{|\pdeg{C}|}$. We assume that $C$ was altered for all $\mu_1,\ldots, \mu_{i-1}$ for some $i\le \pdeg{C}$.

First, suppose that the case (3) from Claim~\ref{claim:impPath} applies
with $C$ and $\mu_i$ in place of $C$ and $\mu$, respectively. 
If the $\mu_i$-important vertices in $V(C)$ induce a subdivided edge in $\hat{G}_3$, we suppress the subdividing $\mu_i$-important vertices. Subsequently, we alter $C$, for $\mu_{i+1}$, if $i<\pdeg{C}$. If $i=\pdeg{C}$, we continue with processing an arbitrary unprocessed connected component or proceed to the next stage if no such component exists.

Second, suppose that the case (1) or (2) from Claim~\ref{claim:impPath} applies.
Let $P_i$ be a path formed by $\mu_i$-important edges of $C$, or a path consisting of the only $\mu_i$-important vertex $q_i$. We will process $C$ so that there will exist exactly two $\mu_i$-important vertices in $C$ joined by an edge,
or in other words so that the case (3) applies and we subsequently proceed as in the previous paragraph.
To this end we apply flips and vertex splits to vertices of $P_i$ or $q_i$ which will render the vertices of (the path corresponding to) $P_i$ or $q_i$, respectively, even in $\hat{\psi}_3$. Given that this is the case, if~(1) of Claim~\ref{claim:impPath} occurs  we contract $P_i$
into an even vertex $q_i$, which leave us to deal with~(2) of Claim~\ref{claim:impPath}. This is possible since the new vertex that we obtain by contracting $e$ is not incident to any loop ($C$ is a tree), and the new vertex is even in the resulting drawing. Thus, by contracting $e$ we necessarily produce a clone.

Indeed, recalling the definition of the clone we need to verify two implications which can be carried out as follows. First, in a compatible embedding of the graph we get from $\hat{G}_3$ after contracting $e$, the vertex corresponding to $e$ can be split in order to obtain embedding of $\hat{G}_3$ compatible with $\hat{\psi}_3$. Second, contracting an edge in any embedding will again produce an embedding. Formally, we have the following.

\begin{claim}
\label{claim:contraction}
Let $(G,H,\varphi,\psi_0)$ be such that $(G,H,\varphi)$ is an instance and $\psi_0$ is its $\mathbb{Z}_2$-approximation.
Let  $e$ be an edge $E(G)$
such that $\varphi(e)=\nu$ and $\varphi^{-1}[\nu]$ is a forest, where $\nu\in~V(H)$, and such that both end vertices of $e$ are even in $\psi_0$.
Let $(\hat{G},H,\hat{\varphi},\hat{\psi}_0)$ be obtained from $(G,H,\varphi,\psi_0)$ by 
contracting $e$. Then $(\hat{G},H,\hat{\varphi},\hat{\psi}_0)$ is a clone of 
$(G,H,\varphi,\psi_0)$.
\end{claim}

In the subsequent step, we  apply to $q_i$ the vertex $A_{q_i}$-split,
where $A_{q_i}$ contains the neighbors of $q_i$ mapped by $\hat\varphi_3$ to $\mu_i$, and the neighbors in the connected components of $C\setminus q_i$, whose incident pipe edges are all mapped to $\nu\mu_i$. (Since $q_i$ is $\mu_i$-important each of the  remaining neighbors of $q_i$ is  either mapped by   $\hat\varphi_3$ to  $\mu_j$, for some $\mu_j\not=\mu_i$, or belongs to a connected component of  $C\setminus q_i$, none of whose incident pipe edges is mapped to $\nu\mu_i$.)
Thereby we turn $q_i$ into a pair of $\mu_i$-important vertices joined by an edge.
Before accomplishing this goal we present a pair of auxiliary facts discovered in~\cite[Section 9.1.2]{F14+_towards},  both of which  play an important role in the rest of this section.


Let $e_1,e_2,e_3$ and $e_4$ be four distinct edges incident to $v\in C$. 
We assume that in the drawing of $\hat{\psi}_3$ the edges $e_1$ and $e_2$ do not alternate with $e_3$ and $e_4$ in the rotation at $v$. 
If $e_1\in E(C)$, let $C_1$ be the connected component of $C\setminus e_1$ containing the end vertex of $e_1$ different from $v$. If $e_1\not\in E(C)$, $C_1$ is a subgraph of $\hat{G}_3$ consisting of $e_1$.
We define $C_2,C_3$ and $C_4$ analogously.
For $l=1,2,3,4$, let  $\mu_{i_l}\in \pNeigh{C_{l}}$,  if $C_l$ is induced by $V_{\nu}$;
and let  $\mu_{i_l}$ be such that $\hat{\varphi}_3(C_{i_l})=\nu\mu_{i_l}$ otherwise.
The relation $\equiv_2$ stands for the equality $\mod 2$.
See Figure~\ref{fig:basiclemma} for an illustration.

\begin{figure}
\centering
\includegraphics[scale=0.8]{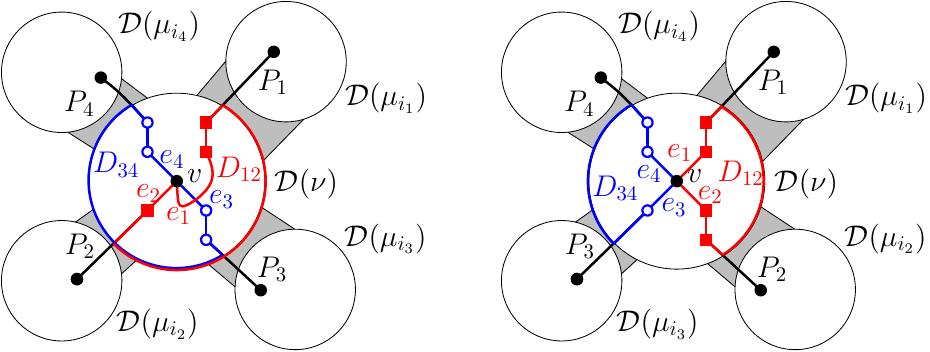}
\caption{Illustration of Claim~\ref{claim:basic} and its proof.}
\label{fig:basiclemma}
\end{figure}

\begin{claim}
\label{claim:basic}
The edges $\nu\mu_{i_{1}}$ and $\nu\mu_{i_{2}}$ do not alternate
with $\nu\mu_{i_{3}}$ and $\nu\mu_{i_{4}}$, where $\mu_{i_1},\mu_{i_2}\not=\mu_{i_3},\mu_{i_4}$, in the rotation at $\nu$, if and only if
$ \nCross{e_1}{e_3}{\hat{\psi}_3}+\nCross{e_1}{e_4}{\hat{\psi}_3}+\nCross{e_2}{e_3}{\hat{\psi}_3}+\nCross{e_2}{e_4}{\hat{\psi}_3}\equiv_2 0$.
In particular, if $i_1=i_2$ or $i_3=i_4$, we have
$ \nCross{e_1}{e_3}{\hat{\psi}_3}+\nCross{e_1}{e_4}{\hat{\psi}_3}+\nCross{e_2}{e_3}{\hat{\psi}_3}+\nCross{e_2}{e_4}{\hat{\psi}_3}\equiv_2 0$.
\end{claim}
\begin{proof}
Let $P_1, P_2, P_3$ and $P_4$ be paths in $\hat{G}_3$ of minimal length starting with $v$ then following $e_1,e_2,e_3$ and $e_4$, and ending by a pipe edge mapped by $\hat{\varphi}_3$ to $\nu\mu_{i_{1}},\nu\mu_{i_{2}},\nu\mu_{i_{3}}$ and $\nu\mu_{i_{4}}$, respectively.
Let $D_{12}$ be a closed curve in the plane obtained by concatenating parts of $P_1$ and $P_2$  
in the drawing of $\hat{\psi}_3$ inside $\mathcal{D}(\nu)$ and the part of the boundary of $\mathcal{D}(\nu)$ not crossing $P_3$. Let $D_{34}$ be a closed curve in the plane obtained by concatenating parts of $P_3$ and $P_4$ in the drawing of $\hat{\psi}_3$ inside $\mathcal{D}(\nu)$ and the part of the boundary of $\mathcal{D}(\nu)$ not crossing $P_1$.

We have $\sum_{e,e'}\nCross{e}{e'}{\hat{\psi}_3}\equiv_2 0$, where we sum over edges $e$
fully or partially contained in $D_{12}$ and $e'$ in $D_{34}$, if and only if the edges $\nu\mu_{i_{1}}$ and $\nu\mu_{i_{2}}$ do not alternate
with $\nu\mu_{i_{3}}$ and $\nu\mu_{i_{4}}$. Since $\hat{\psi}_3$ is independently even
we have $ \sum_{e,e'}\nCross{e}{e'}{\hat{\psi}_3}\equiv_2 \nCross{e_1}{e_3}{\hat{\psi}_3}+\nCross{e_1}{e_4}{\hat{\psi}_3}+\nCross{e_2}{e_3}{\hat{\psi}_3}+\nCross{e_2}{e_4}{\hat{\psi}_3}$, and the claim follows.
\end{proof}

Refer to Figure~\ref{fig:rotation}.
Let $v_i$ be a vertex on $P_i$ which is formed by $\mu_i$-important edges.
Let $E(v_i)$ be a nonempty set of edges incident to $v_i$ consisting of  pipe edges mapped by $\hat{\varphi}_3$ to $\nu\mu_i$, and  edges joining $v_i$ with trees $T$ of $C\setminus v_i$ such that $\mu_i\in \pNeigh{T}$.
Note that $E(v_i)$ is not uniquely defined due to the edges  incident to $v_i$ that lie on $P_i$ for which we can choose whether we put them in $E(v_i)$.
 
 Let $E'(v_i)$ be the nonempty set of edges incident to $v_i$
 and not in $E(v_i)$. Thus,  $E(v_i)$ and $E'(v_i)$ form a non-trivial partition of edges incident to $v_i$, since $P_i$ is formed by $\mu_i$-important edges.
Furthermore, due to~\eqref{it:pruning} every edge in $E'(v_i)$ is either a pipe edge mapped by $\hat{\varphi}_3$ to $\nu\mu_j$ for some $\mu_j\not=\mu_i$, or joins $v_i$ with a tree $T$ of $C\setminus v_i$  such that $\mu_j\in \pNeigh{T}$, where $\mu_j\not=\mu_i$.

\begin{claim}
\label{claim:split}
There exists a subset of the set of edges incident to $v_i$
such that by pulling the edges in the subset over $v_i$ in $\hat{\psi}_3$
results in a drawing in which every edge in $E(v_i)$ crosses every edge in $E'(v_i)$ an even number of times.
\end{claim}

Claim~\ref{claim:split} (proved below)  will allow us to apply the following claim.

\begin{claim}
\label{claim:splitting}
Let $(G,H,\varphi,\psi_0)$ be such that $(G,H,\varphi)$ is an instance and $\psi_0$ is its $\mathbb{Z}_2$-approximation.
Let $v=v_i$ be a vertex of $V(G)$ on $P_i$.
Let $A=E(v)$ be the subset of edges incident to $v$ appearing consecutively in the rotation at $v$ and let  $B$ the set of remaining edges incident to $v$. Suppose that every edge in $A$ crosses every edge in $B$ an even number of times in $\psi_0$. 
Let $(\hat{G},H,\hat{\varphi},\hat{\psi}_0)$ be obtained from $(G,H,\varphi,\psi_0)$ by performing a vertex $A$-split on $v$ in $\psi_0$ such that for the resulting edge $e$ we have $\hat{\varphi}(e)={\varphi}(v)$.
Then $(\hat{G},H,\hat{\varphi},\hat{\psi}_0)$ is a clone of 
$(G,H,\varphi,\psi_0)$.
\end{claim}
\begin{proof}
Similarly as for Claim~\ref{claim:contraction}, we need to verify two implications of the clone definition. First, in an  embedding of $\hat{G}$ compatible with $\hat{\psi}_0$ contracting $e$ yields an embedding of $G$ compatible with $\hat{\psi}$. Second, performing $A$-split on $v$ in an embedding of $G$ approximating $\varphi$ results in an embedding of $\hat{G}$ approximating $\hat{\varphi}$, since edges of $A$ and $B$ do not interleave in the rotation at $v$. Indeed, due to the definition of $A=E(v)$ such an interleaving would necessarily lead to an edges crossing (contradiction).
\end{proof}

\begin{proof}[proof of Claim~\ref{claim:split}]
We correct the rotation at $v_i$ by flips so that the edges in $E(v_i)$ appear consecutively in the rotation at $v_i$.
Since we assume that the edges of $E(v_i)$ appear consecutively in the rotation at $v_i$, the same holds for $E'(v_i)$.
The claim is easy if the size of $E(v_i)$ is one, since then we just pull over $v_i$ the edges in $E'(v_i)$ incident to $v_i$ and crossing the edge in $E(v_i)$ an odd number of times.
Similarly we are done if $E'(v_i)$ contains only one edge.
Let $e_1,e_2\in E(v_i)$ and let $e_3,g_4\in E'(v_i)$ be four distinct edges incident to $v_i$.

By Claim~\ref{claim:basic}, we have
$ \nCross{e_1}{e_3}{\hat{\psi}_3}+\nCross{e_1}{e_4}{\hat{\psi}_3}+\nCross{e_2}{e_3}{\hat{\psi}_3}+\nCross{e_2}{e_4}{\hat{\psi}_3}\equiv_2 0$.
Thus, for every pair of edges $e_{1}$ and $e_{2}$ in $E(v_i)$ either $\nCross{e_1}{e_3}{\hat{\psi}_3}\equiv_2 \nCross{e_2}{e_3}{\hat{\psi}_3}$ for all $e_3\in E'(v_i)$,
or $\nCross{e_1}{e_3}{\hat{\psi}_3}\not\equiv_2 \nCross{e_2}{e_3}{\hat{\psi}_3}$ for all $e_3\in E'(v_i)$. Indeed, for $e_3,e_4\in E'(v_i)$ violating the claim we have
$\nCross{e_1}{e_3}{\hat{\psi}_3}+\nCross{e_2}{e_3}{\hat{\psi}_3} \not \equiv_2 \nCross{e_1}{e_4}{\hat{\psi}_3}+ \nCross{e_2}{e_4}{\hat{\psi}_3}.$ Hence,
$\nCross{e_1}{e_3}{\hat{\psi}_3}+\nCross{e_2}{e_3}{\hat{\psi}_3} +\nCross{e_1}{e_4}{\hat{\psi}_3}+ \nCross{e_2}{e_4}{\hat{\psi}_3}\not\equiv_2 0.$ 
It follows that we can partition $E(v_i)$ into two parts 
$E_1(v_i)$ and $E_2(v_i)$ such that for every $e_1$ and $e_2$ coming from the same part we have $\nCross{e_1}{e_3}{\hat{\psi}_3}\equiv_2 \nCross{e_2}{e_3}{\hat{\psi}_3}$ for all $e_3\in E'(v_i)$, and for every $e_1$ and $e_2$ coming from different parts we have $\nCross{e_1}{e_3}{\hat{\psi}_3}\not \equiv_2 \nCross{e_2}{e_3}{\hat{\psi}_3}$ for all $e_3\in E'(v_i)$. 

By pulling every edge of $E_1(v_i)$ over $v_i$ we obtain $\nCross{e_{1}}{e_3}{\hat{\psi}_3}\equiv_2 \nCross{e_{2}}{e_3}{\hat{\psi}_3}$ for all $e_3\in E'(v_i)$
and $e_{1},e_{2}\in E(v_i)$.
Thus, we obtained a drawing in which for all $e_3\in E'(v_i)$ either $\nCross{e_{1}}{e_3}{\hat{\psi}_3} \equiv_2 1$ for all $e_1\in E(v_i)$ or $\nCross{e_{1}}{e_3}{\hat{\psi}_3} \equiv_2 0$ for all $e_1\in E(v_i)$.
By pulling every $e_3\in E'(v_i)$, for which $\nCross{e_{1}}{e_3}{\hat{\psi}_3} \equiv_2 1$ for all $e_1\in E(v_i)$, over $v_i$ in the obtained drawing, we arrive at a desired drawing of $\hat{G}_3$ and that concludes the proof.
\end{proof}

\begin{remark}
We present an alternative proof strategy for the proof of the previous lemma whose successful execution might lead to a significant simplification of the proof of Claim~\ref{claim:normal_form}. Let $\mathcal{P}$ denote a set of $|E(v_i)|$ minimal paths emanating from $v_i$, whose respective first edges are the edges in $E(v_i)$ and the last end vertices are mapped by $\hat{\varphi}_3$ to $\mu_i$. Analogously, let $\mathcal{P}'$ denote a set of $|E'(v_i)|$ minimal paths emanating from $v_i$, whose first edges are the edges in $E'(v_i)$ and the last end vertices are  mapped by $\hat{\varphi}_3$ neither to $\nu$ nor $\mu_i$.
The existence of $\mathcal{P}$ and $\mathcal{P}'$ is guaranteed by the  properties of $E(v_i)$ and $E'(v_i)$, respectively.

 It is perhaps possible to prove the previous lemma by considering a generic continuous deformation argument taking place inside $\mathcal{D}(\nu)$, during which we eliminate all the crossings between all pairs of paths $P$ and $P'$, where $P\in \mathcal{P}$ and $P'\in \mathcal{P}'$. The resulting drawing does not have to be independently even, but it can be turned into an independently even drawing by performing appropriate finger moves taking place inside $\mathcal{D}(\nu)$. However, performing  finger moves  can destroy the desired property that every edge in $E(v_i)$ crosses every edge in $E'(v_i)$ an even number of times. Therefore an additional argument is required to show that we can make the drawing independently even while still maintaining this property. To make this work does not seem to be straightforward, even if we use the fact that $C$ is a tree.
\end{remark}

When using Claim~\ref{claim:basic} and Claim~\ref{claim:split} we tacitly assume the fact that~(\ref{it:tree}) and~(\ref{it:pruning}) hold which makes the claims applicable.
In the following paragraph we process the path $P_i$ consisting of $\mu_i$-important edges or a single $\mu_i$-important vertex so that $P_i$ can be contracted to an even vertex $q_i$ as desired which is then subsequently split into an edge. To this end we first make all the vertices on $P_i$ even.

\paragraph{Processing $P_i$.}
Suppose that $v_i$ is a vertex of $P_i$, where $P_i$ can be also a trivial path consisting of a single vertex.

\begin{figure}
\centering
\subfloat[]{
\includegraphics[scale=0.8]{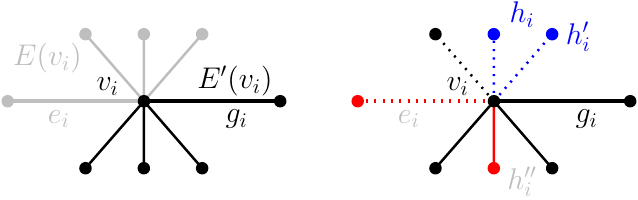}
\label{fig:rotation}}
\subfloat[]{
\hspace{2pt}
\includegraphics[scale=0.8]{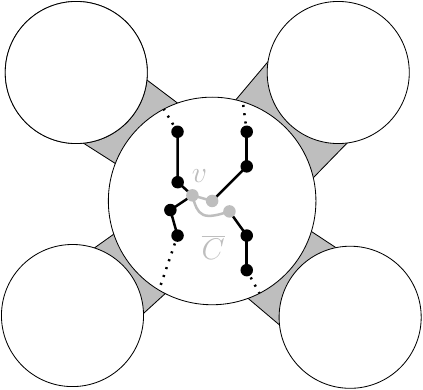}
\label{fig:rotation2}}
\caption{(a) The edges incident to $v_i$. On the left, the edges in $E(v_i)$ are colored grey
and the edges in $E'(v_i)$ are colored black. On the right, every dotted edge crosses every non-dotted edge an even number of times. (b) Every four edges incident to the vertex $v$ in $\overline{C}$ can be made cross each other an even number of times by flips.}
\end{figure}

Refer to Figure~\ref{fig:rotation} (left).
Let $E_0(v_i)$ be the set of all edges $e$ incident to $v_i$, such that $e$ is either a pipe edge mapped by $\hat{\varphi}_3$ to $\nu\mu_i$, or joins $v_i$ with a tree of $C\setminus v_i$ incident to pipe edge(s) mapped by $\hat{\varphi}_3$ to $\nu\mu_i$.
If $P_i$ is a trivial path  or $v_i$ is an end vertex of $P_i$ let $E(v_i):=E_0(v_i)$. Otherwise, let $E(v_i):=E_0(v_i)\setminus g_i$,
where $g_i\in E(P_i)$ is incident to $v_i$ and chosen arbitrarily.
 Let $E'(v_i)$ be the set of remaining edges
incident to $v_i$. The set $E'(v_i)$ is non-empty since 
the pipe degree of $C$ is at least 2 and $P_i$ is formed by $\mu_i$-important edges.
If $E(P_i)\cap E(v_i)\not=\emptyset$, let $e_i$ be the edge of $P_i$ in $E(v_i)$.
Note that if $E(P_i)\cap E'(v_i)\not=\emptyset$, the edge  $g_i \in E(P_i)$ is in $E'(v_i)$. Also note that if $g_i$ exists also $e_i$ exists.
We correct the rotation at $v_i$ by flips so that  \\

 (*)  all the edges in $E(v_i)\setminus e_i$ appear consecutively in the rotation at $v_i$ (with $e_i$ included if $e_i$ exists), and all the edges in $E'(v_i)\setminus g_i$ appear consecutively in the rotation at $v_i$ (with $g_i$ included if $g_i$ exists).\\

We achieve by flips of edges incident to $v_i$ that \\

 $(\circ)$ every edge in $E(v_i)$ crosses every edge in $E'(v_i)$ an even number of times and that every edge incident to $v_i$ crosses every edge of $P_i$ incident to $v_i$ an even number of times, while maintaining~(*).   \\

Suppose that $(\circ)$ holds, we show that then we can apply the vertex splits to $v_i$ thereby replacing it with an even vertex on $P_i$, which will produce a clone by Claim~\ref{claim:splitting}.
Due to~(*), the following operation is possible and will lead to an independently even drawing. If $E(v_i)\setminus e_i\not=\emptyset$, we apply the vertex $(E(v_i)\setminus e_i)$-split to $v_i$. In $\hat{\psi}_3$, we obtain a pair of vertices $u_i$ and $w_i$ that replace $v_i$ and that are joined by a short crossing-free edge $e_i'$ such that $u_i$ is incident to all the edges corresponding to the edges of $E(v_i)\setminus e_i$ and $w_i$ is incident to all the remaining edges. Thus, $w_i$ replaces $v_i$ on $P_i$. Analogously, if $E'(v_i)\setminus g_i\not=\emptyset$, we apply a vertex $(E'(v_i)\setminus g_i)$-split to $w_i$ or $v_i$ (depending on whether $v_i$ was split before) thereby obtaining $y_i$ and $z_i$ joined by a short crossing-free edge such that $z_i$ is incident to all the edges corresponding to $E'(v_i)\setminus g_i$ and $z_i$ is incident to the remaining edges. Thus, $z_i$ replaces $w_i$ or $v_i$ on $P_i$. Note that the vertex $z_i$ is even as required. Moreover, note that all the edges on $P_i$ are still $\mu_i$-important. Therefore it does not matter if we first achieve $(\circ)$ for all the vertices on $P_i$ or perform the vertex split(s) as soon as $(\circ)$  holds for a particular vertex. By Claim~\ref{claim:splitting}, splitting  $v_i$ produces a clone.

We show how to achieve with flips that $(\circ)$ holds.
 By Claim~\ref{claim:split} applied to $E(v_i)$, we obtain that every edge in $E(v_i)$ crosses every edge in $E'(v_i)$ an even number of times. Note that if none of $e_i$ and $g_i$ exists, we proved $(\circ)$, and that $e_i$ and $g_i$ cross each other evenly, if they both exist. 
 Thus, we assume that $e_i$ exists. If $\nCross{h_i}{e_i}{\hat{\psi}_3}\equiv_2 1$, for every $h_i\in E(v_i)\setminus e_i$, we flip every edge $E(v_i)\setminus e_i$ with $e_i$, while maintaining~(*) in the end.
A symmetric argument applies to $E'(v_i)$ and $g_i$.
Refer to Figure~\ref{fig:rotation} (right)
 Suppose that there exist at least two edges $h_i$ and $h_i'$ in $E(v_i)\setminus e_i$, such that the parity of crossings between $h_i$ and $e_i$ is not the same as between $h_i'$ and $e_i$. We presently show that this leads to $E'(v_i)=\emptyset$. Indeed, by applying Claim~\ref{claim:basic} to the four-tuple consisting of $h_i,h_i'$, $e_i$ and an edge $h_i''$ in $E'(v_i)$ (playing roles of $e_1,e_2,e_3$ and $e_4$ in this order) we have 
  $\nCross{h_i}{e_i}{\hat{\psi}_3}+\nCross{h_i'}{e_i}{\hat{\psi}_3}+\nCross{h_i}{h_i''}{\hat{\psi}_3}+
  \nCross{h_i'}{h_i''}{\hat{\psi}_3}\equiv_2 0$, but $\nCross{h_i}{h_i''}{\hat{\psi}_3}\equiv_2
  \nCross{h_i'}{h_i''}{\hat{\psi}_3}\equiv_2 0$. Hence, in the case when every $h_i\in E(v_i)\setminus e_i$ does not already intersect $e_i$ an even number of times, either we flip with $e_i$ every $h_i\in E(v_i)$ crossing $e_i$ an odd number of times, when
 $E'(v_i)=\emptyset$, or every $h_i\in E(v_i)\setminus e_i$ intersects $e_i$ an odd number times.
The latter case was already taken care of. A symmetric argument  applies to $E'(v_i)$ and $g_i$, which concludes the  proof of  $(\circ)$.
Hence, we can assume that all the vertices of $P_i$ are even, and we contract $P_i$ to an even vertex $q_i$. 
By Claim~\ref{claim:contraction}, contracting $P_i$ results in a  clone, see Figure~\ref{fig:norm1}.

%
%
%
%
%

 \begin{figure}
\centering
\subfloat[]{
\includegraphics[scale=0.7]{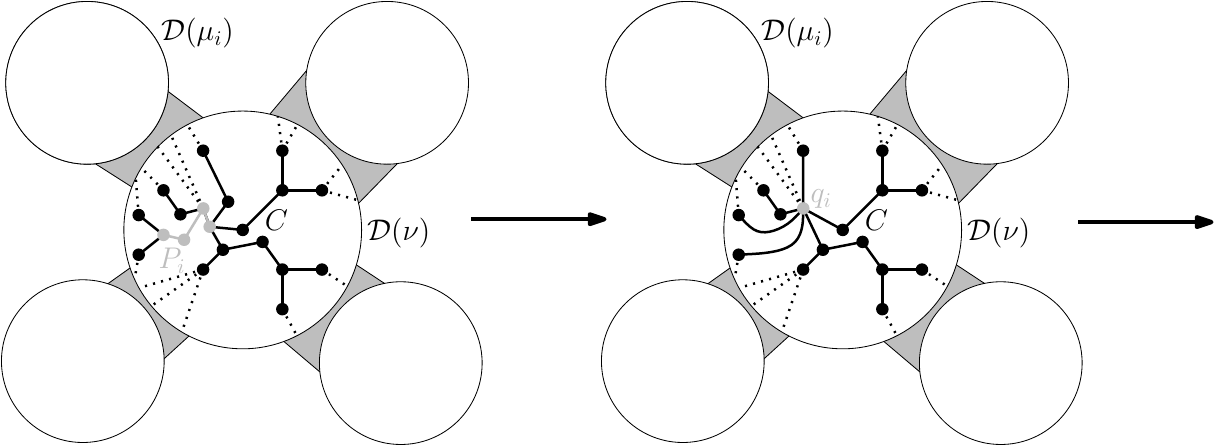}
\label{fig:norm1}}
\hspace{100pt}
\subfloat[]{\includegraphics[scale=0.7]{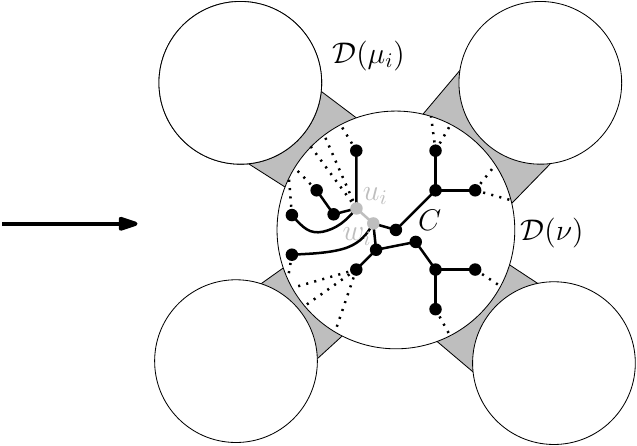}
\label{fig:norm2}}
\hspace{2pt}
\subfloat[]{
\includegraphics[scale=0.7]{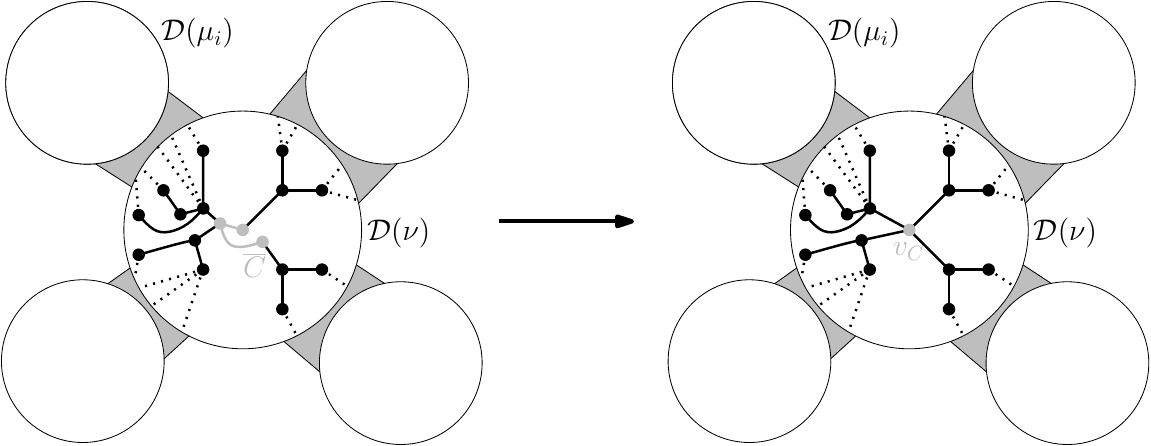}
\label{fig:norm3}}
\caption{(a) Path $P_i$ formed by $\mu_i$-important edges (left); contraction of $P_i$ to $q_i$ (right). (b) Splitting $q_i$ into $u_iw_i$. (c) After contracting $P_i$ composed of $\mu_i$-important edges. The subtree $\overline{C}$ that can be turned into an even subgraph by flips and subsequently contracted into a vertex $v_C$.}
\end{figure}

We are finally in a position to establish property~$(\diamond)$ for $i$. In the following, we first show that splitting $q_i$ is indeed possible and then  carefully argue that by splitting it we do not violate~$(\diamond)$ for $j<i$.

 Refer to Figure~\ref{fig:norm2} and~Figure~\ref{fig:twist}.
Similarly as above let $E(q_i)$ be the set of all edges incident to $q_i$ that are either pipe edges mapped by $\hat{\varphi}_3$ to $\nu\mu_i$, or are joining $q_i$ with trees of $C\setminus q_i$ incident to pipe edge(s) mapped by $\hat{\varphi}_3$ to $\nu\mu_i$.
By Claim~\ref{claim:basic}, the edges in $E(q_i)$, appear consecutively in the rotation at $q_i$.
Indeed, a four-tuple of edges violating our previous claim, is a four-tuple violating 
Claim~\ref{claim:basic}.
Hence, by applying a vertex split we can turn the vertex $q_i$ into 
vertices $u_i$ and $w_i$ joined by a short crossing-free edge $e_i$ (overriding previously defined $e_i$) such that $u_i$ is incident to all the edges corresponding to the edges of $E(q_i)$ as indicated in the figure and $w_i$ is incident to all the remaining edges. Thereby we established property~$(\diamond)$ for $i$.

After performing operations described in previous paragraphs successively in $C$ for every $j\in J$,
and doing likewise for every connected component induced by a cluster of $(\hat{G}_3,\hat{H}_3,\hat{\varphi}_3)$,
we obtain $(\hat{G}_4,\hat{H}_4,\hat{\varphi}_4)$. 
Note that edges obtained by vertex splits inside $V_{\nu}$'s can be contracted in an approximation of $(\hat{G}_4,\hat{H}_4,\hat{\varphi}_4)$, and hence, we can safely perform them. Also the path $P_i$ is disjoint from all the previously obtained $u_{j}$'s, $j<i$,
due to the following. None of the edges incident to such $u_{j}$
is $\mu_i$-important, and the edge corresponding to $u_{j}w_{j}$, i.e., the edge obtained from $u_{j}w_{j}$ by successive  contractions  of paths through $w_j$,
is either undirected or directed away from $u_j$ if we direct it 
 towards a vertex of $C$ incident to an edge mapped by $\hat{\varphi}_3$ to $\nu\mu_i$ (as in the proof of Claim~\ref{claim:impPath}).
Note that the edge $u_{j}w_{j}$ is undirected only if $\pdeg{C}=2$ in which case property $(\diamond)$ holds for $i$, which must be equal to 2,
already  after processing $P_j$, where $j=1$. \\

\begin{figure}
\centering
\includegraphics[scale=0.7]{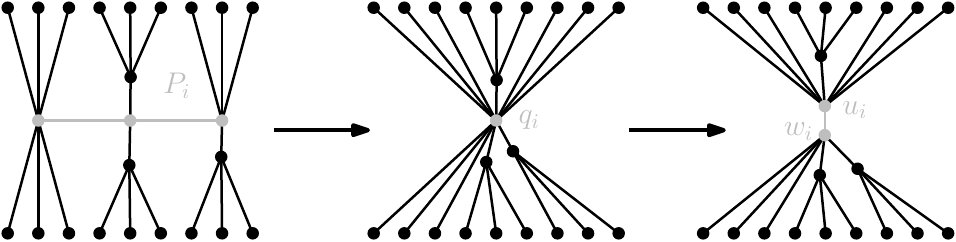}
\caption{Contracting $P_i$ into $q_i$ and splitting $q_i$ into $u_iw_i$.}
\label{fig:twist}
\end{figure}

Refer to Figure~\ref{fig:norm3}.
Let $e_i$ be as defined in the previous paragraph or an edge corresponding to such $e_i$ in the sense as defined in Section~\ref{sec:tools} in the context of contractions and splits. Thus, $e_i$ is an edge whose both end vertices are $\mu_i$-important.
After processing all the paths consisting of $\mu_i$-important edges in $C$ we obtain that   case (3) of Lemma~\ref{claim:impPath} applies to every $\mu_i\in \pNeigh{C}$.
Suppose that $\pdeg{C}=2$. We subdivide the edge $u_1w_1$ connecting the pair of important vertices in $C$ by a vertex $v_C$ drawn very close to $u_1$ so that the resulting modification of $\hat{\psi}_4$ is still independently even. Otherwise, 
let $\overline{C}$ be a connected component of $C\setminus \{e_i|  \ \mu_i\in \pNeigh{C} \}$, where $C$ is induced by $V_{\nu}[\hat{G}_4]$, such that $\overline{C}$ is incident to all $e_i$'s. The fact that such a component necessarily exists is implied by the folklore result claiming that trees have 2-Helly property\footnote{A family of sets $\mathcal{S}$ has 2-Helly property if for every subset  
$\mathcal{S}'\subseteq \mathcal{S}$ holds the following. The intersection $\bigcap_{S\in \mathcal{S}'}S\not=\emptyset$, if for every $S_1,S_2\in \mathcal{S}'$ it holds that $S_1 \cap S_2\not=\emptyset$.}
as follows.
We consider the set of trees $C_1(e_i)$, for $i=1,\ldots, \pdeg{C}$, not containing $u_i$.
Since every pair of trees $C_1(e_i)$ and $C_1(e_j)$, for $i\not=j$, intersect (in $\hat{G}_4$),
their common intersection $\overline{C}$ is also non-empty. Furthermore, 
$C_2(e_i)$'s containing $u_i$'s are pairwise disjoint, and therefore
 $E(\overline{C})$, $\{e_1,\ldots, e_{\pdeg{C}}\}$, $E(C_2(e_1))$,\ldots,  $E(C_2(e_{\pdeg{C}}))$, form a partition of the edge set of $C$.

Next, we process $\overline{C}$ so that it consists only of a single even vertex. 

\paragraph{Processing $\overline{C}$.}
\begin{itemize}
\item {\bf Single-vertex case.}
If $\overline{C}$ is a single vertex $v_C$, this vertex must be $\mu_i$-important for every $\mu_i$ such that an edge incident to $C$ is mapped to $\nu\mu_i$.
If $deg(v_C)\le 3$ we can make $v_C$ even by flips
in $\hat{\psi}_4$.
Otherwise, we apply Claim~\ref{claim:basic} and Claim~\ref{claim:4} to every 4-tuple of edges (and 4-tuples of corresponding neighbors of $\nu$)  incident to $v_C$ as follows.
The edges incident to $v_C$ are precisely $e_1,\ldots, e_{\pdeg{C}
}$. W.l.o.g we assume that $\nu\mu_1,\nu\mu_2,\nu\mu_3$ and $\nu\mu_4$ appear in this order in the rotation at $\nu$. We correct the rotation at $v_C$ by flips so that $e_1$ and $e_2$ do not alternate with $e_3$ and $e_4$; and so that $e_1$ and $e_2$ cross an even number of times and the same holds for $e_3$ and $e_4$. By Claim~\ref{claim:basic}, $ \nCross{e_1}{e_3}{\hat{\psi}_4}+\nCross{e_1}{e_4}{\hat{\psi}_4}+\nCross{e_2}{e_3}{\hat{\psi}_4}+\nCross{e_2}{e_4}{\hat{\psi}_4}\equiv_2 0$. 
If $v_C$ is not even, up to symmetry, there are four cases to consider:
\begin{itemize}
\item
$ \nCross{e_1}{e_3}{\hat{\psi}_4}\equiv_2\nCross{e_1}{e_4}{\hat{\psi}_4}\equiv_2  1 \ \mathrm{and} \ \nCross{e_2}{e_3}{\hat{\psi}_4}\equiv_2\nCross{e_2}{e_4}{\hat{\psi}_4}\equiv_2 0$
\item
$ \nCross{e_1}{e_3}{\hat{\psi}_4}\equiv_2\nCross{e_2}{e_4}{\hat{\psi}_4}\equiv_2  1 \ \mathrm{and} \ \nCross{e_2}{e_3}{\hat{\psi}_4}\equiv_2\nCross{e_1}{e_4}{\hat{\psi}_4}\equiv_2 0$
\item
$ \nCross{e_1}{e_3}{\hat{\psi}_4}\equiv_2\nCross{e_2}{e_4}{\hat{\psi}_4}\equiv_2  0 \ \mathrm{and} \ \nCross{e_2}{e_3}{\hat{\psi}_4}\equiv_2\nCross{e_1}{e_4}{\hat{\psi}_4}\equiv_2 1$
\item
$ \nCross{e_1}{e_3}{\hat{\psi}_4}\equiv_2\nCross{e_2}{e_4}{\hat{\psi}_4}\equiv_2  \nCross{e_2}{e_3}{\hat{\psi}_4}\equiv_2\nCross{e_1}{e_4}{\hat{\psi}_4}\equiv_2 1$
\end{itemize}

\begin{figure}
 \centering
\includegraphics[scale=0.7]{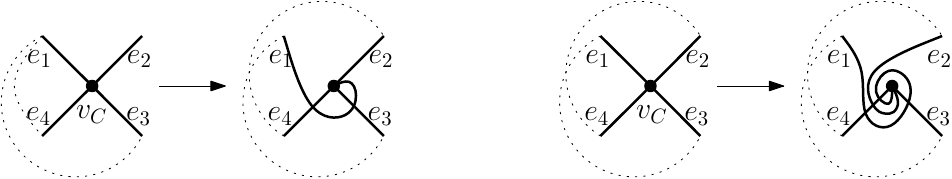}
\caption{Correcting the rotations at $v_C$ by flips so that $v_C$ becomes an even vertex. The dotted arcs join pairs of end pieces of edges that cross an odd number of times.}
\label{fig:local}
\end{figure}

We observe that in every case the rotation at $v_C$ can be corrected by flips so that $v_C$ becomes even, and thus, Claim~\ref{claim:4} implies that $v_C$ can be made even by flips; see Figure~\ref{fig:local} illustrating the redrawing in the first two cases. Indeed, the above argument applies to any 4-tuple of edges incident to $v_C$.

\item {\bf $\overline{C}$ contains at least two vertices.}
We similarly show that we  can apply flips in $\hat{\psi}_4$ to the adjacent edges in $\overline{C}$ so that every edge of $\overline{C}$ crosses every other edge of $\hat{G}_4$ an even number of times.
 To this end let $e\in E(\overline{C})$. Let ${C_1}(e)$ and ${C_2}(e)$ be the pair of trees obtained from ${C}$ after removing $e$.
Suppose that $\mu_1,\mu_2\in \pNeigh{C_1(e)}$ and 
$\mu_3,\mu_4\in \pNeigh{C_2(e)}$, where $\mu_1\not= \mu_2$ and $\mu_3\not= \mu_4$. 

\begin{claim}
\label{claim:rotation}
In the rotation at $\nu$ in $\hat{H}_4=H$ the edges $\nu\mu_1,\nu\mu_2$ do not alternate with
$\nu\mu_3,\nu\mu_4$.
\end{claim}
\begin{proof}
We contract $e$ in $\hat{\psi}_4$. Let $v$ be the vertex that $e$ was contracted into. Let $e_1,e_2,e_3$ and $e_4$ (overriding previous $e_j$'s) denote the four edges incident  to $v$ such that $e_j$, for $j=1,2,3,4$, joins $v$ with a connected component $C_j$ of $C\setminus e_j$ such that $\mu_j\in \pNeigh{C_j}$.
The edge $e_j$ is uniquely determined for every $j$, since $\overline{C}$ contains exactly one $\mu_j$-important vertex of $C$ and this vertex separates $\overline{C}$ from all the vertices of $C$ incident to edges $f$ for which $\hat{\varphi}_4(f)=\mu_j$.
Note that $\nCross{e_1}{e_3}{\hat{\psi}_4}\equiv_2\nCross{e_1}{e_4}{\hat{\psi}_4}\equiv_2  \nCross{e_2}{e_3}{\hat{\psi}_4}\equiv_2\nCross{e_2}{e_4}{\hat{\psi}_4}\equiv_2 0$.
Hence, the claim follows by Claim~\ref{claim:basic}.
\end{proof}

Refer to Figure~\ref{fig:rotation2}.
Consider a vertex $v$ of $\overline{C}$ and a 4-tuple of incident edges $e_1,e_2,e_3$ and $e_4$ to $v$ (again overriding previous $e_j$'s).
Note that every edge  $e_i$, for $i=1,\ldots,4$, belongs to $C$ by the construction of $\overline{C}$.
Let $C_1(e_i)$, for $i=1,\ldots,4$, be such that $v\not\in V(C_1(e_i))$. By the construction of $\overline{C}$ we also have that if both $\pNeigh{C_1(e_i)}\cap \pNeigh{C_1(e_j)}\not=\emptyset$ then $i=j$. Moreover, by Claim~\ref{claim:rotation} for every
pair $1\le i< j \le 4$, the edges joining $\nu$ with vertices in $\pNeigh{C_1(e_i)}$ do not alternate 
in the rotation at $\nu$ with the edges joining it with $\pNeigh{C_1(e_j)}$.
Therefore by Claim~\ref{claim:basic}, every 4-tuple of edges incident to a vertex $v$ of $\overline{C}$ can be made even by flips at $v$ by the same  argument as in the case when $\overline{C}$ consists of a single vertex. It follows that the vertex $v$ of $\overline{C}$ can be made even by flips by Claim~\ref{claim:4}. Thus, by Claim~\ref{claim:contraction} contracting $\overline{C}$ into a  single vertex $v_C$ yields an instance which is a clone of the instance before the contraction. 
 \end{itemize}
 
Doing likewise for all such $\overline{C}$ we arrive at $(\hat{G}_5,\hat{H}_5,\hat{\varphi}_5)$ in the subdivided normal form with $V_s=\{v_C| \ C$ is a connected component of $G[V_{\nu}]$ for some $i \}$.
Indeed, (\ref{it:first}) and~(\ref{it:second}) in the definition of the subdivided normal form are immediately satisfied by $(\hat{G}_5,\hat{H}_5,\hat{\varphi}_5)$ and $V_s$. In particular, (\ref{it:second}) follows from~(\ref{it:tree}).
In order to prove~(\ref{it:first}) it is enough to observe that a connected component $C$ of $\hat{G}_5[V\setminus V_s]$ mapped by $\hat{\varphi}_5$ to at least two edges of $\hat{H}_5$ must contain a vertex $v_s$ of $V_s$. However, this is impossible since $v_s$ would be a cut vertex of $C$ by the construction.     
The rest of the properties, in particular, the properties of vertices $v_s\in V_s$ follow immediately from the construction. 
Finally, by the construction also every vertex $v_s\in V_s$  is even in $\hat{\psi}_5$ and that concludes the proof.
\end{proof}

\section{Derivative of a \texorpdfstring{$\mathbb{Z}_2$}{Z2}-approximation}
\label{sec:derivative_z2approx}

\begin{figure}
\centering
\subfloat[]{\label{fig:derivating}
\includegraphics[scale=0.7]{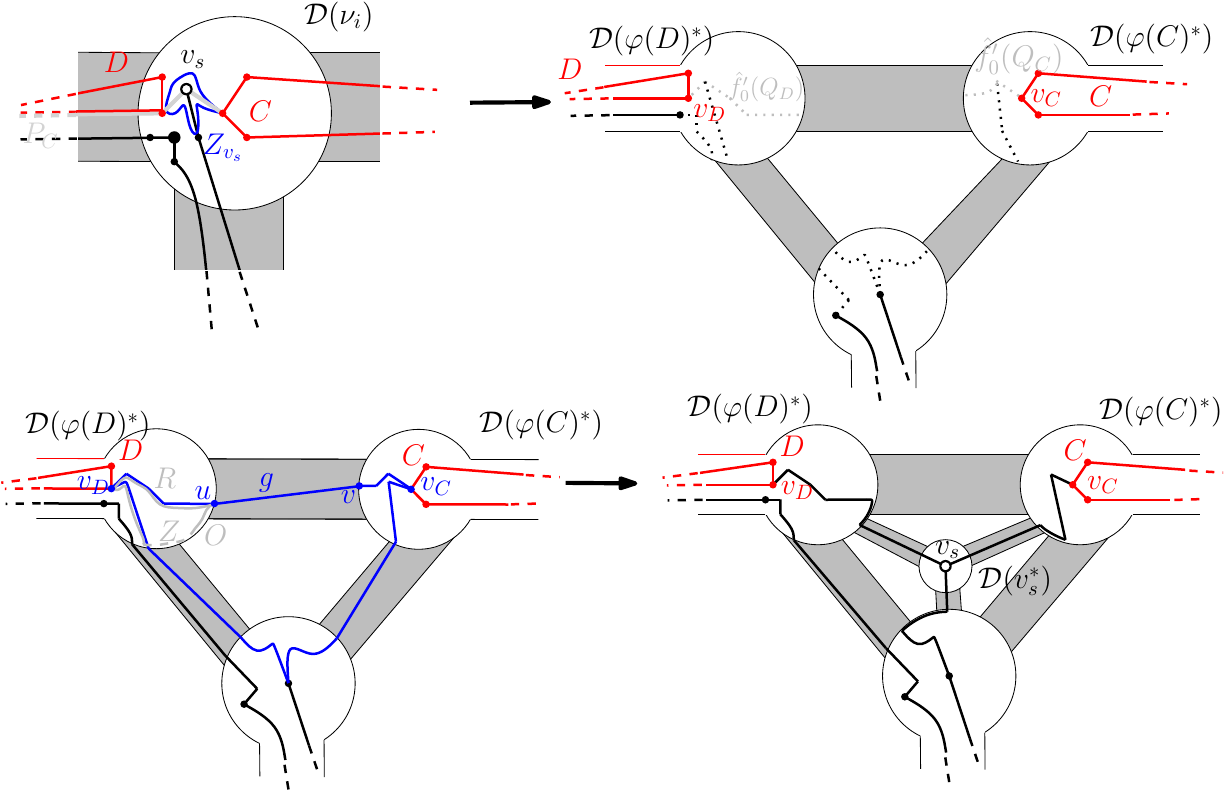}} \\
\subfloat[]{\label{fig:derivating2}
\includegraphics[scale=0.7]{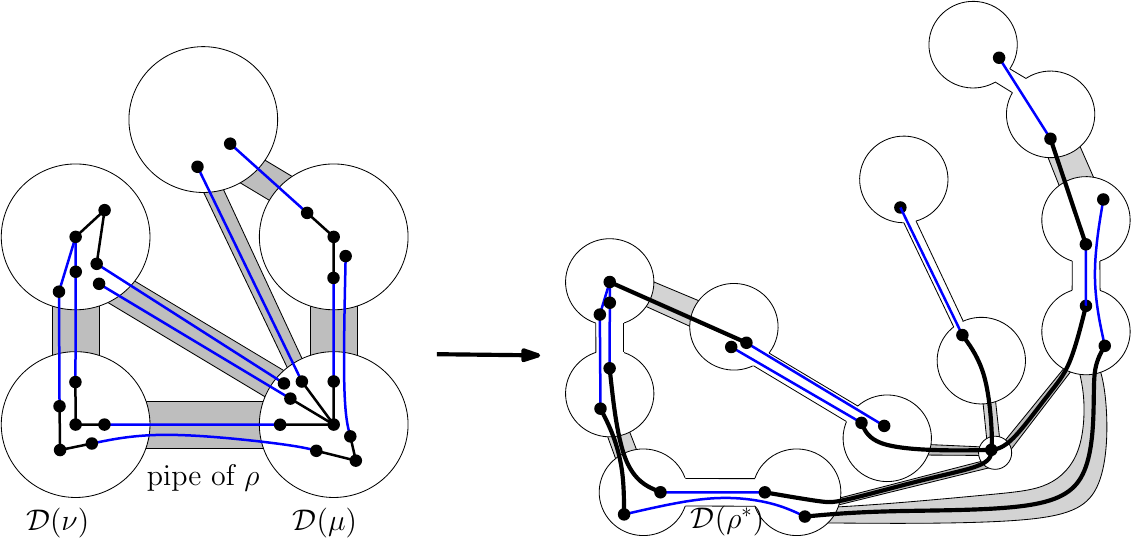} }
\caption{(a) Derivating a $\mathbb{Z}_2$-approximation $\psi_0$, local view. The first step of the construction of $\psi_0'$, in which the edges connecting copies of parts of the original $\mathbb{Z}_2$-approximation are drawn along dotted parts not included in the resulting $\mathbb{Z}_2$-approximation of the derivative (top, and bottom left). The second step of the construction $\psi_0'$, in which we use the part of $g$ inside $\mathcal{D}(\varphi(D)^*)$ to draw an edge incident to $v_s$ (bottom). (b) Derivating a $\mathbb{Z}_2$-approximation $\psi_0$, global view. The blue parts are mapped homeomorphically from the $\mathbb{Z}_2$-approximation of the original instance, on the left, to the $\mathbb{Z}_2$-approximation of its derivative, on the right.}
\end{figure}

Let $\psi_0$ be a {$\mathbb{Z}_2$-approximation} of $(G,H,\varphi)$ in the normal form.

The \emphh{$\mathbb{Z}_2$-derivative} $\psi_0'$ of $\psi_0$ is a {$\mathbb{Z}_2$-approximation} of $(G,H,\varphi)'=(G',H',\varphi')$ constructed as follows. (Similarly as above $\psi_0^{(i)}$ denotes $(\psi_0^{(i-1)})'$.)
We construct $\psi_0'$ in two steps.

\paragraph*{Step 1.}
Refer to Figure~\ref{fig:derivative05} (right), and Figure~\ref{fig:derivating} (top).

The goal in Step 1 is to construct a $\mathbb{Z}_2$-approximation $\hat{\psi}_0'$ of an auxiliary instance $(\hat{G}',\hat{H}',\hat{\varphi}')$ that is a slight modification of $(G',H',\varphi')$, in which
all $v_s\in V_s$ are eliminated by generalized \yDelta operations.

To this end we first construct an auxiliary instance $(\hat{G},H,\hat{\varphi})$,
where $\hat{G}$ is obtained from $G$ by applying the generalized \yDelta operation to every vertex in $V_s$, which must be of degree at least $3$, where the rotation is taken from $\psi_0$. We put $\hat{\varphi}(v):=\varphi(v)$ for all $v\in V(\hat{G})$.
Let $\hat{\psi}_0$ be the $\mathbb{Z}_2$-approximation of $(\hat{G},H,\hat{\varphi})$ obtained from $\psi_0$ by drawing the newly introduced edges along the images of the corresponding edges incident to $v_s$ in $\psi_0$. In particular, we draw a new edge $uv$ by closely following the curve $\psi_0(uv_sv)$. The resulting drawing is independently even, since every vertex $v_s\in V_s$ can be assumed to be even in $\psi_0$ by Claim~\ref{claim:normal_form}. 

 
 \begin{figure}
\centering
\includegraphics[scale=0.7]{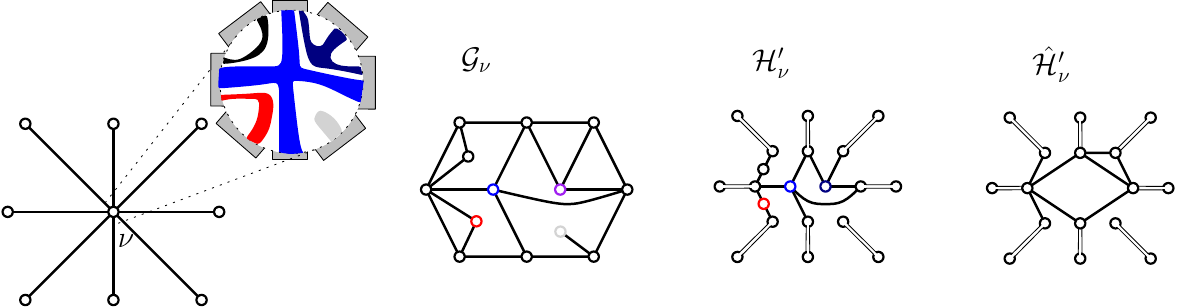}
\caption{The graph $G_{\nu}$ corresponding to a vertex $\nu$ of $H$ (left). Colors encode the correspondence between the connected components of $\varphi^{-1}[{\nu}]$ in $G$ and vertices of $G_{\nu}$. The thickening $\mathcal{H}_i'$ and $\hat{\mathcal{H}}_i'$ of the subgraph $H_{\nu}'$ and $\hat{H}_{\nu}'$, respectively, in $H'$ and $\hat{H}'$ (right). In the thickening, vertices have
a dumbbell shape just to illustrate the homeomorphism $h_\rho$, for $\rho\in E(H)$. In the actual thickening the vertices of $H'$ and $\hat{H}'$ give rise to discs, of course.}
\label{fig:derivative05}
\end{figure} 

We naturally extend the definition of the derivative to $(\hat{G},H,\hat{\varphi})$ thereby obtaining $(\hat{G}',\hat{H}',\hat{\varphi}')$. Formally, $\hat{G}=\hat{G}'$, $\hat{\varphi}'(v):=\varphi'(v)$, and hence, $V(\hat{H}')=\{\rho^*| \ \rho\in E(H)\}$. Then $E(\hat{H}')=\{\hat{\varphi}'(v)\hat{\varphi}'(u)| uv\in E(\hat{G})$ and $\hat{\varphi}'(v)\not=\hat{\varphi}'(u)\}$. 
Let $\hat{H}_{\nu}'$ be the  subgraph of $\hat{H}'$ induced by $\{(\nu\mu)^*|\ \nu\mu\in E(H)\}$ (see Figure~\ref{fig:derivative05} (right)).
The embedding of $\hat{H}'$ is obtained analogously as the embedding of ${H'}$ by merging embeddings of 
$\hat{H}_{\nu}'$'s inherited from embeddings of $G_\nu$'s given by  Claim~\ref{claim:Gi}.
 Note that $\hat{H}_{\nu}'$ is a graph obtained from a subgraph of $G_{\nu}$ (defined in Section~\ref{sec:preliminaries}) by suppressing some degree-$2$ vertices and applying generalized \yDelta operation
 to vertices $v_s^*$, for $v_s\in V_s$.
 Hence, by Claim~\ref{claim:Gi} every $\hat{H}_{\nu}'$ is a planar graph if $(\hat{G},H,\hat{\varphi})$ is $\mathbb{Z}_2$-approximable.
 The rule for putting the signs on the edges of $\hat{H}'$ is analogous to that for $H'$.

We construct $\hat{\psi}_0'$.
This step more-or-less follows considerations of M.~Skopenkov~\cite{Sko03_approximability}.
The derivative $\hat{\psi}_0'$ is constructed by a surgery
of $\hat{\psi}_0$, in which we first copy parts of the drawing $\hat{\psi}_0$, where each part is a subgraph of  $\hat{G}$ consisting of a set edges in $\hat{G}$, whose corresponding edges in $\hat{G}'=\hat{G}$, are inside a single cluster in $(\hat{G}',\hat{H}',\hat{\varphi}')$.
These parts are then reconnected by pipe edges of $(\hat{G}',\hat{H}',\hat{\varphi}')$.
The crucial idea is to define the drawing $\hat{\psi}'$ on the pieces of pipe edges inside discs $\mathcal{D}(\rho^*)$ so that they follow a copy of a drawing of a path by $\hat{\psi}_0$. Then reconnecting the severed ends of pipe edges by curves 
contained in pipes leads somewhat magically to a $\mathbb{Z}_2$-approximation of $(\hat{G}',\hat{H}',\hat{\varphi}')$ as proved in Claim~\ref{claim:ie}. \\

\begin{remark}
Such a surgery did not appear in the previous work of Pelsmajer et al.~\cite{PSS06_removing} and therefore the reader might wonder if the argument cannot be simplified by a more direct 
argument using only continuous deformation followed by a redrawing of $\mathcal{H}$. The reason for why a local continuous deformation argument, that does not take into account the whole $\psi_0$, is likely to fail, is provided, for example, by the previously discussed ``standard winding example''
~\cite[Figure 1]{ReSk98_deleted} and its modifications.

Intuitively, we would like to modify $\psi_0$ inside $\mathcal{D}(\nu)$ so that $\psi_0$ remains independently even, a pair of pipe edges
cross only if they belong to the same pipe, and a pair of  components of $G[V_\nu]$
cross only if they are adjacent to pipe edges in the same pipe. The ``standard winding example'' shows that $\psi_0$ can entangle the edges in a way that no continuous deformation of $\psi_0$ in $\mathcal{D}(\nu)$ will  accomplish this.

A sketch of a formal proof of this, which follows, uses a technique introduced in~\cite{FKMP15} that we discuss in Section~\ref{sec:reduction}. In the example, the relation $<_{\bf p}$ with respect to $\psi_0$  (see Definition~\ref{def:relation}), where ${\bf p}\in \partial\mathcal{D}(\nu)$, is cyclic regardless of the choice of $\nu$. It is not hard to see that a desired continuous deformation of $\psi_0$ taking place inside $\mathcal{D}(\nu)$ renders $<_{\bf p}$ acyclic, which in turn contradicts Claim~\ref{claim:comparison2}. 
\end{remark} 

We will be mimicking the construction of the embedding of $\hat{H}'$ on the level of  $\mathcal{\hat{H}}'$, which is  the thickening of $\hat{H}'$.
Therefore for constructing $\hat{H}'$ we inherit the orientation of $\partial \mathcal{D}((\nu\mu)^*)$, for every $\nu\mu\in E(H)$, from $\partial \mathcal{D}(\mu)$, where $\mu$ is chosen arbitrarily if $M$ is orientable, and otherwise $\mu$ is the attractor (as chosen in the construction of the embedding of $H'$ in Section~\ref{sec:normal-form}) of $\nu\mu$.\footnote{
One can think of this as follows. The attractor of $\nu\mu$, let's say $\mu$, is the one that pretty much stays at its original location in the construction of the embedding of $H'$ on $M$,
while $\nu$ is being dragged closer to $\mu$ thereby pushing $\mu\nu$ off all the cross-caps. Now, considering $\mathcal{H}$ embedded on $M$, if the sign on $\nu\mu$ was negative, the orientation of the portion of  $\partial\mathcal{D}((\nu\mu)^*)$, which is inherited from the original orientation of $\partial\mathcal{D}(\mu)$, corresponding to $\partial\mathcal{D}(\nu)$ is opposite to the original orientation of $\partial\mathcal{D}(\nu)$.}

In what follows we first define the restriction of $\hat{\psi}_0'$ to the subgraph of $\hat{G}'[V_{\rho^*}]$, for $\rho\in E(H)$. The restriction of $\hat{\psi}_0'$ is defined as $h_\rho\circ \hat{\psi}_0$,
where $h_\rho$ is  an orientation preserving homeomorphism mapping the union of the pipe of $\rho=\nu\mu$, $\mathcal{P}(\rho)$ with $\mathcal{D}(\nu)$ and $\mathcal{D}(\mu)$ to the disc $\mathcal{D}(\rho^*)$ in $\hat{\mathcal{H}}'$ such that the corresponding valves are mapped onto each other.
This is illustrated in Figure~\ref{fig:derivative05} (right).
Formally, a valve at $\mathcal{D}(\nu)$ of $\nu \nu_1$, where  $\nu_1\not =\mu$, is mapped bijectively onto the valve of $\rho^*(\nu \nu_1)^*$ and similarly
a valve at $\mathcal{D}(\mu)$ of $\mu \mu_1$, where $\mu_1\not =\nu$, is mapped onto the valve of $\rho^*(\mu \mu_1)^*$.

Refer to Figure~\ref{fig:derivating} (top, 	bottom left).
Second, we define the restriction of $\hat{\psi}_0'$ to pipe edges of $\hat{G}'$, i.e., edges mapped by $\hat{\varphi}'$ to edges of $\hat{H}'$ which will conclude the construction of $\hat{\psi}_0'$. Let $g=v_Cv_D$ be such an edge, where $v_C$ and $v_D$ belong to a connected component $C$ and $D$, respectively, of $\hat{G}'\setminus E_p(\hat{G}')$, where $E_p(\hat{G}')$ is the set of pipe edges in $(\hat{G}',\hat{H}',\hat{\varphi}')$. It must be that if $g$ does not belong to $G$, then $g$ was introduced to $\hat{G}'=\hat{G}$ by the generalized \yDelta operation, and that $\hat{\varphi}'(g)=\varphi(C)^*\varphi(D)^*$.
Let $P_C$ and $P_D$ be shortest paths in $\hat{G}'$ with all the internal vertices contained in $V(D)$ and $V(C)$, respectively, starting at $v_C$ and $v_D$, respectively, and continuing by the edge $v_Cv_D$, such that its last edge is a pipe edge mapped by $\hat{\varphi}$ to $\varphi(D)$ and $\varphi(C)$, respectively.
Let $Q_C$ and $Q_D$ be the subcurves of $P_C$ and $P_D$, respectively, mapped by $\hat{\psi}_0$ to the disc $\mathcal{D}(\nu)$.
Let us split $g$ into three arcs $g_1,g_2$ 
and $g_3$, where $g_1$ contains $v_C$ and $g_3$ contains $v_D$.
Let $h_1$ be a homeomorphism between $g_1$ and $Q_C$ mapping the end vertex $v_C$ of $g_1$ to the end point of $Q_C$ representing $v_C$. Similarly, let $h_{3}$ be a homeomorphism between $g_3$ and $Q_D$ mapping the end vertex $v_D$ of $g_3$ to the end point of $Q_D$ representing $v_D$.
Recall the definition of $h_{\rho}$ from the previous paragraph.
 We put $\hat{\psi}_0'(g_1):=h_{\varphi(C)}\circ \hat{\psi}_0 \circ h_1(g_1)$ and we put $\hat{\psi}_0'(g_3):=h_{\varphi(D)}\circ \hat{\psi}_0 \circ h_3(g_3)$.
 Finally, $\hat{\psi}_0'(g_2)$ is defined so that $\hat{\psi}_0'(g)$ is a non-self-intersecting arc and $\hat{\psi}_0'(g_2)$ is contained 
 in the pipe of $\varphi(C)^*\varphi(D)^*$.
 
 \begin{claim}
 \label{claim:ie}
 If $\hat{\psi}_0$ is independently even then $\hat{\psi}_0'$ is independently even.
 \end{claim}
\begin{figure}
\centering
\includegraphics[scale=0.7]{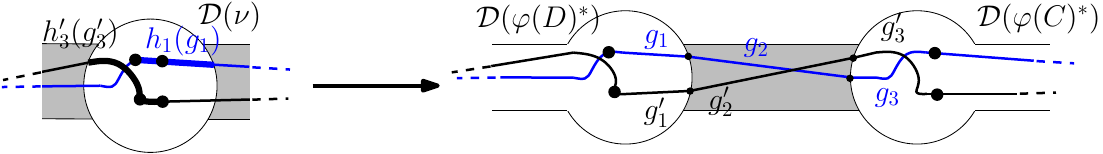}
\caption{Illustration for the proof of Claim~\ref{claim:ie}. The drawing  $\hat{\psi_0}(h_1(g_1))$ and $\hat{\psi_0}(h_3'(g_3'))$ is thickened on the left.}
\label{fig:ie}
\end{figure}

 \begin{proof}
  Refer to Figure~\ref{fig:ie}.
  For the sake of contradiction suppose the contrary.
By the construction, a pair of independent edges of $\hat{G}'$ that cross an odd number of times under $\hat{\psi}_0'$ must be formed by two pipe edges in the same pipe.
Let $g$ and $g'$ be edges in such pair. Let $g_1,g_2,g_3$ and $g_1',g_2',g_3'$, be the partition of the curve representing $g$ and $g'$, respectively, as in the definition of the derivative. In, $\hat{G}=\hat{G}'$ we temporarily subdivide $g'$ and $g''$ as indicated in the figure so that the curves 
$g_1,g_2,g_3$ and $g_1',g_2',g_3'$ are images of the  edges created by the subdivisions in $\hat{\psi}_0'$. We denote the  edges  created by the subdivisions by the same corresponding symbols.
We are done if we prove that $\nCross{g_1}{g_1'}{\hat{\psi}_0'}+\nCross{g_2}{g_2'}{\hat{\psi}_0'}+\nCross{g_3}{g_3'}{\hat{\psi}_0'}\equiv_2 0$ or equivalently,
that 

\begin{equation}
\label{eqn:0}
\nCross{g_1}{g_1'}{\hat{\psi}_0'}+\nCross{g_3}{g_3'}{\hat{\psi}_0'}\equiv_2\nCross{g_2}{g_2'}{\hat{\psi}_0'}.
\end{equation}
With this goal in mind we first show that $\nCross{g_2}{g_2'}{\hat{\psi}_0'}\equiv_2 1$ if and only if the end points of $h_1(g_1)$ and $h_3(g_3)$ alternate along the boundary of $\mathcal{D}(\nu)$ with   $h_1'(g_1)$ and $h_3'(g_3)$, , where $h_1',h_3'$ is defined analogously as $h_1,h_3$.
The claim holds due to the choice of signs of the edges of $\hat{H}'$, and the orientation of $\partial \mathcal{D}(\varphi(C)^*)$ and $\partial \mathcal{D}(\varphi(D)^*)$. The sign on the edge
$\varphi(C)^*\varphi(D)^*$ is negative if and only if   exactly one of the edges $\varphi(C)\in E(H)$ and $\varphi(D)\in E(H)$ has the negative sign and does not have $\nu$ as the attractor.
This means that the sign is also negative if and only if the orientation of exactly one of  $\partial \mathcal{D}(\varphi(C)^*)$ and $\partial \mathcal{D}(\varphi(D)^*)$ was inherited from $\partial \mathcal{D}(\nu)$. Hence, the two ``flips'', the one of the pipe of $\varphi(C)^*\varphi(D)^*$
and the one of $\mathcal{D}(\nu)$ cancel out.

It follows that  $\nCross{g_2}{g_2'}{\hat{\psi}_0'}\equiv_2 1$ if and only if
$$\nCross{h_1(g_1)}{h_1(g_1')}{\hat{\psi}_0}+\nCross{h_1(g_1)}{h_3'(g_3')}{\hat{\psi}_0}+\nCross{h_3(g_3)}{h_1'(g_1')}{\hat{\psi}_0}+\nCross{h_3(g_3)}{h_3'(g_3')}{\hat{\psi}_0}\equiv_2 1.$$
Notice that   
 \begin{equation}
\label{eqn:1}
\nCross{h_1(g_1)}{h_3'(g_3')}{\hat{\psi}_0}+\nCross{h_3(g_3)}{h_1'(g_1')}{\hat{\psi}_0}\equiv_2 0,
 \end{equation} 
 as $H$ has no multi-edges or loops,
 and $\hat{\psi}_0$ is independently even.
 Hence,~(\ref{eqn:0}) follows from~(\ref{eqn:1}) and the previous observation,
 since 
  $\nCross{g_1}{g_1'}{\hat{\psi}_0'} = \nCross{h_1(g_1)}{h_1(g_1')}{\hat{\psi}_0}$ and $\nCross{g_3}{g_3'}{\hat{\psi}_0'}=\nCross{h_3(g_3)}{h_3'(g_3')}{\hat{\psi}_0}$.
 \end{proof} 
 
\paragraph*{Step 2.}
We use the drawing $\hat{\psi}_0'$ of $\hat{G}'$ to define a desired drawing $\psi_0'$ of $G'$ thereby proving Claim~\ref{claim:derivative} which we restate for the reader's convenience.

\bigskip

\rephrase{Claim}{\ref{claim:derivative}}{
If the instance $(G,H,\varphi)$ is $\mathbb{Z}_2$-approximable by the drawing $\psi_0$ then $(G'=G,H',\varphi')$ is
$\mathbb{Z}_2$-approximable by the drawing $\psi_0'$ such that $\psi_0$ is compatible with $\psi_0'$. Moreover, if $\psi_0$ is crossing free so is $\psi_0'$.
}

\begin{proof}
Refer to Figure~\ref{fig:derivating} (bottom). 
We use the notation from Step 1.
In what follows we modify  $\hat{\psi}_0'$ thereby obtaining a desired independently even drawing $\psi_0'$ of $G'$.
By Claim~\ref{claim:ie}, $\hat{\psi}_0'$ is independently even.
Let $g\in E(\hat{G}')$, such that $g\not\in E(G')$, be an edge joining a pair of connected components $C$ and $D$ of $\hat{G}'\setminus E_p(\hat{G}')$.
Since, $g\not\in E(G')$, $g$ belongs to a cycle $Z_{v_s}$ obtained by the application of the generalized \yDelta operation to $v_s\in V_s$ of degree at least $3$. 
In what follows we define $\psi_0'$ on the edges incident to $v_s$ by altering the edges of $Z_{v_s}$ in $\hat{\psi}_0'$ .
By repeating the same procedure for every such cycle
$Z_{v_s}$ we obtain $\psi_0'$.

By the definition of the subdivided normal form, $g$ is a pipe edge in $(\hat{G}',\hat{H}',\hat{\varphi}')$. Let us direct the edges of $Z_{v_s}$ so that when following the inherited directions along the cycle $\hat{\varphi}'(Z_{v_s})$ we travel clockwise around
a point in its interior. We subdivide
every edge $g_0$ on $Z_{v_s}$ by a pair of vertices at both valves that it intersects. 
A subdividing vertex is assumed to belong to the same cluster as the end vertex of $g_0$ that is its neighbor. 

Let $u$ be a vertex subdividing $g$ that is joined with the source of $g$ with respect to the chosen directions on the edges of $Z_{v_s}$. Let's say that $u\in V_{\varphi(D)^*}$.
We show that by performing the subdivisions one by one we can maintain the drawing $\hat{\psi}_0'$ independently even.
Indeed, it cannot happen that an edge not sharing a vertex with $g$ crosses exactly one newly created edge an odd number of times, since $\hat{\psi}_0'$ was independently even before the subdivision.
Hence, we just pull edges crossing both newly created edges an odd number of times over the subdividing vertex,
see Figure~\ref{fig:argument}. As illustrated by the figure, this can be performed while not forcing the edge being pulled to cross the boundary of the cluster more than once, since we can pull the edge along the boundary of the cluster.
Note that this would not be possible if more than one subdividing vertex is contained in the same valve, since then the edge being pulled would have to avoid such vertex\footnote{This also kills a na\"ive approach to c-planarity via a Hanani--Tutte variant.}.

\begin{figure}
\centering
\includegraphics[scale=0.7]{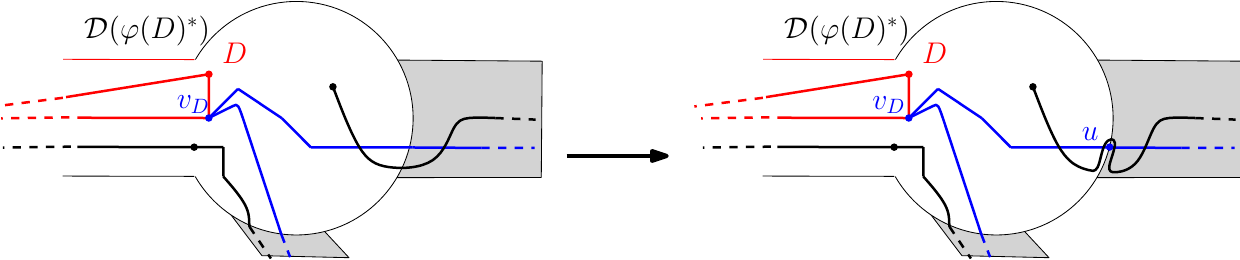}
\caption{Pulling an edge along the boundary of the cluster so that both newly created edges resulting from the subdivision of $g$ cross the edge an even number of times.}
\label{fig:argument}
\end{figure}

Let us put into $\hat{H}'$ all the edges and vertices of $H'$ missing in $\hat{H}'$. We assume that the image of $\hat{\psi}_0'$ is contained in the thickening of the obtained graph. This can be done while keeping the resulting graph crossing free.
Let $O:=O(v_s)$ be a curve following  $\partial\mathcal{D}(\varphi(D)^*)$ joining $u$ and a point of the valve
of $\varphi(D)^*v_s^*$ that is disjoint from all the other valves.
We claim that \\

(*) $O$ crosses every edge in the drawing $\hat{\psi}_0'$ evenly, and hence, zero number of times. \\

If the claim is true, the part of edge $g$ not contained in $R$ can be redrawn so that
from $u$ the edge $g$ continues along $O$ and ends in $v_s$. Then we can apply the same redrawing procedure to every other edge of $Z_{v_s}$.
Note that repeating the same procedure to every other cycle $Z_{v_s}$ does not introduce a pair of non-adjacent edges crossing an odd number of times in the resulting drawing. Indeed, since all the cycles $Z_{v_s}$ are directed clockwise, a pair of curves $O$ defined for different $v_s$'s do not end on the same valve, and hence, the redrawn pieces of edges do not cross.

\begin{figure}
\centering
\includegraphics[scale=0.7]{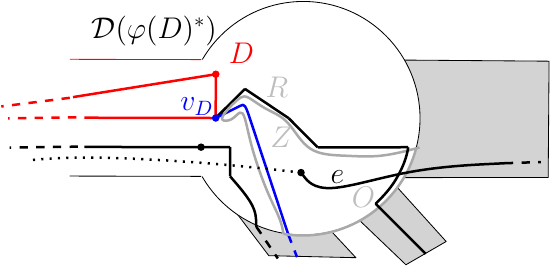}
\caption{The edge $e$ crossing $O$ an odd number of times. The dotted edge cannot exist in $\hat{G}'$.}
\label{fig:argument2}
\end{figure}

Hence, in order to show that $\psi_0'$ is independently even it remains to prove (*). Refer to Figure~\ref{fig:argument2}.
 Let $R:=R(v_s)$ be the path contained in $Z_{v_s}$, such that $\hat{\psi}_0'(R)\subset \mathcal{D}(\varphi(D)^*)$, joining $u$ with the other subdividing vertex of an edge on $Z_{v_s}$.
Let the closed curve $Z:=Z(v_s)$ be obtained by concatenating $O$, $R$ and the part $\partial\mathcal{D}(\varphi(D)^*)$ 
intersecting only valves containing its end points.
We claim that the curve $Z$ crosses every edge crossing the valve containing $u$ evenly.

 Indeed, if such an edge $e$ crosses $Z$ an odd number of times, one of its end vertices is inside $Z$. This follows since $e$ is not incident to a vertex contained in $R$ by the definition of $(\hat{G}',\hat{H}',\hat{\varphi}')$, because $e$ is a pipe edge. 
Thus, $e$ is forced to be incident to a connected component $D_0$
of $\hat{G}'\setminus E_p(\hat{G}')$  mapped by $\hat{\varphi}$
to $\varphi(D)$ such that $D_0$ does not contain a vertex of $R$.
Hence, $D_0$ must have at least one end vertex outside $Z_{v_s}$ (in the sense of the definition from Section~\ref{sec:preliminaries}), which inevitably leads to a pair of non-adjacent edges crossing an odd number of times in $\hat{\psi}_0$ (contradiction). 

 Due to the fact that $R$ consists only of edges of the graph it also crosses the edge $e$ evenly. Hence, the same holds 
for $O$, which concludes the proof of~(*) and therefore we proved that $\psi_0'$ is independently even.

Note that 
the order of the end pieces of edges at corresponding vertices in $\psi_0$ and $\psi_0'$ is the same, since we did not alter them in any way, which shows that $\psi_0$ is compatible with $\psi_0'$.

Finally, the ``moreover part'' follows straightforwardly from the construction of $\psi_0'$.
\end{proof}



\bigskip


\section{Integration} 
\label{sec:integration}

Let $(G',H',\varphi')=(G,H,\varphi)'$ and $\psi_0'$ be the $\mathbb{Z}_2$-approximation obtained by Claim~\ref{claim:derivative} from a $\mathbb{Z}_2$-approximation $\psi_0$ of $(G,H,\varphi)$, i.e., $\psi_0'$ is the $\mathbb{Z}_2$-derivative of $\psi_0$. 
The aim of this section is to prove that an approximation $\psi'$ of $(G',H',\varphi')$ yields an approximation $\psi$ of $(G,H,\varphi)$.
More precisely, we prove Claim~\ref{claim:integration} that we restate  for the reader's convenience. \\


\medskip  
\rephrase{Claim}{\ref{claim:integration}}
{If the instance $(G',H',\varphi')$ is approximable by an embedding $\psi'$ compatible with $\psi_0'$ then $(G,H,\varphi)$ is
 approximable by an embedding compatible with $\psi_0$. 
}
\jk{tohle mi prislo strasne matouci, jakto, ze se najednou objevil claim 16/17, kdyz tvrzeni okolo maji cisla jako 28... Myslim, ze by se bez tohohle kopirovani dalo obejit.}

\begin{proof}
By following the approach of~\cite{F17_pipes}, see also~\cite{AFT18+_weak}, we show that there exists a desired embedding $\psi$ of $G$ that is an approximation of $\varphi$. This follows basically from
the well-known fact that for every plane bipartite graph there exists a simple closed curve that intersects every edge
of the graph exactly once and in a proper crossing.
This is also a consequence of the planar case of Belyi's theorem~\cite{B83_self}. In order to simplify the notation below we subdivide pipe edges of $G=G'$ and the corresponding edges of $H'$ so that $(G',H',\varphi')$ is the derivative of the instance  $(G,H,\varphi)$ after being brought into the  subdivided normal form. At the end we pass to the normal form again.

By Claim~\ref{claim:newSurface} (Section~\ref{sec:normal-form}), we assume that $\mathcal{H}'$ is embedded in $M$ as a small neighborhood 
of an embedding of $H'$ obtained by the procedure in the proof of Claim~\ref{claim:newSurface}.
Thus, the embedding of $H'$ is in some sense compatible with the embedding of $H$ in the former instance $(G,H,\varphi)$.
In what follows, we will first define  for each $\rho^*=(\nu\mu)^*\in V(H')$ a curve $C_{\rho^*}\subset \mathcal{D}(\rho^*)$, such that $\partial C_{\rho^*} \subset \partial\mathcal{D}(\rho^*)$ and 
$C_{\rho^*}\setminus \partial C_{\rho^*} \subset \mathcal{D}(\rho^*)\setminus \partial\mathcal{D}(\rho^*)$,  splitting $\mathcal{D}(\rho^*)$ into two parts as follows. Vertices coming from $V_\nu$ are in one part and and the vertices coming from $V_\mu$ are in the other part.  Having constructed $C_{\rho^*}$'s we connect them appropriately by curves contained in $M\setminus \mathcal{H}'$ so  that they yield boundaries of the discs and pipes homeomorphic to $\mathcal{H}$ containing the image of a required $\psi$.

\paragraph{Construction of $C_{\rho^*}$.}
By property~(\ref{it:second}) of the subdivided normal form, $\varphi|_{\nu}^{-1}(\varphi(D))$, where $D$ is a connected component of $G[V\setminus V_s]$ and $\nu \in \varphi(D)$, is a forest.
Let $T$ be a tree of the forest $\varphi|_{\nu}^{-1}(\varphi(D))$.
In $G'$, we contract every such tree $T$ into a vertex, and update $\psi'$ and $\varphi'$ accordingly.
Here, we crucially used~(\ref{it:second}), since we cannot afford to create loops at vertices by contractions, which would violate bipartiteness, as we will see next.
 Let $\hat{G}',\hat{\psi}'$ and $\hat{\varphi}'$ denote the updated $G',\psi'$ and $\varphi'$, respectively.

 Refer to Figure~\ref{fig:integration0}.
Since $G$ is in the subdivided normal form, in its derivative $(\hat{G}',H',\hat{\varphi}')$, every cluster $V_{\rho^*}$ of $\hat{G}'$, where $\rho^*=(\nu\mu)^*$, induces a bipartite graph $\hat{G}'[V_{\rho^*}]=(V_{\nu}' \uplus V_{\mu}', E_{\rho^*})$, whose part $V_{\nu}'$ and $V_{\mu}'$ contains vertices obtained by contracting trees corresponding to connected components of $G[V_{\nu}\setminus V_s]$ and $G[V_{\mu}\setminus V_s]$, respectively. 

By the same token, for every pipe edge $e'$ of $G'$ we have $\hat{\varphi}'(e)=\rho^*v_s^*$, for some $\rho\in E(H)$ and $v_s\in V_s$.
We subdivide once every pipe edge  $e$ of $\hat{G}'$, such that $\hat{\varphi}'(e)=\rho^*v_s^*$ for some $\rho\in E(H)$.
In the embedding $\hat{\psi}'$, we draw the subdividing vertex as $\hat{\psi}'(e)\cap \partial\mathcal{D}(\rho^*)$. Let $\partial_\nu\mathcal{D}((\nu\mu)^*)$ and $\partial_\mu\mathcal{D}((\nu\mu)^*)$ be the partition of $\partial\mathcal{D}((\nu\mu)^*)$ into two closed nonempty parts such that 
$\partial_{\nu}\mathcal{D}((\nu\mu)^*)$ and $\partial_\mu\mathcal{D}((\nu\mu)^*)$ contains all the vertices drawn as  $\hat{\psi}'(e)\cap \partial\mathcal{D}((\nu\mu)^*)$ for the pipe edges $e\in E(\hat{G}')$ incident to a vertex in $V_\nu'$ and $V_\mu'$, respectively. By slightly abusing the notation we denote also the resulting graph by 
$\hat{G}'$ and the resulting embedding by $\hat{\psi}'$ and will continue doing so in all their subsequent modifications.

\begin{figure}
\centering
\includegraphics[scale=1]{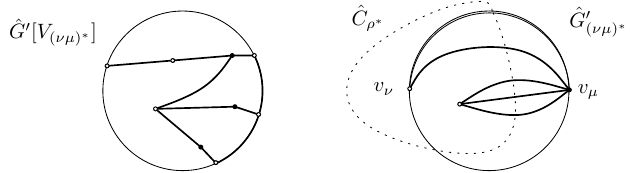}
\caption{On the left, the subgraph of $\hat{G'}$ induced by the cluster $(\nu\mu)^*$  embedded inside $\mathcal{D}((\nu\mu)^*)$. On the right, the graph $\hat{G}_{(\nu\mu)^*}'$ embedded inside embedded inside $\mathcal{D}((\nu\mu)^*)$. The empty and filled vertices belong to $V_\nu'$ and $V_\mu'$, respectively.
}
\label{fig:integration0}
\end{figure}

Next we  augment $\hat{G}'$ and $\hat{\psi}'$ as follows. We join by an edge, that is embedded in $\partial_\nu\mathcal{D}((\nu\mu)^*)$, every pair of vertices $w$ and $z$ such that  $\hat{\psi}'(w),\hat{\psi}'(z)\in\partial_\nu\mathcal{D}((\nu\mu)^*)$, and
   the part of $\partial_\nu\mathcal{D}((\nu\mu)^*)$ between $\hat{\psi}'(w)$ and $\hat{\psi}'(z)$ does not intersect the (image of) drawing $\hat{\psi}'$.
Then for every $(\nu\mu)^*\in V(H')$ we consider the subgraph ${G}_{(\nu\mu)^*}'$ of $\hat{G}'$ induced by the vertices drawn by $\hat{\psi}'$ in $\mathcal{D}((\nu\mu)^*)$. Let $\hat{G}_{(\nu\mu)^*}'$ be obtained from
${G}_{(\nu\mu)^*}'$ by successively contracting every edge incident to a vertex embedded by $\hat{\psi}'$ in $\partial\mathcal{D}((\nu\mu) ^*)$. Let $v_\nu$ and $v_\mu$ denote the vertices of  $\hat{G}_{(\nu\mu)^*}'$  that the edges of ${G}_{(\nu\mu)^*}'$ were contracted into in the previous step.  Either  of $v_\nu$ and $v_\mu$ might not exist. We assume that $v_\nu$ and $v_\mu$ is embedded in 
$\partial_\nu\mathcal{D}((\nu\mu)^*)$ and $\partial_\mu \mathcal{D}((\nu\mu)^*)$, respectively, and $\hat{G}_{(\nu\mu)^*}'$ is embedded in  $\mathcal{D}((\nu\mu)^*)$.  If $v_\nu$ and $v_\mu$ does not exist we embed $v_\nu$ and  $v_\mu$, respectively, as an isolated vertex in $\partial_\nu\mathcal{D}((\nu\mu)^*)$ and $\partial_\mu \mathcal{D}((\nu\mu)^*)$.
 Finally, we augment $\hat{G}_{(\nu\mu)^*}'$ by joining $v_\nu$ and $v_\mu$ by an edge drawn in $\partial\mathcal{D}((\nu\mu)^*)$ and by a slight abuse of notation we call 
 the resulting graph  $\hat{G}_{(\nu\mu)^*}'$.
 
 The graph $\hat{G}_{(\nu\mu)^*}'$ is bipartite and embedded in a disc. Therefore Belyi's theorem applies and yields a closed curve $\hat{C}_{\rho^*}$ crossing every edge of $\hat{G}_{(\nu\mu)^*}'$ in $\hat{\psi}'$ exactly once. Note that we can assume that  $|\hat{C}_{\rho^*}\cap \partial\mathcal{D}((\nu\mu)^*)|=2$, where 
 both intersection points are crossing points. 
 By taking ${C}_{\rho^*}:=\hat{C}_{\rho^*}\cap \mathcal{D}((\nu\mu)^*)$ and reversing the contractions we recover the original embedding of $G'[V_{(\nu\mu)^*}]$ so that 
 ${C}_{\rho^*}$ crosses its every edge exactly once and splits $\mathcal{D}((\nu\mu)^*)$ into two parts as desired.

\begin{figure}
\centering
\includegraphics[scale=1]{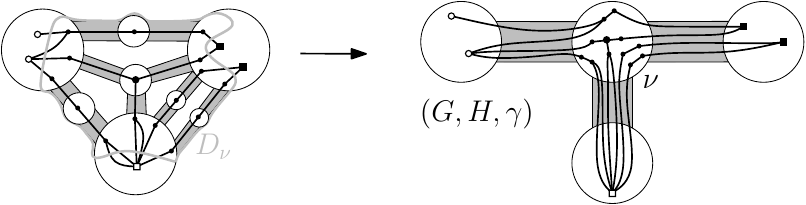}
\caption{Construction of the disc $D_\nu$ and the deformation into $\mathcal{H}$. Note that the instance on the left is in the subdivided normal form.}
\label{fig:integration}
\end{figure}

\paragraph{Construction of $\psi$.}
Refer to Figure~\ref{fig:integration}.
We construct $\psi$ as $h \circ \psi'$, where $h$ is a homeomorphism from a closed subset of $M$ to $\mathcal{H}$. In the following, by slightly abusing the introduced terminology we define this closed subset as the union of discs $\mathcal{D}(\nu)$, for all $\nu\in H$, and pipes joining them, that is homeomorphic via $h$ to $\mathcal{H}$. 
Let $\mathcal{C}_\nu=\{C_{\rho^*}| \ \rho=\nu\mu\in E(G)\}$. Let $D_\nu$ be a close curve obtained by concatenating curves in $\mathcal{C}_\nu$ in the order corresponding to the rotation of $\nu$ in $H$ by curves following the boundary of $\mathcal{H}'$.
Here, we assume that all the edges of the cycle $C_{\nu}$ defined in Section~\ref{sec:preliminaries} are  present in $H'$ and embedded as specified in the construction of the embedding of $H'$ on $M$ in Section~\ref{sec:normal-form} , for every $\nu\in V(H)$, which can be assumed without loss of generality, since they do not violate the embeddability of $H'$ on $M$.

We will show that the curve $D_\nu$ bounds a desired discs $\mathcal{D}(\nu)$ containing the images of all the vertices in $V_{\nu}$ in the embedding.
Note that $D_\nu$ and $D_\mu$ are overlapping if $\nu\mu\in E(H)$.
We perturb these curves to eliminate such overlaps, and introduce pipes, whose valves correspond to formerly overlapping parts.
This is done in a way which can be thought of as a reversal of the construction of the embedding of $H'$ in the proof of Claim~\ref{claim:newSurface}. Therefore the original signs on the edges of $M$ are easily recovered.
Note that $D_\nu$, for every $\nu\in V(H)$, passes an even number of times through the cross-caps of $M$ by the definition of the signs of the edges in $E(H')$ in Section~\ref{sec:normal-form}.  Indeed, $D_\nu$ passes along each edge of $H$ through cross-caps in both directions. Since we can assume that the embedding of $H'$ is cellular, it follows that $D_\nu$ is contractible, and hence, indeed bounds a topological disc $\mathcal{D}(\nu)$.
Then taking the union of $\mathcal{D}(\nu)$'s and their connecting pipes  we are done, since a desired approximation of $(G,H,\varphi)$ is obtained as $h\circ \psi'$, where $h$ is a homeomorphism between the union and $\mathcal{H}$ taking $\mathcal{D}(\nu)$ to $\mathcal{D}(\nu)\subseteq \mathcal{H}$. \\

Finally, we pass to the normal form by suppressing subdividing vertices accordingly.
Since we have not altered the rotation system in the proof, and the rotation system of $\psi_0$ and $\psi_0'$ is the same at the corresponding vertices in $G$ and $G'$, respectively, this concludes the proof.
\end{proof}

\jk{mostly skipped}

\section{Proof of Theorem~\ref{thm:main}}
\label{sec:reduction}


Let $(G,H,\varphi)$ be an instance that is $\mathbb{Z}_2$-approximable by an independently even drawing $\psi_0$. We start with a claim that helps us to identify instances that cannot be further simplified by derivating.

We show that by successively applying the derivative we eventually obtain an instance such that $\varphi$ is locally injective.

Let $E_p(G)$ be the set of pipe edges in $G$, and let $p(G,H,\varphi)=|E_p(G)| - |E(H)|$.

\begin{claim}
\label{claim:semi-trivial} If $(G,H,\varphi)$ is in the normal form then
$p(G',H',\varphi') \le p(G,H,\varphi)$. If 
additionally 
$\varphi$ is not locally injective after suppressing in $G$ all degree-$2$ vertices incident to an edge induced by a cluster, the inequality is strict; that is,
$p(G',H',\varphi') < p(G,H,\varphi)$.

Furthermore, if $G$ is connected and every connected component $C$ induced by $V_{\nu}$, for all $\nu\in V(H)$, has pipe degree at most $2$, then $|E_p(G')|\le |E_p(G)|$.
\end{claim}

\begin{proof}
We prove the first part of the claim and along the way establish the second part.
Let $\psi_0$ be a $\mathbb{Z}_2$-approximation of $(G,H,\varphi)$. 
Note that in $G'$ the pipe edges incident to a central vertex $v_s\in V_s$ (every such $v_s$ has degree at least $3$) and edges in $H'$ incident to $\varphi'(v_s)$ contribute together zero towards $p(G',H',\varphi')$.
Let $H_0'=H'\setminus \{\nu_{v_s}|\ v_s\in V_s\}$.\\ 
 $(*)$ The number of edges in $H_0'$ is at least $|V(H_0')|-c= |E(H)|-c$, where $c$ is the number of connected components of $H_0'$ that are trees. We use this fact together with a simple charging scheme in terms of an injective mapping $\mapping$ defined in the next paragraph to prove the claim.
 
 Suppose for a while that $H_0'$ is connected.
 The set of pipe edges of $G'$ not incident to any $v_s\in V_s$ forms a matching $M'$ in $G'$ by Definition~\ref{def:derivativeAbstract}. Let $D(v)$, $v\in V(G')$, be the connected component of $G'\setminus E_p(G')$ containing the vertex $v$. By the first property of the components $G[V\setminus V_s]$ in Definition~\ref{def:normal}, for every  $v\in V(M')$, the component $D(v)$ contains at least one former pipe edge, i.e., a pipe edge in $(G,H,\varphi)$.  Let $V_p$ be the set of vertices in $G'$ incident to these former pipe edges. We construct an injective mapping $\mapping$ from the set $V(M')$ to $V_p$.
The mapping $\mapping$ maps a vertex $v\in V(M')$ to a closest vertex (in terms of the graph theoretical distance in $D(v)$) in $V_p\cap V(D(v))$. The mapping $\mapping$ is injective by the fact, that in the corresponding subdivided normal form, every connected component of $G[V_{\nu}]$, for $\nu\in V(H)$, contains at most one central vertex.
 Indeed, recall that this central vertex is suppressed in the normal form and the edge thereby created becomes a pipe edge $e$ in $H'$, and thus, belongs to $M'$. Each end vertex $v$ of $e$ is mapped by $\zeta$ to a vertex $u$ such that $\varphi(u)=\varphi(v)$.
Thus, the injectivity could be violated only by the end vertices of $e$. However, this cannot happen since every connected component of $G[V_{\nu}]$ is a tree.
The injectivity of $\mapping$ implies that $2|M'|$ is upper bounded by $2|E_p(G)|$, and therefore $|M'|$ is upper bounded by $|E_p(G)|$, which proves the second part of the claim. Furthermore, $|E_p(G)|=|M'|$ only if after suppressing all vertices of degree $2$ incident to an edge induced by a cluster, $\varphi$ is locally injective, and $H_0'$ contains a cycle.
 
If $H_0'$ is disconnected, then we have $|M'|\le |E_p(G)|-c$, where the inequality is strict if $\varphi$ is not locally injective
after suppressing all degree-$2$ vertices incident to an edge induced by a cluster.
Indeed, if $|M'|= |E_p(G)|-c$, then there exist
exactly $2c$ vertices $v$, $v\in e\in E_p(G)$, that are not 
in the image of the map $\mapping$.
However, 
there are at least $2c$ vertices in $G'$  each of which is mapped by $\varphi'$ to a vertex of degree at most $1$ in $H_0'$. This follows since a connected component of $H_0'$, that is an isolated vertex $\nu$, contributes  at least two end vertices of an edge $e\in E_p(G)$ such that $\varphi'(e)=\nu$; and a connected component of $H_0'$, that is a tree, has at least two leaves
each of which contributes by at least one end vertex of an edge in $E_p(G)$ mapped to it by $\varphi'$. Hence, if $|M'|= |E_p(G)|-c$ then all the vertices that are not contained in the image of $\mapping$, are accounted for by these $2c$ vertices.

Putting it together, we have $|M'|\le |E_p(G)|-c$ and 
$(*)$ \  $|E(H)|-c \le |E(H_0')|$, where the first inequality is strict if $\varphi$ is not locally injective after suppressing all degree-$2$ vertices incident to an edge induced by a cluster.
Since the remaining pipe edges of $G'$ and edges of $H'$ contribute together zero towards $p(G',H',\varphi')$, summing up the inequalities concludes the proof.
\end{proof}

\jk{mostly skipped, but corrected some formulations}
 
\begin{claim}
\label{claim:semi-trivial2} Suppose that $\varphi$ is locally injective after suppressing all degree-$2$ vertices incident to an edge induced by a cluster.
Applying the derivative $|E_p(G)|$ many times yields an instance in which no connected component of $G$ is a path.
\end{claim} 

\begin{proof}
It suffices to observe that the maximum length of a path in $G$, that is a connected component, decreases after the derivative.
Trivial connected components are then eliminated by bringing the instance into the normal form.
\end{proof}

%

\begin{proof}[Proof of Theorem~\ref{thm:main}]
We assume that every edge of $H$ is in the image of $\varphi$ and
proceed by induction on $p(G,H,\varphi)$. First, we discuss the inductive step.
By Claim~\ref{claim:normal_form},
we assume that $(G,H,\varphi)$ is in the normal form, which leaves $p(G,H,\varphi)$ unchanged. Suppose that $\varphi$ is not locally injective after suppressing degree-$2$ vertices incident to an edge induced by a cluster. Derivating $(G,H,\varphi)$ decreases $p(G,H,\varphi)$ by 
Claim~\ref{claim:semi-trivial}. 
By Claims~\ref{claim:derivative} and~\ref{claim:integration}, $(G',H',\varphi',\psi_0')$ is a clone of $(G,H,\varphi,\psi_0)$, where $\psi_0'$ is obtained by  Claim~\ref{claim:derivative}. Hence, in this case we are done by induction.

Thus, we assume that $\varphi$ is locally injective, which includes the case when $p(G,H,\varphi)=0$. This means that either we reduced $G$ to an empty graph, or every connected component $C$ of $G[V_{\nu}]$, for every $\nu\in V(H)$, is a single vertex. We suppose that $G$ is not a trivial graph, since otherwise we are done.
 The proof will split into two cases, the {\bf acyclic} and the {\bf cyclic} case below, but first we introduce some tools from~\cite{FKMP15} that we use extensively in the argument.

We will work with a $\mathbb{Z}_2$-approximation $\psi_0$ of $(G,H,\varphi)$ unless specified otherwise.
Let $P$ be a path of length $2$ in $G$. Let the internal vertex $u$ of $P$ belong to $G[V_{\nu}]$, for some $\nu\in V(H)$. The curve obtained by intersecting the disc $\mathcal{D}(\nu)$ with $P$ is a \emphh{$\nu$-diagonal} \emphh{supported} by $u$. 
By a slight abuse of notation we denote in different drawings  $\nu$-diagonals with the same supporting vertex and joining the same pair of valves by the same symbol.
Let $Q$ be a $\nu$-diagonal supported by a vertex $u$ of $G[V_{\nu}]$. 
Since $\varphi$ is locally injective, $Q$ must connect a pair of distinct valves.
Let ${\bf p}$ be a point on the boundary of the disc $\mathcal{D}(\nu)$ of $\nu$ such that ${\bf p}$ is not contained in any valve. 

\begin{definition}
\label{def:relation}
(See Figure~\ref{fig:order} (left) for an illustration.)
For a vertex $v \in V_{\nu}$, $v \neq u$, we define $v<_{\bf p} Q$ if
in the two-coloring of the complement of $Q$ (such that connected regions sharing a non-trivial part of the boundary receive different colors)
in the disc $\mathcal{D}(\nu)$, $v$ receives the same color as the component having ${\bf p}$ on the boundary.
Let $Q_1$ and $Q_2$ be a pair of $\nu$-diagonals connecting the same pair of valves.
We define $Q_1 <_{\bf p} Q_2$ if for the vertex $v$ supporting $Q_1$ we have $v <_{\bf p} Q_2$.
Analogously we define the relation $>_{\bf p}$.
\end{definition}

\begin{figure}
\centering
\includegraphics[scale=0.8]{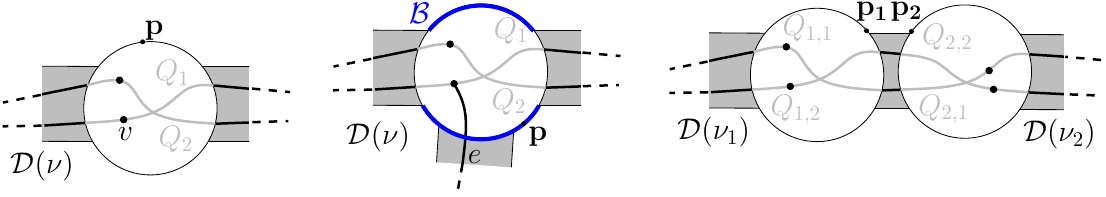}
\caption{A pair of $\nu$-diagonals $Q_1$ and $Q_2$, and a vertex $v$ such that $v<_{\bf p} Q_1$ and $Q_2<_{\bf p} Q_1$ (left). The edge $e$ forcing $Q_1<_{\bf p} Q_2$ (middle).
The relation $Q_{1,1}<_{\bf p_{1}}Q_{2,1}$ forces $Q_{1,2}<_{\bf p_{2}}Q_{2,2}$ (right).}
\label{fig:order}
\end{figure}

Recall that $H$ contains no multiple edges, since we do not introduce them by derivating. Thus, the upcoming Claims~\ref{claim:comparison} and~\ref{claim:comparison2} follow easily by the same argument as (1) and (2) in~\cite[Theorem 13]{FKMP15}.
\begin{claim}
\label{claim:comparison}
The relation $<_{\bf p}$ is anti-symmetric:
If for a pair of $\nu$-diagonals $Q_1,Q_2$ of $G[V_{\nu}]$ we have $Q_1<_{\bf p}Q_2$ then
$Q_1\not >_{\bf p}Q_2$.
\end{claim}

By Claim~\ref{claim:comparison}, the relation $<_{\bf p}$ defines a \emphh{tournament}, that is, a complete oriented graph, on $\nu$-diagonals joining the same pair of valves.
A pair of a $\nu_{1}$-diagonal $Q_1$ and a $\nu_{2}$-diagonal $Q_2$ of $G$ is \emphh{neighboring} if $Q_1$ 
and $Q_2$ have endpoints on the same (pipe) edge.

Let $Q_{1,i}$ and $Q_{2,i}$ be a neighboring pair of a $\nu_{1}$-diagonal and a $\nu_{2}$-diagonal sharing a pipe edge $e_i$, for $i=1,2$, such that $\varphi(e_1)=\varphi(e_2)=\rho=\nu_{1}\nu_{2}$.
Let ${\bf p_{1}}$ and ${\bf p_{2}}$ be on the boundary of $\mathcal{D}(\nu_{1})$ and $\mathcal{D}(\nu_{2})$, respectively, very close to the same side of the pipe of $\rho$.

\begin{claim}
\label{claim:comparison2}
If $Q_{1,1}<_{\bf p_{1}}Q_{1,2}$ then $Q_{2,1}<_{\bf p_{2}}Q_{2,2}$; see Figure~\ref{fig:order} (right)
for an illustration.
\end{claim}

Let $\mathcal{D}_1$ and $\mathcal{D}_2$ be a set of $\nu_{1}$-diagonals and $\nu_{2}$-diagonals, respectively, in $G$ of the same cardinality such that every $\nu_{1}$-diagonal in $\mathcal{D}_1$ ends on the valve of $\rho$ and forms a neighboring pair with a $\nu_{2}$-diagonal in $\mathcal{D}_2$.
We require that all the diagonals in $\mathcal{D}_1$ join the same pair of valves.
Let $\overrightarrow{G(\mathcal{D}_1)}$ and $\overrightarrow{G(\mathcal{D}_2)}$
be the tournament with vertex set $\mathcal{D}_1$ and $\mathcal{D}_2$, respectively, defined by the relation $<_{\bf p_{1}}$ and $<_{\bf p_{2}}$, respectively.
An oriented graph $\overrightarrow{D}$ is \emphh{strongly connected} if there exists a directed path in 
$\overrightarrow{D}$ from $u$ to $v$ for every ordered pair of vertices $u$ and $v$ in $V(\overrightarrow{D})$. The following claim follows 
 from Claim~\ref{claim:comparison2}.

\begin{claim}
\label{claim:comparison3}
If $\overrightarrow{G(\mathcal{D}_1)}$ is strongly connected then all the diagonals in $\overrightarrow{G(\mathcal{D}_2)}$ join the same pair of valves, and the oriented graph $\overrightarrow{G(\mathcal{D}_2)}$ is strongly connected.
\end{claim}
\begin{proof}
The second claim follows directly from the first one by Claim~\ref{claim:comparison2}.
Thus, it remains to prove the first claim.
 By treating all the valves at $\mathcal{D}(\nu_{2})$ distinct from the valve of $\rho$ as a single valve, both Claim~\ref{claim:comparison} and Claim~\ref{claim:comparison2} apply to $<_{\bf p_{1}}$ and $<_{\bf p_{2}}$, and the diagonals in $\mathcal{D}_1$ and $\mathcal{D}_2$, respectively.
 Let $\mathcal{D}_2'\subseteq \mathcal{D}_2$ be non-empty and formed by all the $\nu_{2}$-diagonals in $\mathcal{D}_2$ joining the same pair of valves. 
  By the definition of $<_{\bf p_2}$,   the edges in $\overrightarrow{G(\mathcal{D}_2)}$ between 
  $\mathcal{D}_2\setminus \mathcal{D}_2'$ and $\mathcal{D}_2'$ form a directed cut, i.e., either all are directed from $\mathcal{D}_2\setminus \mathcal{D}_2'$ to  $\mathcal{D}_2'$, or all are directed from  $\mathcal{D}_2'$ to  $\mathcal{D}_2\setminus \mathcal{D}_2'$.
Now, Claim~\ref{claim:comparison2} implies $\mathcal{D}_2'=\mathcal{D}_2$.
\end{proof}

\emphh{The part of $G$ inside $\mathcal{D}(\nu)$} is the union of all $\nu$-diagonals. The part of $G$ inside $\mathcal{D}(\nu)$ is \emph{embedded}
if $\psi_0$ does not contain any edge crossing in $\mathcal{D}(\nu)$.

\paragraph*{Acyclic case.}
In this case, we assume that for every $\nu\in V(H)$ and ${\bf p}\in\partial\mathcal{D}(\nu)$ not contained in any valve, the relation $<_{\bf p}$  induces an acyclic tournament on every set of pairwise vertex-disjoint $\nu$-diagonals joining the same pair of valves. 

We show that we can embed the part of $G$ inside every disc $\mathcal{D}(\nu)$ while respecting the relation $<_{\bf p}$ defined according to $\psi_0$.
In other words, in every cluster we embed connected components (now just vertices) induced by $V_{\nu}$ together with parts of their incident pipe edges ending on valves so that the relations 
$Q_1<_{\bf p} Q_2$ are preserved for every pair of $\nu$-diagonals $Q_1$ and $Q_2$ joining the same pair of valves.
Then by reconnecting the parts $G$ inside $\mathcal{D}(\nu)$'s  we obtain a required embedding of $G$ which 
will conclude the proof.

By an easy application of the unified Hanani--Tutte theorem we obtain an embedding of the part of  $G$ inside $\mathcal{D}(\nu)$. We apply the theorem to an independently even drawing of an auxiliary graph $G_{aux}(\nu)$ in $\mathcal{D}(\nu)$, where the drawing is obtained as the union of the part of $G$ inside $\mathcal{D}(\nu)$ and $\partial\mathcal{D}(\nu)$. By subdividing edges in $G_{aux}(\nu)$ we achieve that the vertices drawn in $\partial\mathcal{D}(\nu)$ are even and therefore we indeed obtain an embedding of the part of $G$ inside $\mathcal{D}(\nu)$ as required.  Suppose that in the embedding of the part of $G$ inside $\mathcal{D}(\nu)$  we have $Q_1>_{\bf p} Q_2$ while in the drawing $\psi_0$ we have $Q_1<_{\bf p} Q_2$.
For the sake of contradiction we consider the embedding with the smallest number of such pairs, and consider such pair $Q_1$ and $Q_2$ whose end points are closest to each other along the valve that contains them.

First, we assume that both $Q_1$ and $Q_2$ pass through a connected component (a single vertex) of $G[V_{\nu}]$ of pipe degree $2$.
The endpoints of $Q_1$ and $Q_2$ are consecutive along valves, since $<_{\bf p}$ is acyclic. Thus, we just exchange them thereby contradicting the choice of the embedding. 
Second, we show that if $Q_1$ passes through a connected component $C_1$ of $G[V_{\nu}]$ of pipe degree at least $3$ and $Q_2$ passes through a component $C_2$ of pipe degree $2$, then the relation $Q_1>_{\bf p} Q_2$ in the drawing of $\psi_0$ leads to  contradiction as well. Let $\rho$ be an edge of $H$ such that there exists an edge incident to $C_1$ mapped to $\rho$ by $\varphi$ and there does not exist such an edge incident to $C_2$, see Figure~\ref{fig:order} (middle) for an illustration.
Let $\mathcal{B}$ be the complement of the union of the valves containing the endpoints of $Q_1$ or $Q_2$ in the boundary of $\mathcal{D}(\nu)$. 
Suppose that the valve of $\rho$ and ${\bf p}$ are contained in the same connected component of $\mathcal{B}$. It must be that $Q_1<_{\bf p}Q_2$ in every $\mathbb{Z}_2$-approximation of $(G,H,\varphi)$.
If the valve of $\rho$ and ${\bf p}$ are contained in the different connected components of $\mathcal{B}$, it must be that $Q_1>_{\bf p}Q_2$ in every $\mathbb{Z}_2$-approximation of $(G,H,\varphi)$,
in particular also in an approximation.
Finally, we assume that $Q_1$ and $Q_2$ pass through a connected component $C_1$ and $C_2$, respectively, of $G[V_{\nu}]$ of pipe degree at least $3$.
Similarly as above, let $\rho_1$ and $\rho_2$ be edges of $H$ such that there exists an edge incident to $C_1$ and $C_2$, respectively, mapped to $\rho_1$ and $\rho_2$ by $\varphi$, and neither $Q_1$ nor $Q_2$ ends on its valve.
 By Claim~\ref{claim:YY}, we have $\rho_1\not=\rho_2$. By Claim~\ref{claim:cross},
 the valve of $\rho_1$ and $\rho_2$ are not contained in the same connected component $\mathcal{B}$.
 Thus, by the same argument as in the previous case it must be that either in every $\mathbb{Z}_2$-approximation of $(G,H,\varphi)$ we have $Q_1<_{\bf p}Q_2$ or in every $\mathbb{Z}_2$-approximation of $(G,H,\varphi)$ we have $Q_1>_{\bf p}Q_2$.

In order to finish the proof in the acyclic case, we would like to reconnect neighboring pairs of diagonals by curves inside the pipes without creating a crossing. Let $Q_{1,i}$ and $Q_{2,i}$, for $i=1,2$, be a pair of a neighboring $\nu_{1}$-diagonal and $\nu_{2}$-diagonal sharing a pipe edge $e_i$, for $i=1,2$, such that $\varphi(e_1)=\varphi(e_2)=\rho$. We would like the endpoints of $Q_{1,1}$ and $Q_{1,2}$ to be ordered along the valve of $\rho$ consistently with the endpoints of $Q_{2,1}$ and $Q_{2,2}$ along the other valve of $\rho$. We are done if Claim~\ref{claim:comparison2} applies to $Q_{i,j}$'s. However, this does not have to be the case if, let's say $Q_{1,1}$ and $Q_{2,1}$, does not join the same pair of valves.
Nevertheless, by treating all the valves distinct from the valve of $\rho$ at $\mathcal{D}(\nu_{1})$ as a single valve, we see that both Claim~\ref{claim:comparison} and Claim~\ref{claim:comparison2}, in fact, apply to $<_{\bf p_{1}}$ and $<_{\bf p_{2}}$. 

\paragraph*{Cyclic case.}
In this case, we assume that for a vertex $\nu\in V(H)$ and ${\bf p}\in\partial\mathcal{D}(\nu)$ not contained in any valve, the relation $<_{\bf p}$  induces an acyclic tournament on a set of pairwise vertex-disjoint $\nu$-diagonals joining the same pair of valves. 

 We consider at least three $\nu$-diagonals $Q_1,\ldots ,Q_l$ inducing a strongly connected component in the tournament defined by $<_{\bf p}$. 
Let ${\bf p_k}$ and ${\bf q_k}$ be endpoints of $Q_k$, for $k=1,\ldots, l$. We assume that ${\bf p_k}$, for $k=1,\ldots,l$, are contained in the same valve, and therefore the same holds for ${\bf q_k}$.
By Claim~\ref{claim:semi-trivial2}, we assume that no connected component in $G$ is a path.
Hence, by Claim~\ref{claim:comparison3} every $Q_k$ is contained in a (drawing of a) connected component of $G$ which is a cycle. Indeed, a vertex of degree at least 3 in a connected component of $G$, whose vertex supports $Q_k$, would inevitably lead to a pair of independent edges crossing oddly in $\psi_0$, since we assume that $\varphi$ is locally injective.
Thus, by a simple inductive argument using Claim~\ref{claim:comparison3} and the fact that no two distinct strongly connected components in an oriented graph share a vertex we obtain the following property of $Q_1,\ldots ,Q_l$.

Every endpoint ${\bf p_k}$ is joined by a  curve in the closure of $\psi_0(G)\setminus \bigcup_{l'=1}^{l}Q_{l'}$ 
with an endpoint ${\bf q_{k'}}$. This defines a permutation $\pi$ of the set $\{Q_1,\ldots ,Q_l\}$, where $\pi(Q_k)=Q_{k'}$. On the one hand, each orbit 
in the permutation $\pi$ must obviously consist of $\nu$-diagonals supported by vertices in the same connected component of $G$, which is a cycle as we discussed in the previous paragraph.
On the other hand,  every pair of diagonals belonging to different orbits is supported by vertices in different cycles in $G$. Hence, the orbits of $\pi$ are in a one-to-one correspondence with a subset of connected components in $G$ all of which are cycles.
Let $C_1\ldots C_{o}$, $o\le l$, denote these cycles. By a simple inductive argument which uses Claim~\ref{claim:comparison3}, we have that every $\varphi(C_{k})={W_k,\ldots, W_k}$  with $W_k$ being repeated $o_{k}$-times, where $W_k$ is a closed walk of $H$ and $o_k$ is the size of the orbit corresponding to $C_k$. By the hypothesis of Theorem~\ref{thm:main} we assume  that
 $(C_k,H|_{\varphi(C_k)},\varphi|_{C_k})$, for $k=1,\ldots, o$, is a positive instance.

By the previous assumption, if the number of negative signs on the edges in $W_k$ (counted with multiplicities) is even then $o_k=1$. Indeed, a closed neighborhood of an approximation $\psi_{C_k}$ (which is an embedding) of $(C_k,H|_{\varphi(C_k)},\varphi|_{C_k})$  is the annulus, in which (the image of) $\psi_{C_k}$ is a non-self intersecting closed piecewise linear curve.
Analogously, we show that if the number of negative signs on the edges 
in $W_k$ (counted with multiplicities) is odd then
$o_k\le 2$, and $o_k=1$ for at most a single value of $k$, i.e., if $o_{k_1}=o_{k_2}=1$ then $k_1=k_2$.

Suppose that the previous claim holds for every $k=1,\ldots ,o$.
Since $l\ge 3$ and $o_k\le 2$ for $k=1,\ldots, o$, we have that $o\ge 2$.
We assume that $o_2\le o_1$.
 We remove the cycle $C_1$ from $G$ and apply induction.
Let $\psi$ be an approximation of $(G\setminus C_1,H,\varphi|_{G\setminus C_1})$ that we obtain by the induction hypothesis.
We construct the desired approximation of $(G,H,\varphi)$ by extending $\psi$ to $G$ as follows. We embed $C_1$ alongside $C_2$  while satisfying (\ref{it:1st}) and~(\ref{it:2nd}) for $(G,H,\varphi)$, which is possible since $1\le o_2\le o_1\le 2$.

It remains to show the claim.
 For the sake of contradiction we assume that  $o_{k_1}=o_{k_2}=1$, for $k_1\not=k_2$. The curves $\psi_0(C_{k_1})$ and $\psi_0(C_{k_2})$ are one-sided and homotopic, and therefore they must cross an odd number of times in  $\psi_0(G)$ (contradiction with the fact that $\psi_0$ is an independently even drawing).
Finally, for the sake of contradiction suppose that for some $k$, we have $o_k\ge 3$ and that there exists an approximation $\psi_{C_k}$ of  $(C_k,H|_{\varphi(C_k)},\varphi|_{C_k})$ (which is an embedding).
If $o_k$ is odd we replace (the image of) $\psi_{C_k}$ by the boundary of its small closed neighborhood, which is connected.
Thus, we can and shall assume that $o_k$ is even and still bigger than $2$.
A closed neighborhood of $\psi_{C_k}$ is the M\"obius band.  By lifting $\psi(C_k)$ to the annulus via the double cover of the  M\"obius band, we obtain a piecewise linear closed non-self intersecting curve winding $o_k/2>1$ times  around the  center of the annulus (contradiction).
\end{proof}


\section{Tractability}
\label{sec:alg}

\jk{better title suggestion: Algorithm?}

We show that testing whether $(G,H,\varphi)$ is approximable by an embedding can be carried out in polynomial time thereby establishing Theorem~\ref{thm:main2}.

By~\cite[Section 2]{FKMP15}, we can test $\mathbb{Z}_2$-approximability in polynomial time.\jk{not really I am afraid. The algorithm will have to be modified I think. There should be a separate section for this.} Therein an algorithm for c-planarity testing is considered based on solving a system of linear equations over $\mathbb{Z}_2$.
This algorithm can be easily adapted to our setting, if we start with an arbitrary initial drawing of polynomial complexity contained in $\mathcal{H}$ satisfying~(A) and~(B), and forbid edge-cluster switches; see~\cite{FKMP15} for details.
Thus, by our main result, Theorem~\ref{thm:main}, it remains to test the existence of a connected component $C, \ C\subseteq G^{(i)}$ such that $C$ is a cycle and $(C,H^{(i)},\varphi^{(i)})$, for $i=2|E(G)|$, is not approximable. 
In order to rule out the existence of $C$ we do not have to construct $\psi_0^{(i)}$,
although such a construction would lead to an efficient algorithm constructing an actual embedding\footnote{By constructing $\psi_0^{(i)}$ we mean specifying the rotation system and the parity of the number 
crossings between every pair of edges of $G$. By observing that every step of our proof is reflected by a change in such a representation of the $\mathbb{Z}_2$-approximation, it is quite straightforward to turn our algorithm into an efficient one that actually constructs an approximation. Nevertheless, to keep the presentation simple we present only a decision version of the algorithm.}.
Instead it is enough to work with a \emphh{simplified instance}
$(\overline{G},H,\overline{\varphi})$ obtained from
$({G},H,{\varphi})$ by discarding every connected component of pipe degree $0$ induced by a cluster $V_{\nu}$, contracting every connected component $C$ induced by a cluster $V_{\nu}$ to a single vertex $v_C$, deleting created loops and multiple edges, and putting $\overline{\varphi}(v_C)=\nu$. Hence, every cluster of $G$ induces an independent set in $\overline{G}$, and edges
of $\overline{G}$ capture adjacency between connected components of $G$ induced by clusters.
Therefore $\overline{G}$ is \emphh{the adjacency graph of connected components induced by clusters} of $({G},H,{\varphi})$.

Let $I=(G,H,\varphi)$, $\overline{I}=(\overline{G},H,\overline{\varphi})$.
To simplify the notation, by $\overline{I}^{(i)}$ we denote the \emphh{simplified derivative} $\underbrace{\overline{\left(\overline{{\left(\overline{\overline{I}'}\right)'}}\right)'}\ldots }_{i-times}$.
The next claim tells us that performing the simplification before each derivative does not change the outcome of the simplification of the derivative.

\begin{claim}
\label{claim:simplified}
For every $i\in \mathbb{N}$, $\overline{\overline{I}^{(i)}}=\overline{I^{(i)}}$.
\end{claim}

\begin{proof}
We prove the claim by induction on $i$.
In the base step we have $i=1$, and we need to show that the simplified instance of the derivative of $(G,H,\varphi)$ equals to the
simplified instance of the derivative of the simplified instance $(\overline{G},H,\overline{\varphi})$.
We bring $(\overline{G},H,\overline{\varphi})$ into the normal form $(\overline{G}_N,H,\overline{\varphi}_N)$.
Roughly speaking, we show that the simplified instance of the derivative of $(G,H,\varphi)$ is determined by the simplified instance $(\overline{G},H,\overline{\varphi})$.

By bringing an instance into the normal form we do not change the adjacency graph of connected components induced by clusters.
Hence, without loss of generality we assume that $(G,H,\varphi)$ is already in the normal form as well. 

Let $\mathcal{V}$ (resp. $\overline{\mathcal{V}}$) be the the union of the set of central vertices
$V_s$ (resp. $\overline{V}_s$) in $V(G)$ (resp. $V(\overline{G}_N)$), and the set of connected components of $G[V(G)\setminus V_s]$ (resp. 
 $\overline{G}_N[V(\overline{G}_N) \setminus \overline{V}_s]$).
 Let $\mathcal{G}=(\mathcal{V},\mathcal{E})$ (resp. $\overline{\mathcal{G}}=(\overline{\mathcal{V}},\overline{\mathcal{E}})$) be the bipartite graph in which $v_s\in V_s$ (resp. $v_s\in\overline{V}_s$) is joined by an edge
 with $C\in \mathcal{V}\setminus V_s$ (resp. $C\in \overline{\mathcal{V}}\setminus 
 \overline{V}_s$) if $v_s$ is adjacent to $C$ in $G$ (resp. $\overline{G}_N$).

Note that $\mathcal{G}$ and $\overline{\mathcal{G}}$ are isomorphic to the adjacency graph of connected components induced by clusters in the derivative of $(G,H,\varphi)$ and $(\overline{G},H,\overline{\varphi})$, respectively.
The base step follows by the definition of the derivative, since there exists a graph isomorphism $h$ between $\mathcal{G}$ and $\overline{\mathcal{G}}$, such that $\overline{\varphi}_N(h(v_s))=\varphi(v_s)\in V(H)$, for all $v_s\in V_s$, and $\overline{\varphi}_N(h(C))=\varphi(C)\in E(H)$, for all $C\in \mathcal{V}\setminus V_s$.

In the base step (B.S.) we proved $\overline{\left(\overline{{{{I}}}}\right)'}=\overline{I'}$.
The inductive step follows easily by the induction hypothesis (I.H.), since $\overline{\overline{I}^{(i)}}=\overline{\left(\overline{{{\overline{I}}^{(i-1)}}}\right)'}\stackrel{I.H.}{=} \overline{\left(\overline{{{{I}}^{(i-1)}}}\right)'}\stackrel{B.S.}{=}\overline{I^{(i)}}$. 
\end{proof}

By Theorem~\ref{thm:main}, $\varphi^{(i)}$, for some $i\le 2|E(G)|$, is locally injective after suppressing all degree-$2$ vertices incident to an edge induced by a cluster.
In particular, this means that $({G},H,{\varphi})^{(i)}$ is (after suppression) already a simplified instance.

Thus, by Claim~\ref{claim:simplified} in order to produce $({G},H,{\varphi})^{(i)}$ it is enough
to apply $i$ times the operation of the simplified derivative successively to the simplified instance of $(G,H,\varphi)$.
Note that it is enough to consider connected components $G_0$ of $G$, in which all connected components induced by clusters have pipe degree at most $2$. This is because by ``integrating'' such an instance (as in Section~\ref{sec:integration}) we never introduce a connected component of pipe degree more than $2$. Hence,
a cycle that is a connected component in $G^{(i)}$ and that is not approximable but $\mathbb{Z}_2$-approximable must come from such a connected component of $G$.

Here, we are cheating a little bit as we disregard other connected components of $G$ when derivating, and these connected components could definitely change the outcome of the derivative.
However, we do not lose any relevant information by working only with a subgraph $G_0$ of $G$ by the following claim.
We assume that the derived instances are in the normal form unless stated otherwise.

\begin{claim}
\label{claim:hereditary}
If $C\subseteq G^{(i)}$ is a connected component of $G^{(i)}$ then
 there exists an instance $(C_0,H|_{\varphi(C_0)},\varphi|_{C_0})$, whose $i$-th derivative is $(C,H^{(i)}|_{\varphi^{(i)}(C)},\varphi^{(i)}|_C)$, where 
$(G^{(i)},H^{(i)},\varphi^{(i)})=(G,H,\varphi)^{(i)}$ and $C_0\subseteq G$ is 
a connected component of $G$.
Moreover, if $C$ is a cycle then $C_0$ can be chosen to be a cycle that is not necessarily a connected component of $G$.
 \end{claim}
\begin{proof}
In order to simplify the notation we omit the restrictions $X|_Y$.
 The claim is proved similarly as Claim~\ref{claim:simplified} by induction on $i$. 
 For $i=1$, for the first part of the claim we have $(C_0,H,\varphi)'=(C,H',\varphi')$, where $C_0=C$, by the definition, before bringing the derivatives into the normal form, and hence, the same holds after
 normalizing. For the ``moreover'' part we observe that if $C$ is a cycle, by reversing
 the process of bringing the instance into the normal form we obtain a closed walk in $G$ containing
 the desired $C_0$. Here, we just consider the effect of the following operations on $C$: a contraction of an edge, suppression of a degree-$2$ vertex, edge subdivision, vertex split, and generalized \yDelta operation. In case of the generalized \yDelta operation performed on $v\in V(C)$, we replace the pair of edges incident to $v$ on $C$ by a subpath contained in the cycle replacing $v$.
 
 For $i>1$, by the induction hypothesis we can choose $C_0$ such that $(C_0,H,\varphi)^{(i-1)}=(C,H^{(i-1)},\varphi^{(i-1)})$, where $C\subseteq G^{(i-1)}$. Since the instances are already in the normal form $C_0$ we have $(C_0,H,\varphi)^{(i)}=((C_0,H,\varphi)^{(i-1)})'\stackrel{I.H.}{=}(C,H^{(i-1)},\varphi^{(i-1)})'=(C,H^{(i)},\varphi^{(i)})$, before bringing the instance in the normal form, and hence, the same holds after
 normalizing. The ``moreover'' part follows by the same argument as in the base case.
\end{proof}

\jk{nejak doufam, ze tyhle claimy zmizi nebo se aspon dostatecne zredukuji spolu s normalni formou...}

\paragraph{Algorithm.}
In order to simplify the notation we omit the restrictions $X|_Y$.
First, we test $\mathbb{Z}_2$-approximability of $(G,H,\varphi)$ by solving a system of linear equations of polynomial size~\cite[Section 2]{FKMP15}. 
If $(G,H,\varphi)$ is not $\mathbb{Z}_2$-approximable then 
we know that it is also not approximable.
Second, our algorithm constructs the simplified instance of $((G_0)^{(i)},H^{(i)},\varphi^{(i)})$, for all $i\le 2|E(G)|$. We do not increase the number of pipe edges by derivating, which follows by the second part of Claim~\ref{claim:semi-trivial}. Since we always contract connected components induced by clusters and discard loops and multiple edges, there are no other edges in the resulting instance besides pipe edges.
Thus, the simplified instance of $((G_0)^{(i)},H^{(i)},\varphi^{(i)})$, for every $i\le 2|E(G)|$, can be constructed in polynomial time.
Finally, for $i=2|E(G)|$, we test for all connected components $C\subseteq (G_0)^{(i)}$ that are cycles whether $(C,H^{(i)},\varphi^{(i)})$ is approximable by an embedding, which can be done efficiently, for example by the algorithm in~\cite{CDPP09}.
By Theorem~\ref{thm:main} all such tests are positive if and only if the original instance is positive.


\section{C-planarity}
\label{sec:c-planarity}

We show that Theorem~\ref{thm:main2} implies that c-planarity
is tractable for flat clustered graphs with at most three clusters.

\jk{tato definice by snad mela byt nekde brzo v uvodu... hlavni motivace...}
A \emph{flat clustered graph}, shortly flat \emph{c-graph}, is a pair $(G,T)$ where $G=(V,E)$ is a graph and $T=\{V_0, \ldots, V_{c-1}\}$ is a partition of the
vertex set into \emph{clusters}.\jk{$c$ neni dobry jako promenna, kdyz je to soucast nazvu c-planar...}\jk{tuto definici snad mame lip v revisited nebo i v uvodu... tak zkopirovat+prizpusobit na flat}
A flat c-graph $(G,T)$ is \emph{clustered planar} (or briefly \emph{c-planar}) if $G$ has an
 embedding in the plane or on the sphere $S^2$ such that (i)
for every $V_i\in T$ there is a topological disc $D(V_i)$, where $\mathrm{interior}(D(V_i))\cap \mathrm{interior} (D(V_j))=\emptyset$, if $i\not=j$,
 containing all the vertices of $V_i$ in its interior, and (ii)
 every edge of $G$ intersects the boundary of $D(V_i)$ at most once and in a proper crossing, for every $D(V_i)$.
 A \emph{clustered drawing} or \emph{clustered embedding} of a flat c-graph $(G,T)$ is a drawing or embedding, respectively,
 of $G$ satisfying (i) and (ii).
 
\jk{aha, tohle neni hanani--tutte pro 3 clustery... ale ten taky umime - s tim, ze je tam ta obstrukce v podobe namotanych cyklu. Dulezite je, ze z clustered drawing jde snadno udelat clustered drawing, kde vsechna krizeni jsou uvnitr clusteru, pouhym zvetsenim disku reprezentujicich clustery.} 
 
\begin{corollary}\label{thm:c-planarity}
Testing c-planarity of a flat c-graph $(G,T)$ can be carried out in polynomial time if $|T|\le 3$.
\end{corollary}
 
\begin{proof} 
Given a flat c-graph $(G,T)$ we construct an instance $(G,H,\varphi)$ such that 
$(G,T)$ is c-planar if and only if $(G,H,\varphi)$ is approximable by an embedding.

Formally, let $H=(T,E(H))$ be a graph such that $V_iV_j\in E(H)$ if there exists an edge
 in $G$ with end vertices in $V_i$ and $V_j$. 
Since $|T|\le 3$, the graph $H$ must be a subgraph of the triangle, and hence, the isotopy class\jk{co to tu presne znamena? trojuhelnik jde nakreslit minimalne dvema neizotopickymi zpusoby...} of an embedding $H$ is uniquely determined. Finally, for $v\in V_i$ we put $\varphi(v):=V_i$.
Without loss of generality we assume that $H$ is connected.

 Obviously, if $(G,H,\varphi)$ is approximable by an embedding then $(G,T)$ is c-planar. On the other hand, if $(G,T)$ is c-planar we consider a clustered planar embedding $\psi$ of $(G,T)$ on the sphere $S^2$.
We assume that there exists at least a single
 edge of $G$ with end vertices in $V_0$ and $V_1$.
Let $c_0,\ldots, c_k$ be the intersections of the edges of $G$ with $D(V_0)$ listed in the order of appearance along the boundary of $D(V_0)$.
Note that all the intersections $c_j$ of the edges between $V_0$ and $V_1$ form a subsequence of consecutive elements in $c_0,\ldots, c_k$.\jk{tohle je centralni argument, tak by mel byt podrobnejsi nez "note that". Zde se musi vyuzit toho, ze ty clustery jsou jen 3. Pro 4 to uz obecne neplati.}
Let $c_i,\ldots, c_j$, $i<j$, be the crossings with edges between $V_0$ and $V_1$.
Let $e_i$ and $e_j$ be the edge crossing the boundary of $D(V_0)$ in $c_i$ and $c_j$, respectively.

Let $R_{01}$ be the union of the 
discs $D(V_0)$ and $D(V_1)$, and the region
bounded by the pieces of $e_i$ and $e_j$
internally disjoint from $D(V_0)$ and $D(V_1)$ and the parts of the boundary of $D(V_0)$ and $D(V_1)$ containing all the other intersection points of the edges between $V_0$ and $V_1$.

We define $R_{02}$ and $R_{12}$ analogously if there exists an edge between $V_0$ and $V_2$, and $V_1$ and $V_2$, respectively.
If this is not the case $R_{02}$ or $R_{12}$ is the empty set.
Now, a suitable homeomorphism of $R_{01} \cup R_{02} \cup R_{12}$ yields an approximation of $(G,H,\varphi)$.
 \end{proof}

\section{Torus and thickenability}
\label{sec:toroidal} 

A far reaching  generalization of our problem is \emph{thickenability of 2-dimensional simplicial complexes} (shortly \emphh{thickenability}), i.e., determining for a given   2-dimensional simplicial complex, if there exists a 3-dimensional manifold into which the complex embeds.
That this seemingly unrelated problem is indeed a generalization of the problem that we study follows from~\cite[Lemma]{Skop94_thick}\footnote{The Lemma is stated only for connected graphs and only in the case when the target surface is a sphere. However, it is not hard to prove an analogous statement for disconnected graphs and for orientable surfaces of arbitrary genus.}. In the proof of the lemma our problem is reduced to the thickenability by  considering a triangulation of the mapping cylinder of $\varphi: G \rightarrow M$. 
 We remark that a polynomial-time algorithm for deciding whether a 2-dimensional simplicial complex embeds in $\mathbb{R}^3$ would already yield a polynomial-time algorithm for our problem in the case of orientable surfaces. However, NP-hardness~\cite{dMRST17+_embedR3} of this problem was announced recently, so the existence of such an algorithm is highly unlikely.
 Hence, studying the thickenability, whose computational complexity status is still open, appers to be a next natural  step in our investigation.
 
Neuwirth's algorithm~\cite{Skop94_thick} for thickenability can be seen as the following extension of the problem of deciding if $(G,H,\varphi)$ is approximable by an embedding\footnote{In fact, for Neuwirth's algorithm it is enough to consider $G$ to be a finite union of pairwise disjoint 3-cycles.}.
Roughly speaking, the thickening $\mathcal{H}$ will be a 2-dimensional surface, but this time without boundary, partitioned into regions representing clusters and pipes. Then instead of disc clusters $\mathcal{D}(\nu)$, for $\nu\in V(H)$,
we take punctured 2-spheres $\mathcal{S}(\nu)$ with the number of holes equal to $\degr(\nu)$, and each hole being \emphh{designated} for a unique edge incident to $\nu$. The pipe $P(\nu\mu)$
of $\nu\mu\in E(H)$ is replaced by a cylinder that has one boundary component glued to the boundary of the hole on $\mathcal{S}(\nu)$ designated for $\nu\mu$, and the other boundary component glued to the boundary of a hole on $\mathcal{S}(\mu)$
designated for $\nu\mu$. So, the boundaries of cylinders are analogs of valves.
The problem that we are interested in is, of course, to decide, if an embedding of $G$ in $\mathcal{H}$  satisfying analogs of~(\ref{it:1st}) and~(\ref{it:2nd}), in which we replace $\mathcal{D}(.)$ with $\mathcal{S}(.)$, exists.
 
In the following we show that our technique extends easily to give the tractability of the described more general problem if the maximum degree of $H$ is two, i.e., in the \emphh{toroidal case}, see Figure~\ref{fig:exTorus} for an illustration. Note that in this case every cluster is either a cylinder or a sphere with a single hole. To this end it is enough to slightly generalize three steps in the proof of Theorem~\ref{thm:main} and Theorem~\ref{thm:main2}.
We remark that this  already  generalizes a recent  work on leveled planarity and beyond by Angelini et al.~\cite{ADF+16_beyond}, radial planarity~\cite{FPS16_radialII,FPS17_radial} and also the results on c-planarity without two disjoint clusters by Gutwenger et al.~\cite{GJL+02_advances}. Indeed, the considerations of Section~\ref{sec:alg} carry over verbatim   except that we need to introduce variables for the edge-cluster switches in the linear system as described in~\cite[Section 2]{FKMP15} playing the role of (Dehn) twists defined below, see also~\cite[Section 5]{FPS16_radialII}. We are not aware of any result implying a polynomial time algorithm for this problem.

\begin{figure}
\centering
\includegraphics[scale=0.3]{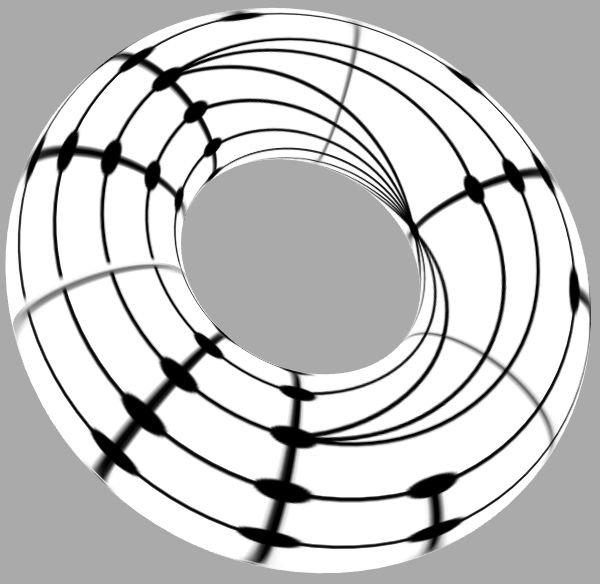}
\caption{The handle body $\mathcal{H}$ in the toroidal case with an embedding satisfying analogs of~(\ref{it:1st}) and~(\ref{it:2nd}). The clusters $\mathcal{S}(\nu)$'s are delimited by three meridians representing (very thin) pipes $P(\nu\mu)$. }
\label{fig:exTorus}
\end{figure}

In the sequel we use the notion of \emphh{winding number ($\mod$ 2)} of a closed curve $C$ drawn on a cylinder understood as  $S^1\times [0,1]$ defined to be 0, if $C$ crosses a line segment $s\times [0,1]$\footnote{Here, we tacitly assume a generic position of $C$ w.r.t. the line segment, i.e., that the line segment an $C$ are intersecting in a finitely many crossing points.}, $s\in S^1$, an even number of times, and to be 1 otherwise.

\subsection{Normalization}

We note that the normalization in Section~\ref{sec:normalform2} goes through except that the application of Claim~\ref{claim:cycle_red} is not possible, in general, since the boundary of a cluster can have two connected components.
In particular, Claim~\ref{claim:cycle_red} is applicable only to cycles induced by a cluster, which is also a cylinder, in $\psi_0$ with winding number~0. A cycle $C$ induced by a cluster $\nu$ with winding number~1 is resolved by modifying the instance and $\psi_0$ as follows.

Analogously to the proof of Claim~\ref{claim:cycle_red} we split $G[V_\nu]$ into 
two parts. Let $\mu_1$ and $\mu_2$ denote the neighbors of $\mu$ each of which might not exists.
Just for the purpose of the next definition we perform the following operation. If $\degr(\nu)=2$, let us turn $\mathcal{S}(\nu)$ into a disc by attaching a disc to the boundary of the hole, to which the cylinder of $\nu\mu_2$ was attached. Thereby a closed curve in a cluster has defined inside and outside in the sense of Section~\ref{sec:preliminaries}.
As in Section~\ref{sec:preliminaries}, let $V_{out}(C)\subset V_{\nu}, D_{out}(C)$, and  $V_{in}(C)\subset V_{\nu}, D_{in}(C)$, denote the set of vertices outside of $C$, the set of outer diagonals of $C$, and vertices inside of $C$, inner diagonals of $C$, respectively.
We construct a clone $(\hat{G},\hat{H},\hat{\varphi},\hat{\psi}_0)$ of $(G,H,\varphi,\psi_0)$ with fewer cycles induced by clusters with winding number 1.

First, we remove $\nu$ and its incident edges from $H$, and introduce two vertices $\nu_1$ and $\nu_2$ to $H$ and edges $\mu_1\nu_1$ and $\mu_2\nu_2$.
Let $\hat{H}$ denote the resulting graph.

Second, we modify  $G$ in order to obtain $\hat{G}$. Roughly, we introduce a second copy of $C$ in $G$, where the second copy takes over the edges between $C$ and  $V_{out}$ and the outer diagonals of $C$. Formally, this amounts to adding to $G$ a new vertex $\overline{v}$, for every $v\in V(C)$. Then adding to $G$ a new edge $\overline{v}\overline{u}$, for every edge $vu\in D_{out}(G)$,
and  $\overline{v}u$, for every edge $vu\in E(G)$, where $v\in V(C)$ and $u\in V_{out}$, 
and finally deleting the edges between $C$ and $V_{out}$ and the edges in $D_{out}$.

Third, we define $\hat{\varphi}(\overline{v}):=\nu_1, \hat{\varphi}(v):=\nu_2$, for $v\in V(C)$, and $\hat{\varphi}(v):=\nu_1$, for $v\in V_{out}$, $\hat{\varphi}(v):=\nu_2$, for $v\in V_{in}$, and $\hat{\varphi}(v):=\varphi(v)$ for the remaining vertices of $\hat{G}$.

Finally, $\hat{\psi}_0$ is naturally inherited from ${\psi}_0$. Since both $\nu_1$ and $\nu_2$ are discs, none of the copies of $C$ has winding number 1 in the resulting instance and therefore we indeed decreased the number of cycles with winding number 1 induced by clusters.
It is straightforward to prove the following.

\begin{claim}
\label{claim:toroidal_red}
$(\hat{G},\hat{H},\hat{\varphi},\hat{\psi}_0)$ is a clone of $(G,H,\varphi,\psi_0)$.
\end{claim}

\subsection{Derivative}

Since we do not have connected component(s) induced by  clusters of  pipe degree more than $2$, in order to extend the derivative of a $\mathbb{Z}_2$-approximation in the toroidal case, we just need to adapt  
Claim~\ref{claim:ie}. We note that reconnecting an edge inside a cylindrical pipe, as in the construction of $\hat{\psi}_0'$ in Section~\ref{sec:derivative_z2approx}, can introduce a pair of non-adjacent edges crossing an odd number of times. 
Here, we still assume that the severed endpoints on the boundary of the pipe are fixed. However, every triple of edges passing through a single pipe will be fine in the sense that the total number of crossings between pairs among them will be even. This follows from the following claim.

Let $C_1,C_2,C_3$ denote the three curves on a cylinder connecting $S^1\times 0$ with $S^1\times 1$ that are internally disjoint from its boundary and that
intersect both  $S^1\times 0$ and $S^1\times 1$ in three distinct points
${\bf p_1, p_2, p_3}$ and ${\bf q_1, q_2, q_3}$, respectively.
Let us assume that ${\bf p_1, p_2, p_3}$ appears in this order along $S^1\times 0$ clockwise, where we assume that  $S^1$ is endowed with a clockwise orientation.
For a subset $S\subseteq S^1 \times [0,1]$, we denote by $I(S)$ the \emphh{projection} of $S$ to $I=[0,1]$, that is $I(S)=\{i\in I| \  (i,s)\in S$ for some $s\in S^1 \}$.

\begin{claim}
\label{claim:parityTriples}
The parity of the total number of crossings between the pairs of curves among $C_1, C_2$ and $C_3$ is even if and only if ${\bf q_1, q_2, q_3}$ appear in this order along $S^1\times 1$  clockwise.
\end{claim}

\begin{proof}
By a continuous deformation fixing $S_1 \times \{0,1\} $ we turn $C_1,C_2$ and $C_3$ into curves, whose projection to $[0,1]$ is injective. We note that during the deformation the parity of the number of crossings between any pair of curves among $C_1,C_2$ and $C_3$ is not changed. W.l.o.g. we assume that the projections to $[0,1]$ of the crossings defined by pairs of curves  $i_1,\ldots , i_k$ are all distinct. The claim follows, since the order of the intersection points  of $C_1,C_2,C_3$ with $S^1 \times \frac{i_j+i_{j+1}}{2}$, let us denote them by ${\bf q_{1,j}, q_{2,j}, q_{3,j}}$, is clockwise for even $j$'s counterclockwise for odd $j$'s.
\end{proof}

 Let $e$ denote an edge not induced by a cluster. Let a \emphh{twist} of $e$ inside the pipe be an operation on the drawing of $G$ that consists changing the (image of the) drawing of $e$ by taking its union with a closed curve of winding number 1 contained inside the pipe intersecting the drawing of $e$. This introduces self-crossing(s) of $e$ that can be get rid of by a standard argument, see Section~\ref{sec:preliminaries}.
 
In what follows we show that Claim~\ref{claim:parityTriples} implies that 
by applying twists to $\psi_0'$ obtained from $\psi_0$ by applying the derivative of Section~\ref{sec:derivative_z2approx}, we obtain a desired $\mathbb{Z}_2$-approximation of $(G',H',\varphi')$ from $\psi_0'$. Here, the definition is extended in a straightforward way such that $\mathcal{D}(\nu)$'s are replaced by cylinders $\mathcal{S}(\nu)$'s. Recall that pipes are also cylinders and valves are circles.

By an argument analogous to the proof of Claim~\ref{claim:ie}, Claim~\ref{claim:parityTriples} implies that for every three edges $e_1,e_2,e_3\in E(G')$ such that $\varphi'(e_1)=\varphi'(e_2)=\varphi'(e_3)=(\mu_1\nu)^*(\mu_2\nu)^*=\rho\in E(H')$, each two of which must be therefore independent,
we have the following.

\begin{figure}
\centering
\includegraphics[scale=0.7]{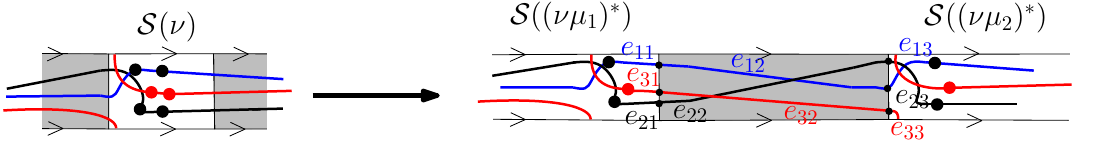}
\caption{Illustration for the proof of Claim~\ref{claim:ieRadial}. The setup here is analogous to the proof of Claim~\ref{claim:ie}.}
\label{fig:ieRadial}
\end{figure}

 \begin{claim}
 \label{claim:ieRadial}
 If $\hat{\psi}_0$ is independently then in $\hat{\psi}_0'$ 
\begin{equation}
\label{eqn:triple}
\nCross{e_1}{e_2}{{\psi}_0'}+\nCross{e_1}{e_3}{{\psi}_0'}+\nCross{e_2}{e_3}{{\psi}_0'}\equiv_2 0.
\end{equation}
 \end{claim}

\begin{proof}
  Refer to Figure~\ref{fig:ieRadial}.
  For the sake of contradiction suppose the contrary.
  Analogously to the proof of Claim~\ref{claim:ie}, we split every edge $e_i$, $i=1,2,3$ into three parts $e_{i1},e_{i2},e_{i3}$. Thus, we temporarily subdivide $e_i$'s as indicated in the figure.
We are done if we prove that $\sum_{i=1}^3\sum_{j=i}^{3}\sum_{k=1}^{3}\nCross{e_{ik}}{e_{jk}}{\psi_0'}\equiv_2 0$ or equivalently,
that  $\sum_{i=1}^3\sum_{j=i}^{3}(\nCross{e_{i1}}{e_{j1}}{\psi_0'}+
\nCross{e_{i3}}{e_{j3}}{\psi_0'})\equiv_2 \sum_{i=1}^3\sum_{j=i}^{3}\nCross{e_{i2}}{e_{j2}}{\psi_0'}$.

We have that $\sum_{i=1}^{3}\sum_{j=i}^{3}\nCross{e_{i2}}{e_{j2}}{\psi_0'}\equiv_2 1$ if and only if the end points of $e_{12},e_{22},e_{32}$ appear on one valve of the pipe
of $\rho$ in a clockwise order and on the other valve of the same pipe in a counterclockwise order by Claim~\ref{claim:parityTriples} or vice-versa. 
On the other hand, by  Claim~\ref{claim:parityTriples} applied inside the cluster $\mathcal{S}(\nu)$ this happens if and only if
 $\sum_{i=1}^3\sum_{j=i}^{3}(\nCross{e_{i1}}{e_{j1}}{\psi_0}+
\nCross{e_{i3}}{e_{j3}}{\psi_0}+\nCross{e_{i1}}{e_{j3}}{\psi_0}+
\nCross{e_{i3}}{e_{j1}}{\psi_0})\equiv_0 1$.

However, $\sum_{i=1}^3\sum_{j=i}^{3}(\nCross{e_{i1}}{e_{j3}}{\psi_0}+
\nCross{e_{i3}}{e_{j1}}{\psi_0})\equiv_20$, as $H$ has no multi-edges or loops,
and $\psi_0$ is independently even, since all the involved pairs of crossing edges cross an even number of times.
 Thus,  $\sum_{i=1}^{3}\sum_{j=i}^{3}\nCross{e_{i2}}{e_{j2}}{\psi_0'}\equiv_2 1$ if and only if 
 $\sum_{i=1}^3\sum_{j=i}^{3}(\nCross{e_{i1}}{e_{j1}}{\psi_0'}+
\nCross{e_{i3}}{e_{j3}}{\psi_0'})\equiv_2 1$, where in the last equality we used the fact that $\nCross{e_{i1}}{e_{j1}}{\psi_0}\equiv_2 \nCross{e_{i1}}{e_{j1}}{\psi_0'}$ and $\nCross{e_{i3}}{e_{j3}}{\psi_0}\equiv_2 \nCross{e_{i3}}{e_{j3}}{\psi_0'}$ for $i\not=j$.
 \end{proof} 

For $\rho\in E(H')$, we consider an auxiliary graph $G_\rho$, whose vertices are 
edges $e$ of $G'$ such that $\varphi'(e)=\rho$, and in which two vertices $e_1$ and $e_2$ are connected by an edge if and only if $\nCross{e_1}{e_2}{{\psi}_0'}\equiv_2 1$. Note that the application of a twist on $e\in V(G_\rho)$ results in a local complementation at $e$, i.e., in a new graph $G_\rho$, in which $e$ is connected with a vertex if and only if $e$ was not connected with the vertex before. By~(\ref{eqn:triple}), it follows that $G_\rho$ is a complete bipartite graph. For if not, $G_\rho$ is either a bipartite graph containing three vertices that induce a single edge (contradiction); or $G_\rho$ contains a (shortest) odd cycle, which must be a triangle (contradiction). The last claim about a triangle follows, since in $G_\rho$, every edge in an odd cycle of length at least 5 must be adjacent to a diagonal of the cycle.  Therefore by applying a twist to every vertex in one part we obtain a desired $\mathbb{Z}_2$-approximation.


\section{Acknowledgments}

We are grateful to Arkadiy Skopenkov for informing us about~\cite{Sko03_approximability}, Mikhail Skopenkov for reading
carefully preliminary version(s) and providing valuable critical comments, and anonymous referees for comments that helped us to improve the presentation of the results.



\end{document}